\documentclass[draftcls,onecolumn,a4paper]{IEEEtran}

\usepackage{amsmath}
\usepackage{amssymb}
\usepackage{amsthm}
\usepackage{amsfonts}
\usepackage{bm}
\usepackage{dsfont}

\usepackage[dvips]{graphicx}
\usepackage{epsfig}

\usepackage[svgnames]{xcolor} 
\usepackage{pstricks,pst-node,pst-plot,pstricks-add}
\usepackage[tight,footnotesize]{subfigure}
\usepackage{algpseudocode}
\usepackage{algorithm}
\usepackage{graphics}
\usepackage{mathrsfs}
\usepackage{psfrag}
\usepackage{tikz}
\usepackage{accents}
\usepackage{bbm}

\graphicspath{{figs/}}
\usepackage{tikz}
\usepackage{pgfplots}

\usepackage{parham}
\usepackage{amsmath,centernot}

\date{\today}

\begin{document}

\title{Universal Neyman-Pearson Classification with a Known Hypothesis}
\author{Parham Boroumand and Albert Guill\'en i F\`abregas

%

\thanks{P. Boroumand is with the Department of Engineering, 
University
of Cambridge, Cambridge CB2 1PZ, U.K. (e-mail: pb702@cam.ac.uk).

A.~Guill\'en i F\`abregas is with the Department of Engineering, 
University
of Cambridge, Cambridge CB2 1PZ, U.K., and also with the Department of Information and Communication Technologies, Universitat Pompeu Fabra, 08018 Barcelona, Spain (e-mail: guillen@ieee.org)
}
\thanks{
This work was supported in part by the European Research Council under
Grant 725411 and by the Spanish Ministry of Economy and Competitiveness under Grant PID2020-116683GB-C22.
}
}

\maketitle

\begin{abstract} 
	 We propose a universal classifier for binary Neyman-Pearson classification where null distribution is known while only a training sequence is available for the alternative distribution. The proposed classifier interpolates between Hoeffding's classifier and the likelihood ratio test and attains the same error probability prefactor as the likelihood ratio test, i.e., the same prefactor as if both distributions were known. In addition, like Hoeffding's universal hypothesis test, the proposed classifier is shown to attain the optimal error exponent tradeoff attained by the likelihood ratio test whenever the ratio of training to observation samples exceeds a certain value. We propose a lower bound and an upper bound to the training to observation ratio. In addition, we propose a sequential classifier that attains the optimal error exponent tradeoff. 

\end{abstract}

\section{Introduction}

The problem of deciding between multiple statistical descriptions of a given observation is one of the main problems in statistics and finds applications in a number of fields including engineering, signal processing, medicine and finance among others (see e.g. \cite{Lehmann}). Depending of whether or not the possible probability distributions of the observation are known, the problem is termed hypothesis testing or classification.
In the case where there are only two possible distributions of the observation one typically refers to these decision problems as binary. When priors on the distributions are available, the problem is cast as Bayesian and the average probability of error determines the quality of the detection. In the absence of priors, the design of tests and classifiers minimizes one pairwise error probability while keeping the other upper bounded by a constant. This setting, proposed by Neyman and Pearson \cite{Neyman} has been adopted as the de-facto test design setting since it allows for both pairwise error probabilities being bounded, unlike average risk minimization.

Upon observing a vector of $n$ observations, the hypothesis test that minimizes the above pairwise error probability tradeoff when the two possible distributions are known is given by the likelihood ratio test \cite{Neyman}. Instead, when decisions need to be made each time a new observation arrives, the sequential probability ratio test provides the best exponential decay of the error probability (or error exponent) \cite{Wald, Tanseq}. When one of the hypothesis distributions belongs to some class of distributions and only the class is known, Hoeffding proposed the generalized likelihood ratio test that attains the optimal error exponent \cite{Hoeffding}. When the distributions are not known, several classifiers such as logistic regression, support vector machines, naive Bayes, have been proposed in the literature \cite{VC, Shirani, Gaborpattern, CoverNN}. However, none of these classifiers ensures a guarantee on the type-\RNum{1} error probability  resulting in the possibility of a large type-\RNum{1} errors.  In the Neyman-Pearson setting, \cite{Tong} gives uniform bounds on type-\RNum{1} error probability for the plugin likelihood ratio test when both distributions are unknown, and only training sequences from both distributions are available. Gutman \cite{gutman1989asymptotically} proposed universal classifiers that guarantee a certain type-\RNum{1} error exponent under any probability distribution pair while achieving the lowest type-\RNum{2} error probability. Tan \cite{Tan} proved the second order optimality of Gutman's classifier.

This paper considers the setting when the null hypothesis generating distribution is known while only a training sample is available for the alternative hypothesis. This can be the case in the unbalanced training sample size when a large number of training samples is available for the null hypothesis so that with high accuracy, the distribution can be estimated. In contrast, the alternative distribution cannot be accurately estimated due to small number of training samples. The scenario when the number of test and training samples is fixed can also be viewed as composite hypothesis testing with known null hypothesis and an additional training sequence. Hoeffding's test can achieve the optimal error exponent when the null hypothesis distribution is given. However, the prefactor attained by Hoeffding's test is only optimal for distributions defined over binary alphabets.

In this paper, we propose a classifier that achieves the optimal prefactor using the additional training samples while attaining the optimal error exponent tradeoff whenever the training to observation sample ratio exceeds a certain value. The proposed classifier interpolates between the plugin likelihood ratio and Hoeffding's tests. We also study the case when the test and training samples are generated sequentially. There is no classifier similar to Hoeffding's in this scenario that can achieve the error exponent gain by taking test samples sequentially when only the null hypothesis is known. However, our proposed classifier can improve the achievable error exponent by using the additional training samples from the unknown distribution.

This paper is structured as follows. Section \ref{sec:pre} introduces notation and reviews the preliminaries about the likelihood ratio test and Hoeffding's generalized likelihood ratio test. Section \ref{sec:fixedclas}, proposes the new classifier for a fixed sample size setting and shows the finite length improvements over Hoeffding's test. Section \ref{sec:steinclas} discusses the classification problem in the Stein regime and shows the optimality of Gutman's test in this regime. Section \ref{sec:seqclas},  proposes a sequential classifier achieving the highest error exponent tradeoff in the proposed setting. Proofs of the main results can be found in the Appendices.

\section{Prelimenaries} \label{sec:pre}

Consider the following binary classification problem where an observation $\xv=(x_1,\dotsc,x_n)$ is generated  in an i.i.d. fashion from either of two possible distributions $P_0$ or $P_1$  defined on the probability simplex $\Pc(\Xc)$ with alphabet size $|\Xc|<\infty$. We assume that the distribution $P_0$ is known while only a sequence of training samples $\X=(X_1,\dotsc,X_k)\sim P_1^k$  generated  in an i.i.d. fashion  from $P_1$ is available; training and test sequences are sampled independently from each other. We also assume that both $P_0(x)>0, P_1(x)>0$ and $\frac{P_0(x)}{P_1(x)} \leq c$ for each $x\in\Xc$ for some positive $c$. Also we let $k$, the length of the training, be such that $k=\alpha n$ for some positive $\alpha$. 

The type of an $n$-length sequence $\yv$  is defined as $\hat{T}_{\yv}(a)=\frac{N(a|\yv)}{n}$, where $N(a|\yv)$ is the number of occurrences of symbol $a\in\Xc$ in sequence $\yv$. The types of the observation and training sequences $\xv, \X$ are denoted by $\Tx, \TX $ respectively.  The set of all sequences of length $n$ with type $P$, denoted by $\mathcal{T}_{P}^n$, is called the type class. The set of types formed with length $n$ sequences on the simplex $\Pc(\Xc)$ is denoted as $\Pc_n(\Xc)$. 

Let  $\phi(\X,\xv): \Xc^k \times \Xc^n \rightarrow \{0,1\}$ be a classifier that decides the distribution that generated the observation $\xv$ upon processing the training sequence $\X$. We consider deterministic classifiers $\phi$ that decide in favor of $P_0$ if $\xv\in \Ac_0(P_0,\X)$, where $ \Ac_0(P_0,\X)\subset \Xc^n$ is the decision region for the first hypothesis and is a function of $P_0$ and the training samples $\X$. We define $\Ac_1(P_0,\X)=\Xc^n   \setminus \Ac_0$ to be the decision region for the second hypothesis. If we assume no prior knowledge on either distribution, the two possible pairwise error probabilities determine the performance of the classifier. Specifically, the type-\RNum{1} and type-\RNum{2} error probabilities are defined as
\begin{align}\label{eq:e1}
\epsilon_0 (\phi)& = \sum_{\X \in  \Xc^k} P_1(\X) \sum_{\xv \in \Ac_1(P_0,\X)} P_0(\bx^n),\\
\epsilon_1 (\phi)&=  \sum_{\X \in  \Xc^k} P_1(\X) \sum_{\xv \in \Ac_0(P_0,\X)} P_1(\bx^n).  
\end{align}
In the case where both distributions are known, the training sequence is not needed and the classifier becomes a hypothesis test. In this case, the classifier is said to be optimal whenever it achieves the optimal error probability tradeoff given by
\begin{equation}\label{eq:trade}
\min_{\phi: \epsilon_0 (\phi) \leq \xi} \epsilon_1 (\phi),
\end{equation}
where $\xi \in [0,1]$. It is well known that likelihood ratio test defined as
\begin{equation}
\phi^{\rm lrt}(\xv)= \mathbbm{1} \bigg \{ \frac{P_1^n(\bx)}{P_0^n(\bx)}  \geq e^{n\gamma} \bigg\},
\end{equation}
attains the optimal tradeoff \eqref{eq:trade} for every $\gamma$  \cite{Neyman}. The likelihood ratio test can also be expressed as a function of the type of the observation $\Tx$ as \cite{Cover,Csiszar}
\begin{align}\label{eq:LRTtype}
\phi^{\rm lrt}(\Tx)= \mathbbm{1} \big\{ D(\Tx\|P_0)-D(\Tx\|P_1)    \geq  \gamma \big\}.
\end{align}
where $D(P\|Q)= \sum_{\Xc} P(x) \log \frac{P(x)}{Q(x)}$ is the relative entropy between distributions $P$ and $Q$. The optimal error exponent tradeoff $(E_0,E_1)$ is defined as
\begin{align}\label{eq:tradefix}
E_1^*(E_0) \triangleq \sup \big\{E_1\in \mathbb{R}_{+}: \exists \phi , \exists n_0 \in \mathbb{Z}_+  \  \text{s.t.} \  \forall   n>n_0 , \epsilon_0(\phi) \leq e^{-nE_0}  \quad \text{and} \quad  \epsilon_1(\phi) \leq e^{-nE_1}\big\}.
\end{align}
By using Sanov's Theorem \cite{Cover,Dembo}, the optimal error exponent tradeoff $(E_1,E_0)$, attained by the likelihood ratio is given by
\begin{align}
E_0(\phi^{\rm lrt})=\min_{Q \in \mathcal{Q}_0(\gamma)} D(Q\|P_0)\label{eq:min1},\\
E_1(\phi^{\rm lrt})=\min_{Q \in \mathcal{Q}_1(\gamma)} D(Q\|P_1)\label{eq:min2},
\end{align}
where 
\begin{align}
\mathcal{Q}_0(\gamma)&=  \big\{Q\in \mathcal{P}(\Xc): D(Q\| P_0)-D(Q\|P_1) \geq  \gamma  \big\},\\
\mathcal{Q}_1(\gamma)&= \big\{Q\in \mathcal{P}(\Xc): D(Q\| P_0)-D(Q\|P_1) \leq  \gamma    \big\}.
\end{align}
Furtheremore, the minimizing distribution in \eqref{eq:min1}, \eqref{eq:min2} is the tilted distribution
\begin{equation}\label{eq:tilted}
Q_{\lambda^*}(x)= \frac{ P_{0}^{\lambda^*}(x) P_{1}^{1-\lambda^*}(x) } {\sum_{a \in \Xc }  P_{0}^{1-\lambda^*}(a) P_{1}^{\lambda^*}(a) }, ~~~0\leq\lambda^* \leq 1,
\end{equation}
where $\lambda^*$ is the solution of $D(Q_{\lambda^*}\| P_0)-D(Q_{\lambda^*} \| P_1) = \gamma$. 

The classification problem described above with known $P_0$ and a training sequence from $P_1$, can also be viewed as the composite binary hypothesis problem where additional training sequence samples are given for the second hypotheses. In the case of a composite hypothesis testing problem where $P_0$ is given, and the other hypothesis is unrestricted to $\Pc(\Xc)$, Hoeffding proposed in \cite{Hoeffding}  a generalized likelihood-ratio test given by
\begin{align}
\phi^{\rm glrt}(\xv)=  \mathbbm{1} \big\{ D(\Tx\|P_0)   >  E_0 \big\},
\end{align}
which attains the optimal error exponent tradeoff in \eqref{eq:tradefix}. By Sanov's theorem, the error exponent of Hoeffding's test is given by
\begin{align}
 E_0(\phi^{\rm glrt})&=E_0,\\
 E_1(\phi^{\rm glrt})&=\min_{\substack{Q\in \mathcal{P}(\Xc), \\ D(Q\| P_0) \leq  E_0   }} D(Q\|P_1),\label{eq:min2H}
\end{align}
where the minimizing distributions in \eqref{eq:min2H} is given by
\begin{equation}\label{eq:tiltedH}
Q_{\mu^*}(x)= \frac{ P_{0}^{\frac{\mu^*}{1+\mu^*}}(x) P_{1}^{\frac{1}{1+\mu^*}}(x) } {\sum_{a \in \Xc }  P_{0}^{\frac{\mu^*}{1+\mu^*}}(a) P_{1}^{\frac{1}{1+\mu^*}}(a) }, ~~~\mu^*\geq0,
\end{equation}
and $\mu^*$ is the solution to $D(Q_{\mu^*}\|P_0)=E_0$. Using a large deviations refinement \cite{Iltis,Albert}, the type-\RNum{1}  error probability of the likelihood ratio test is
\begin{equation}\label{eq:LRTsaddle}
\epsilon_0(\phi^{\rm lrt}) = \frac{1}{\sqrt{n}} e^{-nE_0}\big (c+ o(1)\big),
\end{equation}
while, Hoeffding's test type-\RNum{1} error probability is given by \cite{Veeravalli,Iltis}
\begin{equation}\label{eq:Hoefasymp}
\epsilon_0(\phi^{\rm glrt})  =n^{\frac{|\Xc|-3}{2}} e^{-nE_0}\big (c'+ o(1)\big)
\end{equation}  
where $c,c'$ are constants that only depend on $P_0, P_1$ and the corresponding test thresholds. As a result, when the number of observations is large, Hoeffding's test, although attaining the optimal error exponent tradeoff, suffers in exponential prefactor when compared to the likelihood ratio's $\frac{1}{\sqrt{n}}$ for observation alphabets such that $|\Xc|>2$.

\section{Fixed Sample Sized Universal Classifier}\label{sec:fixedclas}

In this section, we propose a classifier that interpolates between the likelihood ratio and Hoeffding's tests that attains a prefactor that is independent of the alphabet size and is equal to $\frac{1}{\sqrt{n}}$. In addition, we show that if the ratio of training samples to the number of test samples $\alpha$ exceeds a certain threshold, the proposed test also achieves the optimal error exponent tradeoff.

Hoeffding's test can favor the second hypothesis for test sequences with types close to $P_0$ while far from $P_1$.
Suppose we have a training sequence type $\TX$, we can relax the Hoeffding's test from a ball centered at $P_0$ to a hyperplane tangent to the Hoffding's test ball, directed towards the type of the training sequence. As we will see, this is precisely what enables the improvement in the prefactor of the type-\RNum{1} probability of error. We propose the following classifier
\begin{align}\label{eq:clsfix}
\cls(\Tx,\TX)=\mathbbm{1} \big\{ \beta D(\Tx\|\TX') - D(\Tx\|P_0)\   \leq  \gamma(E_0,\TX') \big\},
\end{align}
where $0 \leq \beta \leq 1$ controls how much the training weights in the decision, the threshold $\gamma(E_0,Q_1)$ is given by
\begin{align}
\gamma(E_0,Q_1)&= \beta \min_{\substack{Q\in \mathcal{P}(\Xc), \\  D(Q\| P_0) \leq  E_0    }    } D(Q\|Q_1) -E_0 \label{eq:threshcls}, 
\end{align}
the perturbed training type $\TX'(a)$ is
\begin{align}\label{eq:typeplus}
\TX'(a)&=\big(1-\delta_n\big)\TX(a)+\frac{\delta_n}{|\Xc|},
\end{align}
where, $\delta_n$ can be chosen as any function of the order  $o(n^{-1})$. We add this small perturbation of the training type to avoid the cases where some of the alphabet symbols have not been observed in the training sequence. We define the decision regions of the proposed classifier by
 \begin{align}
\Az&=\{Q: Q\in \Pc(\Xc),  \cls(Q,\TX)=0 \}, \\
\Ao&=\{Q: Q\in \Pc(\Xc), \cls(Q,\TX)=1 \}. 
 \end{align}
Observe that when $\beta=0$ we recover Hoeffding's test while for $\beta=1$ the test is reminiscent of a likelihood ratio test where instead of $P_1$, we have the perturbed training type $\TX'(a)$. Intuitively, as long as we have enough training samples, the training type $\TX'(a)$ will be close to $P_1$ and we will attain the optimal error exponent tradeoff. This is indeed what will be shown next.

 Figure \ref{fig:classification} illustrates the proposed classifier for a realization of the training sequence when $\beta=1$. The proposed classifier becomes the plug-in likelihood ratio test, where the test's threshold has been adjusted by the training samples such that the resulting hyperplane of the likelihood ratio test is tangent to the relative entropy ball of radius $E_0$ centered at $P_0$. The type-\RNum{1} error exponent will be equal to $E_0$ for any realization of $\TX$. However, the type-\RNum{2} error exponent for a given training sample $E_1(\TX')$ is the projection of $P_1$ into the separating hyperplane determined by the test, which is the function of the training sequence and can vary accordingly.

\begin{figure}[!h]
	\centering
	\begin{tikzpicture}[scale=0.8]
	\draw [line width=0.3mm]  (5.5,1.2) -- (0,-5) -- (11,-5) --(5.5,1.2)  ;
	\draw [line width=0.3mm, dashed] (3.1,-1.5) --   (4.5,-5)  ;
	\node at (7.4,-4.5) {\small $D(\Tx\|P_0)-D(\Tx\|\TX')    =   \gamma(E_0,\TX') $};
	\node[draw,circle,inner sep=1pt,fill] at (3,-3.5) {};
	\node[draw,circle,inner sep=1pt,fill] at (6.5,-2.6) {};
	\node at (2.7,-3.4) {\small $ P_0$};
	\node at (6.8,-2.95) {\small $\TX'$};
	\draw (2.8,-3.5) ellipse (1.05cm and 0.7cm);
	\node at (1,-4.7) {$\mathcal{P}({\Xc})$};
	\draw [<-,>=stealth] (6.5,-2.6) -- (6.,-1.85);
	\draw [->,>=stealth] (3,-3.5) -- (3.2,-4.15);
	\node at (2.9,-3.9) {\small $E_0$};
	\node[draw,circle,inner sep=1pt,fill] at (6,-1.85) {};
     \draw [line width=0.25mm, gray]plot [smooth, tension=1.5] coordinates{(6,-1.85)  (5,-2.1)  (3.55,-2.7)} ;
	\draw [line width=0.25mm, gray]plot [smooth, tension=1.5] coordinates{(3,-3.5)  (4.1,-3.2) (5.5,-2.8) (6.5,-2.6)} ;
	\node at (4.6,-1.75) {\small $ E_1(\TX')$};
	\node at (6.1,-1.65) {\small $P_1$};
    \node[draw,circle,inner sep=1pt,fill] at (3.57,-2.7) {};
    \node[draw,circle,inner sep=1pt,fill] at (3.8,-3.3) {};
	\end{tikzpicture}
	\caption{Proposed classifier with known distribution $P_0$ and a training sequence with type $\TX$.} 
	\label{fig:classification}
\end{figure}
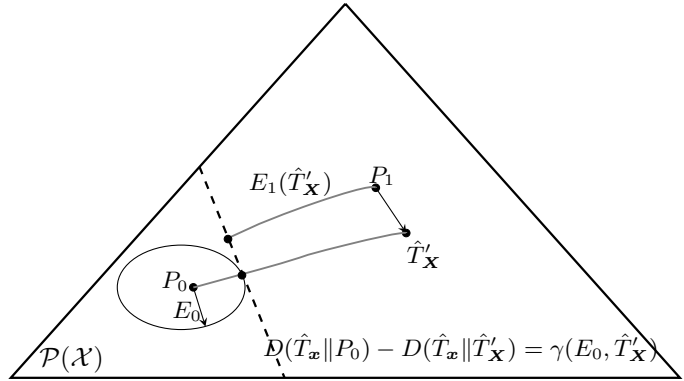

Next, we find a refined expression for the type-\RNum{1} error probability and show that the error probability prefactor is $O(\frac{1}{\sqrt{n}})$, i.e., of the same order of the prefactor achieved by the likelihood ratio test.

\begin{theorem}\label{thm:E0}
For every $P_0, P_1, 0<\beta \leq 1$ the type-\RNum{1} probability of error  of the classifier $\cls$ defined in \eqref{eq:clsfix} can be expressed as
\begin{equation}\label{eq:type0classifier}
\epsilon_0(\cls)=  \frac{1}{\sqrt{n}} e^{-nE_0}(c+ o(1)),
\end{equation}
where $c$ is a positive constant that only depends on the data distributions and $E_0$.
\end{theorem} 

\begin{proof}
The proof can be found in Appendix \ref{sec:proofth1}.
\end{proof}

Having established the improvement of the prefactor of the type-$\RNum{1}$ error probability, we are ready to study the type-\RNum{2} error exponent of our proposed classifier.
\begin{theorem}\label{thm:E1}
	For every $P_0, P_1$ there exists a finite training to sample size ratio $\alpha^*_\beta$ such that for any $\alpha>\alpha^*_\beta$ we have
	\begin{equation}
	E_1(\cls)=E_1^*(E_0),
	\end{equation}
	that is, for any $E_0$, the proposed classifier achieves the optimal error exponent tradeoff defined in \eqref{eq:tradefix}.
\end{theorem}
\begin{proof}
The proof can be found in Appendix \ref{sec:proofth2}.
\end{proof}

We have shown that the classifier proposed in \eqref{eq:clsfix} not only achieves the optimal error exponent tradeoff for $\alpha>\alpha^*_\beta$ but also achieves the same prefactor of the type-\RNum{1} error probability of the likelihood ratio test. This represents a significant improvement with respect to the Hoffding's universal test for observation alphabets $|\Xc|>2$. However, this might come as an expense of worse polynomial decay of type-\RNum{2} error probability. In the following result, by using the refinements of large deviation techniques, we show that this is not the case and that the proposed classifier achieves the same prefactor as the type-\RNum{2} error probability of the likelihood ratio test, establishing the optimality of the proposed classifier up to a constant. 

\begin{theorem}\label{thm:refinement}
	For every $P_0, P_1$, and  $\alpha^*_\beta < \alpha $ the type-\RNum{2} probability of error  of the classifier $\cls$ is 
	
	\begin{equation}\label{eq:type1classifier}
	\epsilon_1(\cls)=  \frac{1}{\sqrt{n}} e^{-nE_1(E_0,P_1)}(c+ o(1)), 
	\end{equation}
	where $c$ is a positive constant depending only on data distributions and $E_0$.
\end{theorem}
 
\begin{proof}
The proof can be found in Appendix \ref{sec:proofth3}.
\end{proof}

\begin{example}\label{ex:ex1}
We present a numerical example to illustrate the performance of the proposed classifier in \eqref{eq:clsfix}. Consider two trinary distributions $P_0=[0.3,0.3,0.4], P_1=[0.35,0.35,0.3]$ and set $\alpha= 2$ and $E_0=0.005$. Each point in the figure is obtained by estimating the average error probability as follows. For each length of the test sequence $n$, we estimate the type-\RNum{1} and type-\RNum{2} error probabilities of the classifier in \eqref{eq:clsfix} as well as those of the likelihood ratio test in \eqref{eq:LRTtype} and Hoeffding's classifier over $5*10^7$ independent experiments. As it can be seen in Figures \ref{fig:e0}, the type-\RNum{1} error exponent of the likelihood ratio test and the proposed classifier are very close to each other and outperform the Hoeffding's test. In addition, we observe that the type-\RNum{2} error exponent of the proposed classifier is slightly worse than that of the likelihood ratio and Hoeffding's tests for small $n$; as $n$ increases, the proposed classifier achieves the optimal exponent. In addition, in order to clearly observe the effect of the prefactor, in Figures \ref{fig:e0asymp} we plot $\log\frac{\epsilon_i(\phi)}{\frac{1}{\sqrt n}e^{-nE_i}}$  for $i\in \{0,1 \}$ for the same classifiers. We first notice that in the case of  $\log\frac{\epsilon_1(\phi)}{\frac{1}{\sqrt n}e^{-nE_1}}$, the three classifiers achieve the optimal $\frac{1}{\sqrt{n}}$ prefactor, though with different constants. We also observe that  $\log\frac{\epsilon_0(\phi)}{\frac{1}{\sqrt n}e^{-nE_0}}$ converges to a constant for the likelihood ratio test and the new classifier, as predicted by our analysis. Instead, as \eqref{eq:Hoefasymp} suggests, Hoeffding's classifier fails to attain the optimal $\frac{1}{\sqrt{n}}$ prefactor for $\epsilon_0$.

\end{example}

\begin{figure}[htpb]
		 \centering
%
%
\definecolor{mycolor1}{rgb}{0.00000,0.44700,0.74100}%
\definecolor{mycolor2}{rgb}{0.85000,0.32500,0.09800}%
\definecolor{mycolor3}{rgb}{0.92900,0.69400,0.12500}%
\begin{tikzpicture}

\begin{axis}[%
width=3.9096in,
height=2.905in,
at={(1.011in,0.651in)},
scale only axis,
xmin=0,
xmax=2000,
xlabel style={font=\color{white!15!black}},
xlabel={$n$},
ymin=0,
ymax=0.06,
ylabel style={font=\color{white!15!black}},
ylabel={$-\frac{1}{n} \log{\epsilon_0}$},
axis background/.style={fill=white},
xmajorgrids,
ymajorgrids,
legend style={legend cell align=left, align=left, draw=white!15!black}
]
\addplot [color=blue, line width=1.0pt, mark=*, mark options={solid, blue}]
  table[row sep=crcr]{%
20	0.0550557709443114\\
40	0.0343591447563654\\
60	0.0259316190917678\\
80	0.0211398801788295\\
100	0.0185553028512672\\
200	0.0126421614794129\\
300	0.0105145097628423\\
400	0.0094200231132231\\
500	0.0086602350598729\\
600	0.00815475376406645\\
700	0.00779917932420399\\
800	0.00752585925723111\\
900	0.00729983943374415\\
1000	0.00711527157441065\\
1100	0.0069389145532319\\
1200	0.00683775535855637\\
1300	0.00668626307221598\\
1400	0.00656537044123971\\
1500	0.00649765586915097\\
1600	0.00643625795037872\\
1700	0.0063193042607495\\
1800	0.00630411723694036\\
1900	0.00618349884131121\\
2000	0.00597960628379932\\
};
\addlegendentry{New Classifier}

\addplot [color=red, densely dotted,line width=1.0pt, mark=o, mark options={solid, red}]
  table[row sep=crcr]{%
20	0.0441669815867235\\
40	0.0288636980249603\\
60	0.0222283651833453\\
80	0.0190350200721819\\
100	0.0171589903353596\\
200	0.012377673519945\\
300	0.0104641106573867\\
400	0.00946105415208606\\
500	0.0086856118430412\\
600	0.00817712546462002\\
700	0.00775164391151132\\
800	0.00747441997030339\\
900	0.00726858248902153\\
1000	0.00711527157441065\\
1100	0.0069389145532319\\
1200	0.00683775535855637\\
1300	0.00668626307221598\\
1400	0.00656537044123971\\
1500	0.00649765586915097\\
1600	0.00643625795037872\\
1700	0.0063193042607495\\
1800	0.00630411723694036\\
1900	0.00618349884131121\\
2000	0.00597960628379932\\
};
\addlegendentry{LRT}

\addplot [color=black, dashed, line width=1.0pt, mark=*, mark options={solid, black}]
  table[row sep=crcr]{%
20	0.0022394091395067\\
40	0.0034442232436224\\
60	0.00436690542838672\\
80	0.00497243252254931\\
100	0.00493208389237933\\
200	0.00496089374453306\\
300	0.00496407427460175\\
400	0.00500128017201726\\
500	0.00499958580456907\\
600	0.00500030359937937\\
700	0.00501633777824441\\
800	0.00499276209552012\\
900	0.00500571541324552\\
1000	0.00500432027586349\\
1100	0.0049935434152163\\
1200	0.00499806987243492\\
1300	0.00500685542090989\\
1400	0.00499701554132185\\
1500	0.00501030598514974\\
1600	0.00500761528310738\\
1700	0.00499702504851848\\
1800	0.00501510323826574\\
1900	0.00501102628898509\\
2000	0.00502022670380464\\
};
\addlegendentry{GLRT}

\end{axis}

\end{tikzpicture}%
%
%
\definecolor{mycolor1}{rgb}{0.00000,0.44700,0.74100}%
\definecolor{mycolor2}{rgb}{0.85000,0.32500,0.09800}%
\definecolor{mycolor3}{rgb}{0.92900,0.69400,0.12500}%
\begin{tikzpicture}

\begin{axis}[%
width=3.9096in,
height=2.905in,
at={(1.011in,0.651in)},
scale only axis,
xmin=0,
xmax=600,
xlabel style={font=\color{white!15!black}},
xlabel={$n$},
ymin=0,
ymax=0.2,
ylabel style={font=\color{white!15!black}},
ylabel={$-\frac{1}{n} \log{\epsilon_1}$},
axis background/.style={fill=white},
xmajorgrids,
ymajorgrids,
legend style={legend cell align=left, align=left, draw=white!15!black}
]
\addplot [color=blue, line width=1.0pt, mark=*, mark options={solid, blue}]
  table[row sep=crcr]{%
20	0.0420916410478073\\
40	0.0258381760503044\\
60	0.0208924712794868\\
80	0.0180835188891873\\
100	0.0164851195634384\\
200	0.0127634633918716\\
300	0.0112170443780314\\
400	0.0102902640650093\\
500	0.00968168377575914\\
600	0.00923579994762714\\
};
\addlegendentry{New Classifier}

\addplot [color=red, densely dotted,line width=1.0pt, mark=o, mark options={solid, red}]
  table[row sep=crcr]{%
20	0.0741534806809647\\
40	0.0413112450756848\\
60	0.0304329718769123\\
80	0.0249543896893467\\
100	0.0215943758377556\\
200	0.0145846054656485\\
300	0.0120335609122988\\
400	0.0107999781093861\\
500	0.0100262768486287\\
600	0.00942498718414462\\
};
\addlegendentry{LRT}

\addplot [color=black, dashed, line width=1.0pt, mark=*, mark options={solid, black}]
  table[row sep=crcr]{%
20	0.184684550555558\\
40	0.0713124526367021\\
60	0.0444596711424749\\
80	0.033125214623145\\
100	0.0272596165778936\\
200	0.0168800025299494\\
300	0.013656780586642\\
400	0.0118924257430854\\
500	0.0108584016788739\\
600	0.0101950009326153\\
};
\addlegendentry{GLRT}

\end{axis}
\end{tikzpicture}%
		\caption{Type-\RNum{1} and type-\RNum{2}  error exponents for the likelihood ratio test, Hoeffding's test and the proposed classifier. } 
		\label{fig:e0}
	\end{figure}
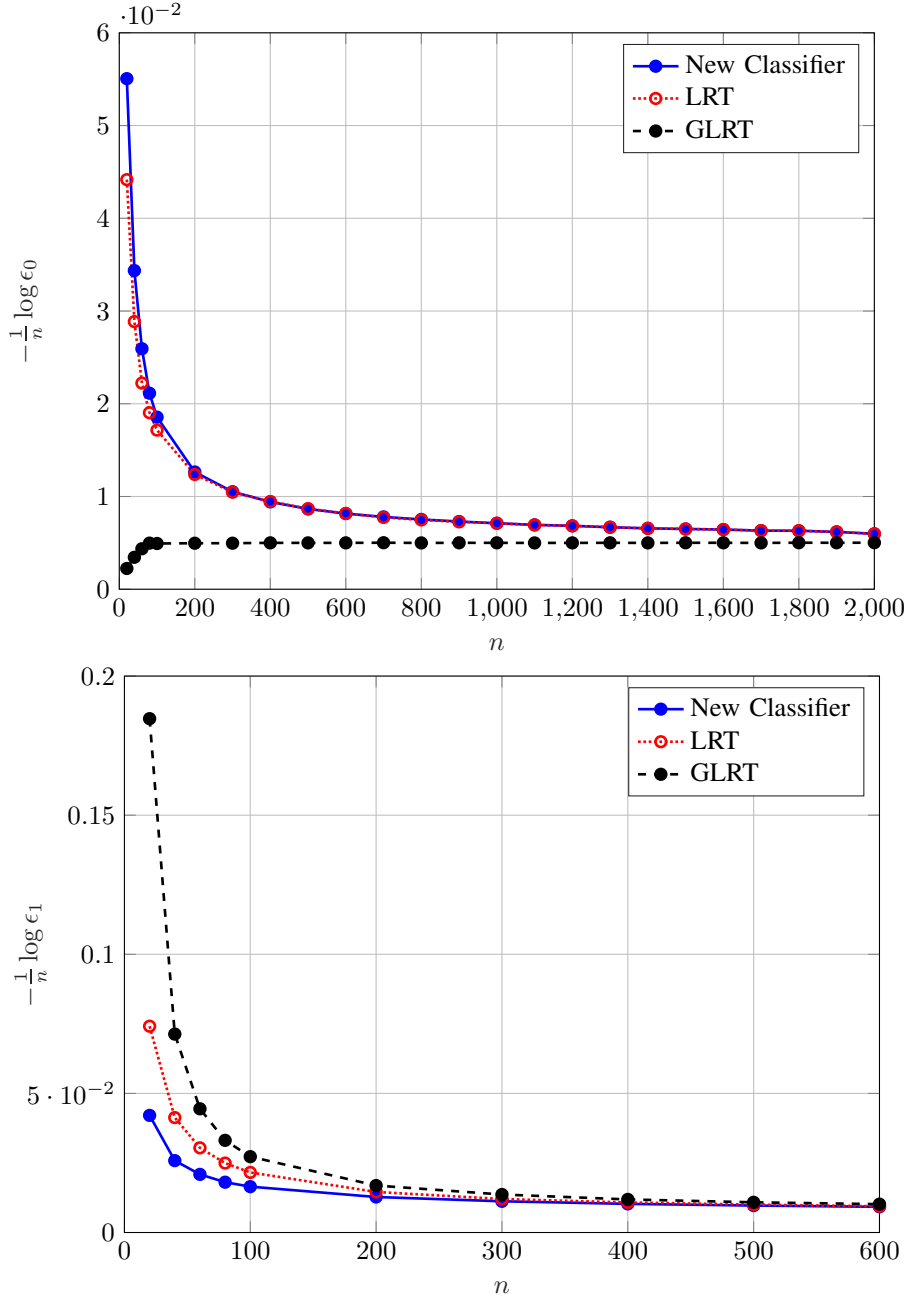

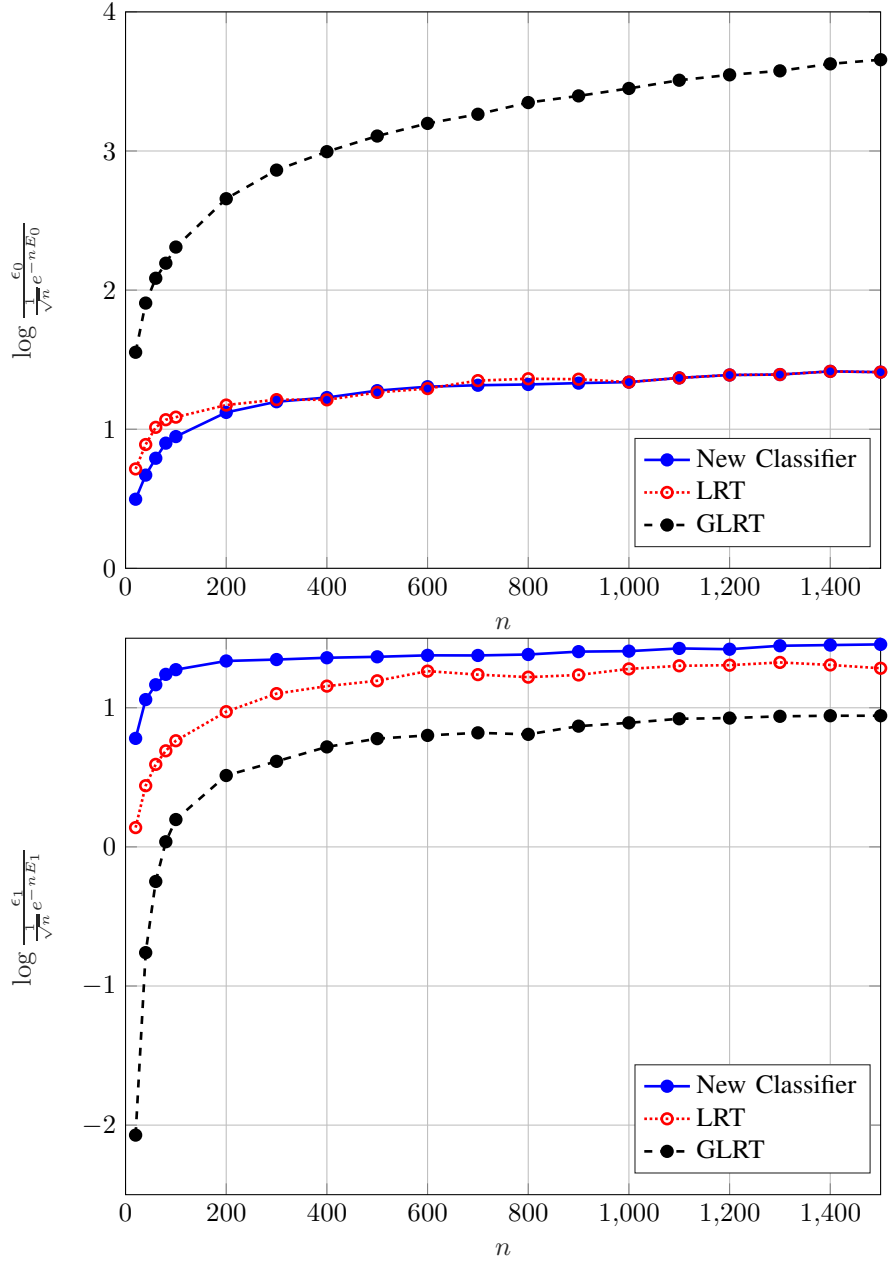
\begin{figure}[htpb]
		 \centering
%
%
\definecolor{mycolor1}{rgb}{0.00000,0.44700,0.74100}%
\definecolor{mycolor2}{rgb}{0.85000,0.32500,0.09800}%
\definecolor{mycolor3}{rgb}{0.92900,0.69400,0.12500}%
\begin{tikzpicture}

\begin{axis}[%
width=3.9096in,
height=2.905in,
scale only axis,
xmin=0,
xmax=1500,
xlabel style={font=\color{white!15!black}},
xlabel={$n$},
ymin=0,
ymax=4,
ylabel style={font=\color{white!15!black}},
ylabel={$\log \frac{\epsilon_0}{{\frac{1}{\sqrt{n}}}e^{-nE_0}}$},
axis background/.style={fill=white},
xmajorgrids,
ymajorgrids,
legend style={at={(0.674,0.038)},anchor=south west,, legend cell align=left, align=left, draw=white!15!black}
]
\addplot [color=blue, line width=1.0pt, mark=*]
  table[row sep=crcr]{%
20	0.496750717890767\\
40	0.670073936802351\\
60	0.79127513560498\\
80	0.899822903030581\\
100	0.947054807867327\\
200	1.12072638739145\\
300	1.1975383084754\\
400	1.22772302826475\\
500	1.27718651927464\\
600	1.3056125691682\\
700	1.31611464057891\\
800	1.32161845804908\\
900	1.33134189129242\\
1000	1.33860606508042\\
1100	1.36872672083814\\
1200	1.38970669709784\\
1300	1.39291777784404\\
1400	1.41576540709503\\
1500	1.41012638981869\\
};
\addlegendentry{New Classifier}

\addplot [color=red, densely dotted, line width=1.0pt, mark=o,mark options={solid, red}]
  table[row sep=crcr]{%
20	0.714526505042526\\
40	0.889891806058555\\
60	1.01347037011033\\
80	1.06821171156239\\
100	1.08668605945809\\
200	1.17362397928501\\
300	1.2126580401121\\
400	1.21131061271957\\
500	1.2644981276905\\
600	1.29218954883606\\
700	1.34938942946378\\
800	1.36276988759126\\
900	1.35947314154278\\
1000	1.33860606508042\\
1100	1.36872672083814\\
1200	1.38970669709784\\
1300	1.39291777784404\\
1400	1.41576540709503\\
1500	1.41012638981869\\
};
\addlegendentry{LRT}

\addplot [color=black, dashed, line width=1.0pt, mark=*,mark options={solid, black}]
  table[row sep=crcr]{%
20	1.55307795398686\\
40	1.90667079731207\\
60	2.08515795540785\\
80	2.193218715533\\
100	2.30937670375611\\
200	2.65697993436741\\
300	2.86266895494758\\
400	2.99522020474709\\
500	3.10751114692656\\
600	3.19828266798045\\
700	3.26410372275061\\
800	3.34809618741787\\
900	3.39605350974119\\
1000	3.44955736362758\\
1100	3.5086349726553\\
1200	3.54735457096615\\
1300	3.57614772454196\\
1400	3.62629199995109\\
1500	3.65615121582054\\
};
\addlegendentry{GLRT}

\end{axis}

\end{tikzpicture}%
%
%
\definecolor{mycolor1}{rgb}{0.00000,0.44700,0.74100}%
\definecolor{mycolor2}{rgb}{0.85000,0.32500,0.09800}%
\definecolor{mycolor3}{rgb}{0.92900,0.69400,0.12500}%
\begin{tikzpicture}

\begin{axis}[%
width=3.9096in,
height=2.905in,
scale only axis,
xmin=0,
xmax=1500,
xlabel style={font=\color{white!15!black}},
xlabel={$n$},
ymin=-2.5,
ymax=1.5,
ylabel style={font=\color{white!15!black}},
ylabel={$\log \frac{\epsilon_1}{{\frac{1}{\sqrt{n}}}e^{-nE_1}}$},
axis background/.style={fill=white},
xmajorgrids,
ymajorgrids,
legend style={at={(0.674,0.038)},anchor=south west, legend cell align=left, align=left, draw=white!15!black}
]
\addplot [color=blue, line width=1.0pt, mark=*]
  table[row sep=crcr]{%
20	0.78003331582085\\
40	1.05891268504479\\
60	1.16562400434184\\
80	1.24033180620196\\
100	1.2740731366502\\
200	1.3364660048997\\
300	1.34677792391868\\
400	1.35962664755029\\
500	1.36646216133152\\
600	1.37698485903179\\
700	1.3761884552086\\
800	1.38304464466562\\
900	1.40369889989909\\
1000	1.40709319185616\\
1100	1.42673275799722\\
1200	1.42087617095554\\
1300	1.44585797503964\\
1400	1.45084480630002\\
1500	1.45561329944019\\
};
\addlegendentry{New Classifier}

\addplot [color=red, densely dotted, line width=1.0pt, mark=o,mark options={solid, red}]
  table[row sep=crcr]{%
20	0.138796523157702\\
40	0.439989924029577\\
60	0.593193968496312\\
80	0.690662142189203\\
100	0.763147509218485\\
200	0.972237590144317\\
300	1.10182296363845\\
400	1.15574102979956\\
500	1.19416562489672\\
600	1.2634725171213\\
700	1.23841313960174\\
800	1.22019845362197\\
900	1.23596255123015\\
1000	1.2794852467065\\
1100	1.3018127170257\\
1200	1.30629313802982\\
1300	1.32605677522702\\
1400	1.30798208624661\\
1500	1.28410265875046\\
};
\addlegendentry{LRT}

\addplot [color=black, dashed, line width=1.0pt, mark=*,mark options={solid, black}]
  table[row sep=crcr]{%
20	-2.07182487433416\\
40	-0.760058378411115\\
60	-0.248407987437446\\
80	0.0369961474853406\\
100	0.196623435204683\\
200	0.513158177284136\\
300	0.614857061335491\\
400	0.718761976319825\\
500	0.778103209774147\\
600	0.80146426803889\\
700	0.819898550290024\\
800	0.809073251390055\\
900	0.868001555079049\\
1000	0.891806404614616\\
1100	0.920999906242523\\
1200	0.926065962511522\\
1300	0.938610172621534\\
1400	0.942767380759747\\
1500	0.942809435662099\\
};
\addlegendentry{GLRT}

\end{axis}

\end{tikzpicture}%
		\caption{$\log\frac{\epsilon_i(\phi)}{{\frac{1}{\sqrt{n}}}e^{-nE_i}}$  for $i\in \{0,1 \}$ and for the likelihood ratio test, Hoeffding's test and the proposed classifier. } 
		\label{fig:e0asymp}
	\end{figure} 

Theorem \ref{thm:E1} shows that the classifier in \eqref{eq:clsfix} can achieve the optimal error exponent tradeoff if the training to sample ratio $\alpha$ is large enough. As it is evident by the proof of Theorem \ref{thm:E1}, the proposed classifier cannot achieve the optimal error exponent tradeoff for $\alpha$ close to zero. Therefore, it is desirable to find the smallest $\alpha^*_\beta$ such that the above theorem holds. The next results introduce  lower and upper bounds to $\alpha^*_\beta$.

\begin{theorem}\label{thm:lower}
		 For every $P_0, P_1 \in \Pc(\Xc)$, and $E_0 > 0$, we have $ \underline{\alpha}_\beta  \leq \alpha^*_\beta $ where 
	\begin{equation}
	\underline{\alpha}_\beta=- \Lambda_{\min}(\Hm) ,
	\end{equation}
	and $ \Lambda_{\min}$ is the smallest eigenvalue of the matrix
\begin{equation}\label{eq:Hess}
\Hm=\beta \eta^*_\beta \sqrt{\Jm} \Big [  \Qm+ \eta^*_1 \Vm +(1-\eta^*_1)\Wm    -\Tm      \Big ] \sqrt{\Jm},
\end{equation}
	where $\Wm=\wv \wv^T, \Vm=\vv \vv^T, \Qm=\qv \qv^T $ and
	\begin{align}
	  \wv^T= \frac{1}{{\sqrt{ \Var_{Q_{\mu^*}} \Big (\log \frac{Q_{\mu^*}}{P_0} \Big ) }}} \bigg ( Q_{\mu^*}(1) {\log\Big ( \frac{Q_{\mu^*}(1)}{P_0(1)} \Big )-E_0 }, \ldots,   Q_{\mu^*}(|\Xc|){\log\Big ( \frac{Q_{\mu^*}(|\Xc|)}{P_0(|\Xc|)} \Big )-E_0 } \Bigg),
	\end{align} 
	\begin{align}
	\vv=\frac{1}{\sqrt{\Var{Q_{\eta^*_\beta} }(\Omega)   }}    \big ( Q_{\eta^*_\beta}(1)\Omega(1),\ldots,Q_{\eta^*_\beta}(|\Xc|)\Omega(|\Xc|)   \big ),
	\end{align}
	\begin{align}
	\Omega(i)=\beta \log \frac{P_1(i)}{P_0(i)} +(1-\beta)  \log \frac{Q_{\eta^*_\beta}(i)}{P_0(i)},\ \  i \in \Xc,
	\end{align}	
	\begin{align}
	\qv= \big (Q_{\mu^*} (1),\ldots, Q_{\mu^*}(|\Xc|)\big),
	\end{align}
	 \begin{align}
	\Tm= \diag (Q_{\mu^*} ), \quad \Jm=\diag \bigg(\frac{1}{P_1}  \bigg) \label{eq:fisher}
	\end{align}
	 where $Q_{\mu^*}$ is defined in \eqref{eq:tiltedH} when $Q_1=P_1$, and ${Q}_{\eta_\beta^*}$ defined in \eqref{eq:tiltbeta}, is the projection of $P_1$ onto $\Ac_0(P_1,\beta)$ which equals to $Q_{\mu^*}$. Furthermore, $\eta^*_\beta$ is the optimal Lagrange multiplier in \eqref{eq:lagrangeE} for $0<\beta\leq 1$. Moreover, for $|\Xc| \geq 6$ by further lower bounding $\underline{\alpha}_\beta$ we have
	\begin{align}\label{eq:lower}
	\beta \eta^*_\beta \Bigg [\frac{Q_{\eta^*_\beta}}{P_1} \Bigg ] _{(3)}  \leq \underline{\alpha}_\beta,
	\end{align}
 	where $\Big  [\frac{Q_{\eta^*_\beta}}{P_1} \Big ]_{(3)}$ is the third largest value of $\frac{Q_{\eta^*_\beta(i)}}{P_1(i)}$ for $i \in \{1,\ldots, |\Xc|\}$.
	
\end{theorem}

\begin{proof}
The proof can be found in Appendix \ref{sec:proofth5}.
\end{proof}

\begin{example}\label{exp:lower}
Letting the distributions $P_0=\text{Bern}(0.3), P_1=\text{Bern}(0.4)$, we have set $E_0=0.005$. Figure \ref{fig:alpha_beta} shows the relation of the $\underline{\alpha}_\beta$ for $0<\beta<1$. As can be seen from the figure, $\underline{\alpha}_\beta$ is increasing in $\beta$ and is equal to zero for $\beta=0$ as expected, since for $\beta=0$ the proposed classifier is equal to universal Hoeffding's test and achieves the optimal type-\RNum{2} error exponent for every $P_1$. However, $\beta=0$, as discussed before, is the singularity point as the type-\RNum{1} error probability prefactor looses its independence from dimension. 
	
\end{example}

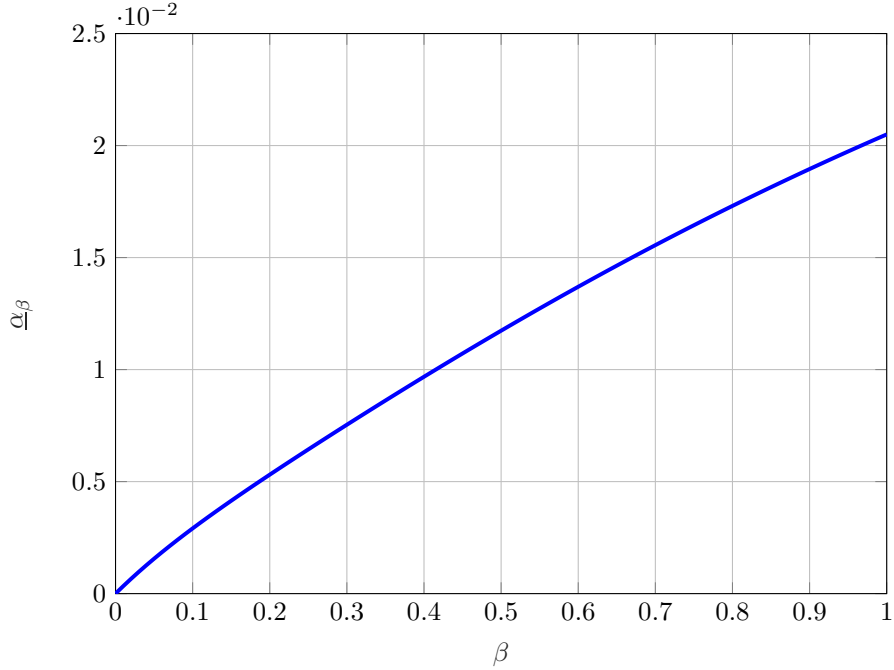
\begin{figure}[h]
		 \centering
%
%
\definecolor{mycolor1}{rgb}{0.00000,0.44700,0.74100}%
\begin{tikzpicture}

\begin{axis}[%
width=3.993in,
height=2.924in,
scale only axis,
xmin=0,
xmax=1,
xlabel style={font=\color{white!15!black}},
xlabel={$\beta$},
ymin=-0,
ymax=0.025,
ylabel style={font=\color{white!15!black}},
ylabel={$\underline{\alpha}_\beta$},
axis background/.style={fill=white},
xmajorgrids,
ymajorgrids
]
\addplot [color=blue, line width=1.5pt]
  table[row sep=crcr]{%
0	-0\\
0.01	0.000334066617038733\\
0.02	0.000655923924337897\\
0.03	0.000966710427175475\\
0.04	0.0012674828164452\\
0.05	0.00155921876254749\\
0.06	0.00184281976430822\\
0.07	0.00211911409764447\\
0.08	0.00238885989189271\\
0.09	0.00265274834547526\\
0.1	0.0029114070777844\\
0.11	0.00316540360153249\\
0.12	0.00341524888983482\\
0.13	0.00366140100521168\\
0.14	0.00390426875350661\\
0.15	0.00414421532422311\\
0.16	0.00438156187961004\\
0.17	0.00461659105752349\\
0.18	0.00484955035715565\\
0.19	0.00508065538166212\\
0.2	0.00531009291709005\\
0.21	0.00553802383243972\\
0.22	0.0057645857908869\\
0.23	0.00598989576694145\\
0.24	0.00621405236848894\\
0.25	0.00643713796618484\\
0.26	0.00665922063553184\\
0.27	0.00688035591919316\\
0.28	0.00710058841872987\\
0.29	0.00731995322606389\\
0.3	0.00753847720563573\\
0.31	0.00775618013852249\\
0.32	0.00797307573978118\\
0.33	0.0081891725600489\\
0.34	0.00840447478202558\\
0.35	0.00861898292193568\\
0.36	0.00883269444545644\\
0.37	0.00904560430694009\\
0.38	0.00925770542008293\\
0.39	0.00946898906751335\\
0.4	0.00967944525611235\\
0.41	0.00988906302424105\\
0.42	0.0100978307064522\\
0.43	0.0103057361606971\\
0.44	0.0105127669625204\\
0.45	0.0107189105702543\\
0.46	0.0109241544647864\\
0.47	0.0111284862670792\\
0.48	0.0113318938362554\\
0.49	0.0115343653507454\\
0.5	0.0117358893746988\\
0.51	0.0119364549116024\\
0.52	0.0121360514468182\\
0.53	0.0123346689805468\\
0.54	0.0125322980525375\\
0.55	0.012728929759709\\
0.56	0.012924555767695\\
0.57	0.0131191683172067\\
0.58	0.0133127602259901\\
0.59	0.0135053248870592\\
0.6	0.0136968562637965\\
0.61	0.0138873488824377\\
0.62	0.0140767978223914\\
0.63	0.0142651987047819\\
0.64	0.014452547679555\\
0.65	0.0146388414114383\\
0.66	0.0148240770650115\\
0.67	0.0150082522891009\\
0.68	0.0151913652006901\\
0.69	0.0153734143685031\\
0.7	0.0155543987964019\\
0.71	0.0157343179067102\\
0.72	0.0159131715235677\\
0.73	0.0160909598563931\\
0.74	0.0162676834835316\\
0.75	0.0164433433361391\\
0.76	0.0166179406823553\\
0.77	0.0167914771118039\\
0.78	0.0169639545204488\\
0.79	0.0171353750958348\\
0.8	0.0173057413027279\\
0.81	0.0174750558691724\\
0.82	0.0176433217729724\\
0.83	0.0178105422286041\\
0.84	0.0179767206745625\\
0.85	0.0181418607611423\\
0.86	0.018305966338651\\
0.87	0.0184690414460509\\
0.88	0.0186310903000229\\
0.89	0.0187921172844478\\
0.9	0.018952126940297\\
0.91	0.0191111239559197\\
0.92	0.0192691131577243\\
0.93	0.0194260995012384\\
0.94	0.0195820880625395\\
0.95	0.0197370840300454\\
0.96	0.0198910926966542\\
0.97	0.0200441194522206\\
0.98	0.0201961697763621\\
0.99	0.0203472492315784\\
1	0.0204973634566779\\
};
\end{axis}

\end{tikzpicture}%
		\caption{Lower bound on the required training to test sample ratio $\underline{\alpha}_\beta$ for achieving the optimal error exponent trade-off with the proposed classifier for $0< \beta\leq 1$.  }
		\label{fig:alpha_beta}
	\end{figure}

%

We next find an upper bound on the optimal training to observation ratio $\alpha^*_\beta$.

\begin{theorem}\label{thm:upper}
	 For every $P_0, P_1 \in \Pc(\Xc)$, and $E_0 > 0$, we have that $\alpha^*_\beta \leq \bar{\alpha} $ where 
\begin{equation}
\bar{\alpha}=\frac{\lambda^*(4+\lambda^*)(1+\kappa)   }{(P_1^{\text{min}})^2},
\end{equation}
$P_1^{\text{min}}\triangleq \min_{x\in\Xc} P_1(x)$, $\lambda^*$ is the optimal Lagrange multiplier in \eqref{eq:tilted}, and 
\begin{equation}\label{eq:upper}
\kappa =\sqrt{\frac{E_1(E_0,P_1)}{\lambda^*(4+\lambda^*)} }.
\end{equation}
\end{theorem}

\begin{proof}
The proof can be found in Appendix \ref{sec:proofth7}.
\end{proof}

Note that the upper bound \eqref{eq:upper} and lower bound \eqref{eq:lower} to $\alpha^*_\beta$ suggests that as $P_1$ approaches the probability simplex boundaries, i.e., as $P_1^\text{min}$ approaches zero, the classification problem becomes more challenging, and we need more training samples to achieve the optimal exponents. For the classification problem in the Example \ref{exp:lower}, the upper bound in \eqref{eq:upper} gives $\bar{\alpha}=14.19$, i.e., if the ratio of training samples to the number of test samples exceeds $14.19$ the proposed test with $\beta=1$ achieves \eqref{eq:type0classifier}, \eqref{eq:type1classifier}. See \cite{viswanath}  for general conditions between alphabets size and training sample size when both hypotheses are unknown.

\section{Stein Regime Classification}\label{sec:steinclas}

In this section, we will study the classification problem in the Stein regime. For the case where both hypotheses are known, the Stein regime is defined as the highest error exponent  under one hypothesis  when the  error probability under the other hypothesis is at most some fixed $\epsilon$ (see e.g., \cite{Cover})
\begin{align}\label{eq:steindef}
E_1^{(\epsilon)} \triangleq \sup \big \{E_1\in \mathbb{R}_{+}: \exists \phi , \exists n_0 \in \ZZ_+  \  \text{s.t.} \  \forall   n>n_0 ~
\epsilon_0 (\phi)\leq \epsilon  \quad \text{and} \quad \epsilon_1(\phi) \leq e^{-nE_1} \big \}.
\end{align} 
The the optimal $E_1^{(\epsilon)}$, given by \cite{Cover}
\begin{equation}\label{eq:Stein2trade}
E_1^{(\epsilon)} (\phi^{\rm lrt}) = D(P_0\|P_1),
\end{equation}
can be achieved by setting the threshold in \eqref{eq:LRTtype} to be ${\gamma} = -D(P_0\|P_1)+\frac{c}{\sqrt{n}}$, where $c$ is a constant that depends on distributions $P_0, P_1$ and $\epsilon$. Similarly, our setting where $P_0$ is known and a training sequence of the second hypothesis is available, we can define the Stein exponent as the highest error exponent  under one hypothesis when the  error probability under the other hypothesis is at most some fixed $\epsilon$ for all probability distributions $\tilde{P}_1$, .i.e.,
\begin{equation}\label{eq:Stein1trade}
E_1^{(\epsilon)} \triangleq \sup \{ E_1\in \mathbb{R}_{+}: \exists \phi , \exists n_0 \in \ZZ_+  \  \text{s.t.} \  \forall   n>n_0, ~   \forall  \tilde{P_1} \in \Pc(\Xc), \epsilon_0(\phi| P_0, \tilde{P_1})\leq \epsilon \  \text{and} \  \epsilon_1(\phi| P_0,P_1) \leq e^{-nE_1}   \},
\end{equation}   
where $\epsilon_0(\phi| P_0, \tilde{P_1})$ is the error probability when the generating distributions are $P_0, \tilde{P}_1$, i.e., for all possible distributions $\tilde{P}_1$ the type-\RNum{1} probability of error is bounded by some $\epsilon$ and $E_1^{(\epsilon)}$  is the maximum achievable type-\RNum{2} error exponent under the actual distribution generating data $P_1$. Similarly,
\begin{equation}\label{eq:Stein0}
E_0^{(\epsilon)} \triangleq \sup \{ E_0\in \mathbb{R}_{+}: \exists \phi , \exists n_0 \in \ZZ_+  \  \text{s.t.} \  \forall   n>n_0 ,~  \epsilon_0(\phi| P_0, P_1)\leq e^{-nE_0},   \ \    \epsilon_1(\phi| P_0,\tilde{P}_1) \leq  \epsilon \ \ \forall  \tilde{P_1} \in \Pc(\Xc)  \}.
\end{equation}  
In other words, we are interested in the best possible error exponent for one of the hypothesis while the probability of error under alternative hypothesis is some $\epsilon$ universally for any distribution $P_1$.

\begin{theorem}\label{thm:Stein}
	Let $\epsilon \in (0,1) $, then for any probability distributions $P_0, P_1$, the Stein regime exponents are given by
	\begin{align}
	E_1^{(\epsilon)} &=D(P_0\|P_1),\label{eq:stein1}\\
	E_0^{(\epsilon)} &=D_{\frac{\alpha}{1+\alpha}}  (P_1\|P_0)\label{eq:stein2},
	\end{align}
	where $D_\rho(P_1\|P_0) = \frac{1}{\rho-1} \log \sum_{x\in \Xc} P_1^\rho P_2^{1-\rho}$ is the R\'enyi divergence of order $\rho$.
\end{theorem}

\begin{proof}
The proof can be found in Appendix \ref{sec:proofth41}.
\end{proof}

Observe that $E_1^{(\epsilon)}$ is equal to the Stein exponent for the likelihood ratio test where both distributions are known. However, since R\'enyi divergence is a non-decreasing function of its order and $\frac{\alpha}{1+\alpha} <1$ hence $E_0^{(\epsilon)}$ is strictly smaller than the Stein regime exponent achieved by the likelihood ratio test.


\section{Sequential Classification with a Known Hypothesis}\label{sec:seqclas}

In this section, we study sequential classification with a known hypothesis. When both hypotheses are known, the sequential probability ratio test can achieve higher exponents compared to the likelihood ratio test. When only one of the hypotheses is known, Hoeffding's test can achieve the best error exponent achieved by the likelihood ratio test in the fixed sample size scenario. However, there is no counterpart  to Hoeffding's test in the sequential case when only one of the hypotheses is known, i.e., no classifier can achieve the same error exponent performance as the sequential probability ratio test \cite{Wald}. We propose a classifier inspired by the sequential probability ratio test and show that having training samples from the second hypothesis can improve the error exponent tradeoff compared to the fixed sample-sized classification.

In the sequential setting, the number of samples is a random variable called the stopping time $\tau$, taking values in $\ZZ_{+}$. A sequential classifier is a pair $\Phi=(\phi: \Xc^\tau\times \Xc^{\alpha\tau} \rightarrow \{0,1\},\tau)$, where for every $n\geq 0$ the event $\{\tau\leq n\} \in \mathscr{F}_n$, and  $\mathscr{F}_n$ is the sigma-algebra induced by random variables $\xv^n , \X^{\alpha n} $, i.e., $ \sigma(\xv^n , \X^{\alpha n})$. We also assume that at every stage, additional training and test samples are available to the classifier, such that $\alpha=\frac{k}{n}$ remains constant. Moreover, $\phi$ is a $\mathscr{F}_{\tau}$ measurable decision rule, i.e., the decision rule determined by casually observing the sequence $\xv^n, \X^{\alpha n}$. In other words, at each time instant, the test attempts to decide in favor of one of the hypotheses or chooses to take new samples from the source $P_1$ as well as new samples from the unknown data source. 

The two possible pairwise error probabilities that measure the performance of the test are defined as
\begin{equation}\label{eq:errprob}
\epsilon_0(\Phi)=\PP_0\big [\phi(\xv^\tau , \X^{\alpha \tau} )   \neq 0  \big]   ~ \text{,} ~ 	\epsilon_1(\Phi)=\PP_1\big[\phi(\xv^\tau , \X^{\alpha \tau})\neq 1  \big],
\end{equation}
where the probabilities are over $P_0, P_1$, respectively. Similarly to the sequential hypothesis testing case we define  the optimal error exponent as
\begin{align}\label{eq:tradeseq1}
E^*_1(E_0) \triangleq \sup \Big \{E_1\in \mathbb{R}^{+}&: \exists \Phi ,\ \exists \ n  \in \ZZ_{+}   \text{ s.t.} \     \mathbb{E}_{P_0} [\tau] \leq n, \mathbb{E}_{P_1} [\tau]  \leq n, \   \epsilon_0(\Phi) \leq 2^{- n E_0}    ~ \text{and} ~  \epsilon_1(\Phi) \leq 2^{- n E_1} \Big \} .
\end{align}

%

When both hypotheses are known, the sequential probability ratio test (SPRT) $\Phi=(\phi,\tau)$  proposed by Wald  \cite{Wald1} achieves the optimal exponent tradeoff. The sequential probability ratio test is given by
\begin{align}
\tau=\inf \big\{t\geq1:S_t& \geq  \gamma_0 \  \text{or} \  S_t \leq -\gamma_1\big\},
\label{eq:deftau}
\end{align} 
where 
\begin{align}\label{eq:LLR}
S_n=\sum_{i=1}^t \log \frac{P_0(x_i)}{P_1(x_i)},
\end{align} 
is the the accumulated log-likelihood ratio (LLR) of the observed sequence $\xv$ and the thresholds $\gamma_0, \gamma_1$ are two positive real numbers. Moreover, the test makes a decision according to the rule
\begin{align}
\phi(\Tx)= 
\begin{cases}
0&  \text{if } S_\tau \geq \gamma_0  \\
1 &  \text{if } S_\tau \leq  - \gamma_1.\\
\end{cases}
\label{eq:defsprt}
\end{align}    
It is shown in  \cite{Wald1} that the above test attains the optimal error exponent tradeoff, i.e., as thresholds $\gamma_0, \gamma_1$  approach infinity, the test achieves the best error exponent trade-off in \eqref{eq:tradeseq1}. It is known that the error probabilities of sequential probability ratio test as a function of $\gamma_0$ and  $\gamma_1$ are \cite{Wood} 
\begin{align}
\epsilon_0 = c_0 \cdot e^{-\gamma_1 }  \quad , \quad  \epsilon_1 =c_1 \cdot e^{-\gamma_0 } ,
\end{align}
as $\gamma_0, \gamma_1 \rightarrow \infty$ where $c_0, c_1$ are positive constants. Moreover, it can also be shown that
\begin{align}
\mathbb{E}_{P_0} [\tau] &= \frac{  \gamma_0}{D(P_0\|P_1)}(1+o(1))     \label{eq:SPRT1}, \\ 
\mathbb{E}_{P_1} [\tau] &= \frac{  \gamma_1}{D(P_1\|P_0)}(1+o(1))   \label{eq:SPRT2}.
\end{align}
Therefore, according to definition \eqref{eq:tradeseq1}, the optimal error exponent tradeoff is given by,
\begin{equation}
E_0 = D(P_1\|P_0)  , ~ E_1= D(P_0\|P_1),
\end{equation}
where thresholds $\gamma_0, \gamma_1$  are chosen as
\begin{equation}\label{eq:thresh}
\gamma_0=n\big (D(P_0\|P_1)+o(1)\big),~ \gamma_1= n\big(D(P_1\|P_0)+o(1)\big). 
\end{equation}
Hence, the sequential probability ratio test achieves the Stein regime error exponents achievable by the standard likelihood ratio test \cite{Neyman} simultaneously.

For every fixed $n$. we propose the following sequential classifier
\begin{align}
\tau=\inf \big\{t\geq n:S_t(\Tx,\TX)& \geq  \gzt\  \text{or} \  S_t(\Tx,\TX) \leq -\got \big\} ,
\label{eq:deftau}
\end{align} 
where 
\begin{align}\label{eq:testseq}
S_t(\Tx,\TX) =\sum_{i=1}^t \log \frac{P_0(x_i)}{\TX'(x_i)},
\end{align} 
is the accumulated log-likelihood ratio (LLR) using the plugin perturbed type of the training sequence evaluated at the observed sequence $\xv$, and 
\begin{equation}
\TX'=(1-\delta_n) \TX +\frac{\delta_n}{|\Xc|}
\end{equation}
where $\delta_n= o(n^{-1})$
and $\gzt, \got$ are chosen as
\begin{equation}\label{eq:thresh}
\gzt=nD(P_0\|\TX')+ (4 |\Xc|+4) \log(t+1),~ \got= nD(\TX\|P_0)+( 4|\Xc|+4) \log(t+1), 
\end{equation}
and the test makes a decision according to the rule
\begin{align}\label{eq:seqclas}
\phi(\Tx,\TX)= 
\begin{cases}
0&  \text{if } S_\tau(\Tx,\TX) \geq \gzt \\
1 &  \text{if } S_\tau(\Tx,\TX) \leq  -\got,\\
\end{cases}
\end{align}    
where $\Tx, \TX$ are types of the test and training samples at the stopping time $\tau$. As can be seen from the above expressions, the proposed classifier is the plugin sequential probability ratio test, replacing $P_1$ by the  perturbed training type $\TX'$.

The next theorem gives a lower bound on the achievable error exponent tradeoff of the proposed sequential classifier. 
\begin{theorem}\label{thm:seq}
 For every $P_0, P_1$, there exists a training to observation ratio $\alpha^*_{\rm seq}$  such that for any $\alpha \geq \alpha^*_{\rm seq}$, the sequential classifier $\Phi^{\rm seq}=(\phi(\Tx,\TX'),\tau)$ defined in \eqref{eq:deftau}, \eqref{eq:seqclas}  achieves
 \begin{equation}
 E_0(\Phi^{\rm seq})E_1(\Phi^{\rm seq})\geq  D(P_0\|P_1)D_{\frac{\alpha}{1+\alpha}}(P_1\|P_0).
 \end{equation}
Furthermore, the average stopping times of the classifier satisfy
\begin{equation}
\mathbb{E}_{P_0}[\tau]= \mathbb{E}_{P_1}[\tau]= n(1+o(1)).
\end{equation}
\end{theorem}

\begin{proof}
The proof can be found in Appendix \ref{sec:proofth51}.
\end{proof}

This theorem shows that similar to the hypothesis testing problem with known distributions, the proposed sequential classifier can achieve the Stein regime exponents simultaneously when only one of the distributions is known.

\begin{example}\label{ex:ex1}
In Figures \ref{fig:e0seq}, we present a numerical example to illustrate the performance of the proposed sequential classifier in $ \Phi^{\rm seq}$ with $\gzt=nD(P_0\|\TX'), \got= nD(\TX\|P_0)$. Consider two binary distributions $P_0=Bern(0.45)$, $Bern(0.55)$ and set $\alpha= 10$. Each point in the figures is obtained by estimating the average error probability as follows. For each length of the test sequence $n$, we estimate the type-\RNum{1} and type-\RNum{2} error probabilities of the sequential classifier in \eqref{eq:testseq} by generating a sample from the test source and $\alpha$ samples from $P_1$ until the test stops and makes a decision. We have plotted the $-\frac{1}{n} \log \epsilon_i$ for $i\in \{0,1 \}$. We can notice that the type-\RNum{1} error exponent converges to $D_{\frac{\alpha}{1+\alpha}}(P_1\|P_0)$ while the type-\RNum{2} error exponent tends to  $D(P_0\|P_1)$ as Theorem \ref{thm:seq} suggests.

\end{example}

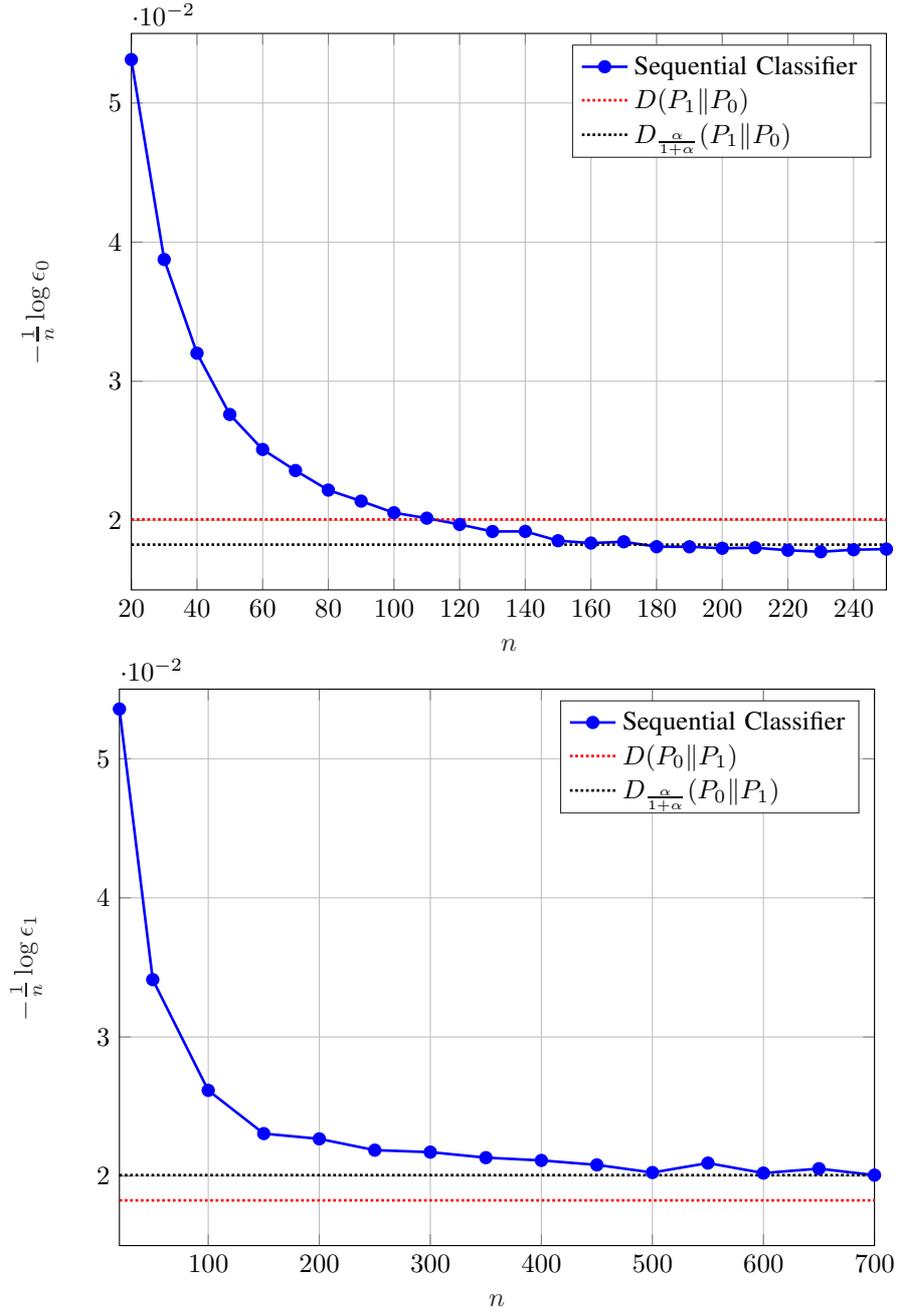
\begin{figure}[htpb]
		 \centering
%
%
\definecolor{mycolor1}{rgb}{0.00000,0.44700,0.74100}%
\definecolor{mycolor2}{rgb}{0.85000,0.32500,0.09800}%
\definecolor{mycolor3}{rgb}{0.92900,0.69400,0.12500}%
\begin{tikzpicture}

\begin{axis}[%
width=3.9096in,
height=2.905in,
at={(1.011in,0.651in)},
scale only axis,
xmin=20,
xmax=250,
xlabel style={font=\color{white!15!black}},
xlabel={$n$},
ymin=0.015,
ymax=0.055,
ylabel style={font=\color{white!15!black}},
ylabel={$-\frac{1}{n} \log{\epsilon_0}$},
axis background/.style={fill=white},
xmajorgrids,
ymajorgrids,
legend style={legend cell align=left, align=left, draw=white!15!black}
]
\addplot [color=blue, line width=1.0pt, mark=*,mark options={solid, blue}]
  table[row sep=crcr]{%
20	0.0531236621026118\\
30	0.0387503671506114\\
40	0.0320153474162786\\
50	0.02760889672662\\
60	0.0250946446357208\\
70	0.0235801027067453\\
80	0.0221739943887466\\
90	0.0213819665982544\\
100	0.0205442377870455\\
110	0.0201547200424036\\
120	0.0197038374726011\\
130	0.0191981953317144\\
140	0.019205270643356\\
150	0.0185338893600136\\
160	0.0183670126177326\\
170	0.0184592010713635\\
180	0.0181046118909735\\
190	0.0180994697307693\\
200	0.0179921477559593\\
210	0.0180326660117816\\
220	0.0178506211056282\\
230	0.0177410506828396\\
240	0.0178866479206081\\
250	0.0179318102130554\\
};
\addlegendentry{Sequential Classifier}

\addplot [color=red, densely dotted,line width=1.0pt]
  table[row sep=crcr]{%
20	0.0200670695462151\\
30	0.0200670695462151\\
40	0.0200670695462151\\
50	0.0200670695462151\\
60	0.0200670695462151\\
70	0.0200670695462151\\
80	0.0200670695462151\\
90	0.0200670695462151\\
100	0.0200670695462151\\
110	0.0200670695462151\\
120	0.0200670695462151\\
130	0.0200670695462151\\
140	0.0200670695462151\\
150	0.0200670695462151\\
160	0.0200670695462151\\
170	0.0200670695462151\\
180	0.0200670695462151\\
190	0.0200670695462151\\
200	0.0200670695462151\\
210	0.0200670695462151\\
220	0.0200670695462151\\
230	0.0200670695462151\\
240	0.0200670695462151\\
250	0.0200670695462151\\
};
\addlegendentry{$D(P_1\|P_0)$}

\addplot [color=black, densely dotted, line width=1.0pt]
  table[row sep=crcr]{%
20	0.0182528707194597\\
30	0.0182528707194597\\
40	0.0182528707194597\\
50	0.0182528707194597\\
60	0.0182528707194597\\
70	0.0182528707194597\\
80	0.0182528707194597\\
90	0.0182528707194597\\
100	0.0182528707194597\\
110	0.0182528707194597\\
120	0.0182528707194597\\
130	0.0182528707194597\\
140	0.0182528707194597\\
150	0.0182528707194597\\
160	0.0182528707194597\\
170	0.0182528707194597\\
180	0.0182528707194597\\
190	0.0182528707194597\\
200	0.0182528707194597\\
210	0.0182528707194597\\
220	0.0182528707194597\\
230	0.0182528707194597\\
240	0.0182528707194597\\
250	0.0182528707194597\\
};
\addlegendentry{$D_{\frac{\alpha}{1+\alpha}}(P_1\|P_0)$}

\end{axis}

\end{tikzpicture}%
%
%
\definecolor{mycolor1}{rgb}{0.00000,0.44700,0.74100}%
\definecolor{mycolor2}{rgb}{0.85000,0.32500,0.09800}%
\definecolor{mycolor3}{rgb}{0.92900,0.69400,0.12500}%
\definecolor{mycolor4}{rgb}{0.49400,0.18400,0.55600}%
\definecolor{mycolor5}{rgb}{0.46600,0.67400,0.18800}%
\begin{tikzpicture}

\begin{axis}[%
width=3.9096in,
height=2.905in,
at={(1.011in,0.651in)},
scale only axis,
xmin=20,
xmax=700,
xlabel style={font=\color{white!15!black}},
xlabel={$n$},
ymin=0.015,
ymax=0.055,
ylabel style={font=\color{white!15!black}},
ylabel={$-\frac{1}{n} \log{\epsilon_1}$},
axis background/.style={fill=white},
xmajorgrids,
ymajorgrids,
legend style={legend cell align=left, align=left, draw=white!15!black}
]
\addplot [color=blue, line width=1.0pt, mark=*, mark options={solid, blue}]
  table[row sep=crcr]{%
20	0.0535741810639929\\
50	0.0341189762562989\\
100	0.026159269122264\\
150	0.0230517848876703\\
200	0.0226735586116977\\
250	0.0218622131436958\\
300	0.021719043970687\\
350	0.0213225176046992\\
400	0.0211238363573561\\
450	0.0208025850260239\\
500	0.0202532622077007\\
550	0.0209325917544913\\
600	0.0202151860073234\\
650	0.0205315491211054\\
700	0.020067\\
};
\addlegendentry{Sequential Classifier}

\addplot [color=red, densely dotted, line width=1.0pt]
  table[row sep=crcr]{%
20	0.0182528707194597\\
50	0.0182528707194597\\
100	0.0182528707194597\\
150	0.0182528707194597\\
200	0.0182528707194597\\
250	0.0182528707194597\\
300	0.0182528707194597\\
350	0.0182528707194597\\
400	0.0182528707194597\\
450	0.0182528707194597\\
500	0.0182528707194597\\
550	0.0182528707194597\\
600	0.0182528707194597\\
650	0.0182528707194597\\
700	0.0182528707194597\\
};
\addlegendentry{$D(P_0\|P_1)$}

\addplot [color=black, densely dotted,line width=1.0pt]
  table[row sep=crcr]{%
20	0.0200670695462151\\
50	0.0200670695462151\\
100	0.0200670695462151\\
150	0.0200670695462151\\
200	0.0200670695462151\\
250	0.0200670695462151\\
300	0.0200670695462151\\
350	0.0200670695462151\\
400	0.0200670695462151\\
450	0.0200670695462151\\
500	0.0200670695462151\\
550	0.0200670695462151\\
600	0.0200670695462151\\
650	0.0200670695462151\\
700	0.0200670695462151\\
};
\addlegendentry{$D_{\frac{\alpha}{1+\alpha}}(P_0\|P_1)$}

\end{axis}
\end{tikzpicture}%
		\caption{Type-\RNum{1} and type-\RNum{2}  error exponents for the proposed sequential classifier. } 
		\label{fig:e0seq}
\end{figure}

\begin{theorem}\label{thm:converse}
 For every sequential classifier $\Phi^{\rm seq}=(\phi^{\rm seq}(\Tx,\TX), \tau)$, such that  $\mathbb{E}_{P_0}[\tau] \leq n,  \mathbb{E}_{P_1}[\tau]\leq n$, we have   
 \begin{equation}
\max_{\Phi^{\rm seq}} \min_{P_1\in \Pc(\Xc)}  E_0(\Phi^{\rm seq})E_1(\Phi^{\rm seq})\leq  D(P_0\|P_1)D_{\frac{\alpha}{1+\alpha}}(P_1\|P_0).
 \end{equation}
\end{theorem}

This suggests that the proposed sequential classifier is universal in the sense that it achieves the highest error exponent tradeoff over all possible classifiers and distributions $P_1$.

\begin{proof}
The proof can be found in Appendix \ref{sec:proofth52}.
\end{proof}

\newpage

\appendices
\section{Proof of Theorem \ref{thm:E0}}
\label{sec:proofth1}
First, we prove that for every realization of the training sequence $\TX$, the type-\RNum{1} error exponent is equal to $E_0$. Note that the solution to the optimization problem in \eqref{eq:threshcls} is a convex problem, and by the Karush-Kuhn-Tucker (KKT) conditions \cite{Boyd}, the minimizer is unique and  is the tilted distribution of $P_0$ and $\TX'$, i.e.,
\begin{equation}\label{eq:tiltedmu}
Q_{\mu^*}(x)= \frac{ P_{0}^{\frac{\mu^*}{1+\mu^*}}(x) {\TX'}^{\frac{1}{1+\mu^*}}(x) } {\sum_{a \in \Xc }  P_{0}^{\frac{\mu^*}{1+\mu^*}}(a) {\TX'}^{\frac{1}{1+\mu^*}}(a) },
\end{equation}
 and $\mu^*$ is the solution to
\begin{equation}\label{eq:KKTgamma} 
D(Q_{\mu^*}\| P_0)= E_0.
\end{equation}	
Therefore, the classifier threshold in \eqref{eq:threshcls} can be written as 
\begin{equation}\label{eq:threshsol} 
\gamma(E_0,\TX')=\beta D(Q_{\mu^*}\| \TX')-D(Q_{\mu^*}\| P_0).
\end{equation}
Now by Sanov's theorem, we can find the type-\RNum{1} error exponent  by solving \eqref{eq:min1} when $P_1$ replaced by $\TX'$ and $\gamma$ is replaced by \eqref{eq:threshsol}. This optimization problem is convex in $Q$ for every $\TX'$ when $\beta=1$, however for $\beta<1$ this is not the case and the KKT conditions are only necessary conditions. Spelling out the Lagrangian we obtain the KKT conditions 
\begin{align}
L(Q,\lambda,\nu)&= D(Q\|P_0)+ \lambda \big(\beta D(Q\|\TX')- D(Q\|P_0)-\gamma(E_0,\TX')\big) +\nu \Big(\sum_{x\in\Xc} Q(x)-1\Big)\\
\frac{\partial L(Q,\lambda,\nu)}{\partial Q(x)}&= 1+\log \frac{Q(x)}{P_0(x)} +\lambda \Bigg(\beta+\beta \log \frac{Q(x)}{\TX'(x)}-1- \log \frac{Q(x)}{P_0(x)}\Bigg)+\nu.
\end{align}
Setting the derivative to zero, we get
\begin{align}\label{eq:tiltedlambda}
Q_{\beta , \lambda^*}  (x)= \frac{ P_{0}^{\frac{1-\lambda^*}{1-\lambda^*+\lambda^* \beta}}(x) {\TX'}^{\frac{\lambda^*\beta}{1-\lambda^*+\lambda^* \beta}}(x) } {\sum_{a \in \Xc }  P_{0}^{\frac{1-\lambda^*}{1-\lambda^*+\lambda^* \beta}}(a) {\TX'}^{\frac{\lambda^*\beta}{1-\lambda^*+\lambda^* \beta}}(a) }, ~~~0\leq\lambda^* ,
\end{align}
and by complementary slackness condition \cite{Boyd}
\begin{equation}
D(Q_{\beta , \lambda^*}\| P_0)-\beta D(Q_{\beta, \lambda^*} \| \TX') = D(Q_{\mu^*}\| P_0)-\beta D(Q_{\mu^*}\| \TX').
\end{equation}
Note that, $\lambda^*$ cannot be zero as that sets $Q_{\beta , \lambda^*}= P_0(x)$ which is invalid solution when $E_0>0$. It is easy to see that this equality is satisfied when $\frac{1-\lambda^*}{1-\lambda^*+\lambda^* \beta} =\frac{\mu^*}{1+\mu^*} $. It is also easy to see this is the unique solution as $\frac{1-\lambda^*}{1-\lambda^*+\lambda^* \beta} $ is strictly decreasing function in $\lambda^*$ and hence $D(Q_{\beta , \lambda^*}\|P_0)$ and $D(Q_{\beta , \lambda^*}\|\TX')$  are strictly increasing and strictly decreasing in $\lambda^*$ respectively.   Therefore, we get 
\begin{equation}\label{eq:KKTgammaLRT} 
D(Q_{\beta, \lambda^*}\| P_0)= E_0.
\end{equation}	
Next, by rewriting the type-\RNum{1} error probability as a function of the types of the training and observation sequences $\Tx,\TX$ we have 
\begin{align}\label{eq:e0saddle}
\epsilon_0(\cls)= \sum_{\TX\in \Pc_k(\Xc)} P_1(\TX ) \sum_{\substack{ \Tx \in \Pc_n(\Xc), \\ \cls(\Tx,\TX) =1}} P_0(\Tx),
\end{align}
that is the type-\RNum{1} error exponent is equal to $E_0$ for every realization of training sequence $\TX$. As the hypothesis test is a type base test, and by letting $P_0(\mathcal{T}_{Q}^n) = n^{\frac{-|\Xc|+1}{2}} e^{-nD(Q\|P_0)}(c+o(1))$ \cite{Csiszarinfo}, we obtain
\begin{equation}\label{eq:prefactore0}
\epsilon_0(\cls)= \sum_{\TX \in \Pc_k(\Xc)} P_1(\TX) \sum_{ \Tx \in \Ao}n^{\frac{-|\Xc|+1}{2}} e^{-nD(\Tx \|P_0)}(c+o(1)).
\end{equation}
To find the polynomial decay of the error probability, we use the following theorems by \cite{Reeds, Iltis, Ney},  to approximate the summation by an integral and then use the saddle point approximation. We also use the shorthand notation $a_n \asymp b_n$ for any two positive real sequences such that $\log \frac{a_n}{b_n} = \Oc(1)$.
	\begin{theorem}\label{thm:reeds} \cite{Reeds}
		Suppose $\psi: \mathbb{R}^d \rightarrow \mathbb{R}$ is a  Lipschitz continuous function and the open set $\Dc \subseteq \mathbb{R}^d  $ has minimally smooth boundary. Then
			\begin{equation}
			\sum_{\wv \in \Dc \cap \Pc_k(\Xc) }  e^{ -k\psi(\wv) } \asymp k^{d} \int_{\Dc} e^{ -k\psi(\wv) } \ d\wv.
			\end{equation} 	
	\end{theorem}
	The precise definition of minimally smooth boundary can be found in \cite{Reeds}; we only work with smooth boundaries which are also minimally smooth.
	
Letting $\TX'$ be fixed, it can be shown that $\phi(Q)=D(Q\| P_0)-\beta D(Q \| \TX') $ is a smooth function, and hence we can conclude the boundary of the decision region $\Ac_0(\TX, \beta)$ is smooth since the graph of a smooth function is a smooth manifold \cite{Lee}. Furthermore, by the continuous differentiability of $\psi(Q)=D(Q\|P_0)$ in $Q$, the conditions of the Theorem \ref{thm:reeds} are satisfied and we can approximate the summation in \eqref{eq:prefactore0} by the integral 
	\begin{equation}\label{eq:reeds}
	\sum_{ \wv \in \Ao \cap \Pc_n(\Xc)  } e^{-nD(\wv \|P_0)} \asymp n^{|\Xc|-1}   \int_{\Ao  }   e^{-nD(\wv \|P_0)} d\wv.
	\end{equation}
	 Next we will use the following  theorem which can be derived from  \cite{Iltis} to approximate the integral.
	\begin{theorem}\cite{Iltis}
	 For a set $\Gamma$ in $\mathbb{R}^d$ having smooth boundary with a unique dominating point $\nu$ where the large deviation rate function $I(x)$ is minimized with $I(\nu)=a$, let the level surface $S_a=\{x\in \mathbb{R}^d : I(x)=a \}$ be described locally by $x_d=f(x)$, where $x = (x_1,, x_{d-1})$, and let the region $\Gamma$ be bounded by the surface $x_d=g(x)$, where $g(x)$ is a three times differentiable function of $x$ in a neighbourhood of $0$. If the Hessians of $g$ and $f$ at $0$ satisfy $H_g(0)>H_s(0 )$ then
	\begin{align}
	 \int_{\wv \in \Gamma}   e^{-nD(\wv \|P_0)} d \wv =n^{-\frac{d+1}{2}} e^{-na} (d_0 + o( 1)), \label{eq:intsaddle}
	\end{align}
	where 
	\end{theorem}
	Since the optimizing distribution solving Sanov's Theorem is unique and equal to $Q_{\lambda^*}$, the dominating point is unique. Hence,  we only need to prove the the Hessian of the decision region $\Dc$ minus the Hessian of the level surface $\Sc_{E_0}=\{Q \in \Pc(\Xc) :D(Q\|P_0)\leq E_0 \} $ is positive definite matrix. Writing $Q_{\beta, \lambda^*}(|\Xc|) = 1- \sum_{x=1}^{|\Xc|-1} Q_{\beta , \lambda^*}(x)$ and taking the derivatives we have:
	\begin{align}
	\Hm_{\Dc}^{(|\Xc|-1) \times (|\Xc|-1)}&= -(1-\beta) \Bigg ( \diag \Big (\frac{1}{Q_{\beta , \lambda^*}} \Big )+\frac{ \onev \onev^T}{Q_{\beta , \lambda^*}(|\Xc|)} \Bigg )  ,\\
	\Hm_{\Sc_{E_0}}^{(|\Xc|-1) \times (|\Xc|-1)}&=- \diag \Bigg (\frac{1}{Q_{\beta , \lambda^*}} \Bigg )-\frac{ \onev \onev^T}{Q_{\beta , \lambda^*}(|\Xc|)} ,
	\end{align}
	and hence
	\begin{equation}
	\Hm_{\Sc_{E_0}}-\Hm_{\Dc}=\beta   \diag \Bigg (\frac{1}{Q_{\beta , \lambda^*}} \Bigg ) + \beta \frac{ \onev \onev^T}{Q_{\beta , \lambda^*}(|\Xc|)} ,
	\end{equation}
	which is positive for any $\beta >0$. Therefore, by  \eqref{eq:prefactore0}, \eqref{eq:reeds},  \eqref{eq:intsaddle},  
	we have  
	\begin{align}
	\epsilon_0(\cls)&= \sum_{\TX} P_1(\TX) n^{\frac{-|\Xc|+1}{2}}  n^{|\Xc|-1} n^{-\frac{|\Xc|}{2}} e^{-nE_0}(c+o(1))\\
	&=  \sum_{\TX} P_1(\TX) \frac{1}{\sqrt{n}} e^{-nE_0}(c+o(1)).
	\end{align}
Finally, since $c$ is finite, positive, and only depends on $P_0, P_1, E_0$, there exists a $\tilde c$ such that
\begin{align}
\epsilon_0(\cls)= &\frac{1}{\sqrt{n}} e^{-nE_0} (\tilde c+ o(1)\big),
\end{align}
which concludes the proof.


\section{Proof of Theorem \ref{thm:E1}}
\label{sec:proofth2}
Using standard properties of the method of types \cite{Csiszarinfo}, the type-\RNum{2} probability of error can be written as
\begin{align}
\epsilon_1 (\cls)&=  \sum_{\xv, \X: \cls(\Tx,\TX)=0  } P_1(\X) P_1(\xv)\\
&=  \sum_{\substack{Q\in \Pc_n(\Xc) , Q_1 \in \Pc_k(\Xc),\\  \cls(Q,Q_1)=0  }} P_1(\mathcal{T}_{Q}^n) P_1(\mathcal{T}_{Q_1}^k) \label{eq:method}\\ 
&\leq \sum_{\substack{Q\in \Pc_n(\Xc) , Q_1 \in \Pc_k(\Xc),\\  \cls(Q,Q_1)=0  }} e^{-nD(Q\|P_1)-kD(Q_1\|P_1)}\\
&\leq (n+1)^{|\Xc|}(k+1)^{|\Xc|} 2^{-n 	\tilde{E}_1^{(n)}(\cls)   }
\end{align}
where 
\begin{equation}\label{eq:exponent}
	\tilde{E}_1^{(n)}(\cls)= \min_{ \substack{ \cls(Q,Q_1)=0 \\ Q,Q_1 \in \Pc(\Xc) } }D(Q\|P_1)+\alpha D(Q_1\|P_1).
\end{equation}
It is worth observing that this optimization problem is non-convex.
In addition, lower bounding \eqref{eq:method} yields
\begin{align}
\epsilon_1 (\cls) \geq (n+1)^{-|\Xc|} (k+1)^{-|\Xc|} e^{-n\tilde{E}_1^{(n)}(\cls) }.
\end{align}
Observe that the proposed test $\cls$ depends on  $n$; hence if the limit of \eqref{eq:exponent} as $n$ goes to infinity exists, then the limit is the type-\RNum{2} error exponent. Using a change of variable, we can write \eqref{eq:exponent} as
\begin{equation}
\begin{aligned}
\tilde{E}_1^{(n)}(\cls)=\min_{Q,Q_1} ~ & D(Q\|P_1)+\alpha D\Bigg(\frac{1}{1-\delta_n} Q_1 - \frac{\delta_n}{1-\delta_n} U \Big \|P_1\Bigg) \\
\textrm{s.t.}  ~& \beta D(Q\|Q_1)  -D(Q\|P_0)  \geq  \gamma(E_0,Q_1)\\
&Q \in \Pc(\Xc)    \\
& Q_1 \in (1-\delta_n)\Pc(\Xc) + \delta_n U 
\end{aligned}
\end{equation}
where $U$ is the uniform distribution on $\Pc(\Xc)$. Now, by using a Taylor expansion of $D\big(\frac{1}{1-\delta_n} Q_1 - \frac{\delta_n}{1-\delta_n} U  \|P_1\big)$ around $\delta_n=0$, we get
\begin{equation}
\begin{aligned}
\tilde{E}_1^{(n)}(\cls)=\min_{Q,Q_1} \quad & D(Q\|P_1)+\alpha D(Q_1 \|P_1) + \Oc\Big(\frac{\delta_n}{1-\delta_n}\Big) \\
\textrm{s.t.} \quad &\beta D(Q\|Q_1)  -D(Q\|P_0)  \geq  \gamma(E_0,Q_1)\\
&Q \in \Pc(\Xc)    \\
& Q_1 \in (1-\delta_n)\Pc(\Xc) + \delta_n U 
\end{aligned}
\end{equation}
Therefore, we can approximate the type-\RNum{2} error exponent by
\begin{equation}\label{eq:approxcontin}
E^{(n)}_1(\cls) = \min_{\substack{\beta D(Q\|Q_1)  -D(Q\|P_0)  \geq  \gamma(E_0,Q_1) \\ Q \in \Pc(\Xc), Q_1 \in \Pc_{\delta_n}(\Xc)}} D(Q\|P_1) +\alpha  D(Q_1\|P_1)
\end{equation}
where 
\begin{equation}
 \Pc_{\delta_n}(\Xc) = \big\{Q: Q=(1-\delta_n)P + \delta_n  U,   P \in \Pc(\Xc)    \big\},
\end{equation}
and the approximation error is of order $\Oc\big(\frac{\delta_n}{1-\delta_n}\big)$. Letting $\delta_n=o(n^{-1})$ we have
\begin{equation}
\epsilon_1 (\cls) = e^{-n{E}_1^{(n)}(\cls)+o(1)} \asymp e^{-n{E}_1^{(n)}(\cls)}.
\end{equation}
Hence, the type-\RNum{2} error exponent can be calculated as the limit of \eqref{eq:approxcontin} when $n$ goes to infinity.  We also define following optimization problems. For every $Q_1 \in \Pc_{\delta_n}(\Xc)$, let 
\begin{equation}\label{eq:firstopt}
E_1(E_0,Q_1)=\min_{\substack{\beta D(Q\|Q_1)  -D(Q\|P_0)  \geq  \gamma(E_0,Q_1)   \\ Q \in \Pc(\Xc), }  } D(Q\|P_1) ,
\end{equation}
which is the error exponent when the type of the training sequence  is $Q_1$. Also, let
\begin{equation}\label{eq:secondopt}
E_1^{(n)}(E_0,r)=\min_{\substack{D(Q_1\|P_1)\leq r \\  Q_1 \in \Pc_{\delta_n}(\Xc) } } E_1(E_0,Q_1).
\end{equation}
The latter optimization problem is the worst case achievable error exponent by $\cls$ if we know the training sequence type $Q_1$ is inside a relative entropy ball around the original distribution $P_1$ of radius $r$ \cite{BoroumandTran} .  Using \eqref{eq:firstopt}, \eqref{eq:secondopt}, we can wirte \eqref{eq:approxcontin}  as
\begin{align}
E_1^{(n)} (\cls)&= \min_{\substack{\beta D(Q\|Q_1)  -D(Q\|P_0)  \geq  \gamma(E_0,Q_1), \\  D(Q_1\|P_1)=r~ , ~  r \geq 0 ,  \\ Q \in \Pc(\Xc) ~,~ Q_1 \in \Pc_{\delta_n}(\Xc)   }    } D(Q\|P_1) +\alpha r\\
&=\min_{\substack{ D(Q_1\|P_1) =  r ~,~  r \geq 0 , \\Q_1 \in \Pc_{\delta_n}(\Xc)} } E_1(E_0,Q_1)+\alpha r \\
&=\min_{\substack{D(Q_1\|P_1) \leq r ~,~  r \geq 0  \\ Q_1 \in \Pc_{\delta_n}(\Xc)} } E_1(E_0,Q_1)+\alpha r  \label{eq:concaveopt} \\
&=\min_{ r \geq 0 } E_1^{(n)}(E_0,r)+\alpha r, \label{eq:mismatch}
\end{align}
where in \eqref{eq:concaveopt} we used the KKT conditions to show that the minimizer $Q_1^*$ of  \eqref{eq:concaveopt} should satisfy the inequality condition $D(Q_1\|P_1)\leq r$ with equality, i.e.,  $D(Q^*_1\|P_1) =  r$. We will use  \eqref{eq:mismatch} to analyze the behaviour of the non-convex problem in \eqref{eq:exponent}. We also define $r_{c}$ as
\begin{equation}
r_c=\min_{\substack{Q_1\in \Pc(\Xc) \\ \min_{i\in\Xc} Q_1(i)=0}} D(Q_1\|P_1),
\end{equation}
that is the minimum distance between $P_1$ and distributions on the boundary of the probability simplex $\Pc(\Xc)$. Since $P_1(i)$ is non zero for every $i$, therefore $r_{c}>0$, and we can lower bound $E_1^{(n)}(\cls)$ by
\begin{align}
E_1^{(n)} (\cls)&\geq \min\Big \{ \min_{0\leq r \leq r_c-\epsilon} E_1^{(n)}(E_0,r)+\alpha r, \alpha(r_c-\epsilon) \Big \},
\end{align}
where $0<\epsilon<r_c$, and we used the fact that $E^{(n)}_1(E_0,r) \geq0$ for $r \geq 0$ by non-negativity of relative entropy.  For $r\in[0,r_c-\epsilon]$, we also define $E_1(E_0,r)$ as
\begin{align}
 E_1(E_0,r)  =& \lim_{n\rightarrow \infty} E_1^{(n)}(E_0,r)\\
 =&  \lim_{n\rightarrow \infty} \min_{\substack{D(Q_1\|P_1)\leq r \\  Q_1 \in \Pc_n(\Xc) } } E_1(E_0,Q_1)\\
 =& \min_{\substack{D(Q_1\|P_1)\leq r \\  Q_1 \in \Pc(\Xc) } } E_1(E_0,Q_1). \label{eq:limitE}
\end{align}
Next, we use the following lemma to find a lower bound on type-\RNum{2} error exponent.
\begin{lemma}\label{lem:cont}
Let  $r \in[0,r_c)$,  then for every $n \in \mathbb{Z}^{+}$, the exponent functions  $E^{(n)}_1(E_0,r)$, and $E_1(E_0,r)$ are continuous in $r$.	
\end{lemma}	
The proof of Lemma \ref{lem:cont} relies on Berge's maximum theorem \cite{Berge} and we defer it to Appendix  \ref{app:cont}. Taking the limit of $E_1^{(n)} (\cls)$ when $n$ goes to infinity, we have
\begin{align}
E_1 (\cls)&= \liminf_{n \rightarrow \infty} E_1^{(n)} (\cls) \nonumber\\
 &\geq\liminf_{n \rightarrow \infty}  \min\Big \{ \min_{0\leq r \leq r_c-\epsilon} E_1^{(n)}(E_0,r)+\alpha r, \alpha(r_c-\epsilon) \Big \}\\
 &= \min\Big \{ \min_{0\leq r \leq r_c-\epsilon}  E_1(E_0,r)+\alpha r, \alpha(r_c-\epsilon) \Big \} \label{eq:lim},
\end{align}
In order to show \eqref{eq:lim}, by Lemma \ref{lem:cont}, $E^{(n)}_1(E_0,r)$, and  $E_1(E_0,r)$ are continuous on the compact interval $[0,r_c-\epsilon]$. Also, since $E_1^{(n)}(E_0,r) \geq E_1^{(n+1)}(E_0,r)$ for every $r$ in the domain, by Dini's theorem the convergence of $E^{(n)}_1(E_0,r)$ is uniform \cite{Rudin}. Finally, since the convergence is uniform we can interchange the minimum and limit to get  \eqref{eq:lim}. 

When $r=0$, we have $Q_1=P_1$ and $\gamma(E_0,P_1)=\beta E_1^*(E_0)-E_0$, resulting  $E_1(E_0,r=0)=E_1^*(E_0)$ where $E_1^*(E_0)$ is the optimal error exponent tradeoff in \eqref{eq:tradefix}. Moreover, for $r_1\leq r_2$, we have $\mathcal{B}(P_1,r_1) \subseteq \mathcal{B}(P_1,r_2)$ where $\mathcal{B}(P,r)=\{ Q \in \Pc(\Xc): D(Q\|P) \leq r\}$. Hence $E_1(E_0,r)$ is a non-increasing function of $r$. Therefore, the sufficient condition to have $E_1(\cls)=E_1^*(E_0)$ is
  \begin{equation}\label{eq:cond}
\min\Big \{ \min_{0\leq r \leq r_c-\epsilon}  E_1(E_0,r)+\alpha r, \alpha(r_c-\epsilon) \Big \} =E_1(E_0,r=0).
  \end{equation}
Equivalently, letting $\epsilon$ to go to zero, \eqref{eq:cond} can be written as the following two sufficient conditions 
\begin{align}
E_1(E_0,r)+\alpha r &\geq E_1(E_0,r=0) ~~ \forall r: ~ 0 \leq r < r_c,\\
\alpha &> \frac{E_1(E_0,r=0)}{r_c}.
\end{align}
Hence,  letting
\begin{equation}\label{eq:optimalalpha}
\alpha^* _\beta=  \max \bigg \{\sup_{0\leq r < r_c} \frac{E_1(E_0,r=0)-E_1(E_0,r)}{r}, \frac{E_1^*(E_0)}{r_c} \bigg \}
\end{equation}
for any $\alpha > \alpha^*_\beta$, we have $E_1(\cls) =E_1^*(E_0)$. Finally, since $E_1(E_0,r)$ is non-increasing and continuous on $r\in [0,r_c)$, and $-\frac{\partial E_1(E_0,r)}{\partial r} \big|_{r=0} < \infty$ as it will be proved in Theorem \ref{thm:lower}, the supremum in \eqref{eq:optimalalpha} is finite and  
\begin{equation}
0<\alpha^*_\beta< \infty,
\end{equation}
 and this concludes the proof.

\section{Proof of Theorem \ref{thm:refinement}}
\label{sec:proofth3}

	Using the method of types, the type-\RNum{2} probability of error can be written as
	\begin{align}
	\epsilon_1 (\cls)&=  \sum_{\xv, \X: \cls(\Tx,\TX)=0  } P_1(\X) P_1(\xv)\\
	&=\sum_{Q_1 \in \Pc_k(\Xc)\cap \Pc_{\delta_n} (\Xc)   } P_1(\mathcal{T}_{Q_1}^k) \sum_{ \substack{Q \in \Pc_n(\Xc),\\  \cls(Q,Q_1) =0}} P_1(\mathcal{T}_{Q}^n) \\
	&\asymp \sum_{Q_1 \in \Pc_k(\Xc) \cap \Pc_{\delta_n} (\Xc)  } c k^{\frac{-|\Xc|+1}{2}} e^{-kD(Q_1\|P_1)} \sum_{ \substack{Q \in \Pc_n(\Xc),\\  \cls(Q,Q_1) =0}} P_1(\mathcal{T}_{Q}^n)\label{eq:errorprob1}
	\end{align}	
	where in the last step we used \cite{Csiszarinfo}
	\begin{equation}
	P_1(\mathcal{T}_{Q_1}^k) \asymp k^{\frac{-|\Xc|+1}{2}} e^{-kD(Q_1\|P_1)}.
	\end{equation}
	We need to show that conditioned on the training sequence having a type $Q_1$, the conditional type-\RNum{2} error probability is
	\begin{equation}\label{eq:mismatchprefac}
	\sum_{ Q: \cls(Q,Q_1) =0} P_1(\mathcal{T}_{Q}^n)= \frac{1}{\sqrt{n}} e^{-nE_1(E_0,Q_1)}    \big (c+ o(1)\big),
	\end{equation}
	when $E_1(E_0,Q_1)$ is positive, i.e,  the error decay has a prefactor of $\frac{1}{\sqrt{n}}$. We first show that for any classifier with $0 \leq \beta \leq 1$ we can upper bound the type-\RNum{2} error probability by the classifier type-\RNum{2} error probability with $\beta=1$. 
	\begin{lemma}\label{lem:inclusion}
	 For every fixed $Q_1$ and for every $\beta_1\leq \beta_2$ we have  $\Ac_0(Q_1,\beta_1) \subseteq \Ac_0(Q_1,\beta_2)$.
	\end{lemma}
	\begin{proof}
	 Let $Q \in \Ac_0(Q_1,\beta_1)$, we need to show that for every such $Q$ , we have $Q \in \Ac_0(Q_1,\beta_2)$. We then have
	\begin{align}
	D(Q\|P_0) - \beta_1 D(Q\|Q_1) \leq D(Q_{\mu^*}\|P_0) - \beta_1 D(Q_{\mu^*}\|Q_1)
	\end{align}
	where $Q_{\mu^*}$ is the solution to optimization in $\gamma(E_0,Q_1)$ in \eqref{eq:tiltedmu} and is independent of $\beta$. Furthermore,
	\begin{align}\label{eq:betainclusion}
	D(Q\|P_0) - \beta_2 D(Q\|Q_1)\leq D(Q_{\mu^*}\|P_0)- \beta_2 D(Q_{\mu^*}\|Q_1) + (\beta_2-\beta_1) \Big ( D(Q_{\mu^*}\|Q_1) -D(Q\|Q_1)\Big ).
	\end{align}	
	Similarly to the proof of Theorem \ref{thm:E0}, using the KKT conditions it can be shown that the projection of $Q_1$ to the set $\Ac(Q_1,\beta)$ also equals to $Q_{\mu^*}$, i.e.,
	\begin{equation}
	Q_{\mu^*}=\argmin_{D(Q\|P_0) - \beta_1 D(Q\|Q_1) \leq \gamma(E_0,Q_1)} D(Q\|Q_1).
	\end{equation}
	Hence for every $Q\in \Ac_0(Q_1,\beta_1)$ we have
	\begin{equation}
	D(Q_{\mu^*}\|Q_1) \leq D(Q\|Q_1), 
	\end{equation}
	and by \eqref{eq:betainclusion}
	\begin{align}
	D(Q\|P_0) - \beta_2 D(Q\|Q_1)\leq D(Q_{\mu^*}\|P_0)- \beta_2 D(Q_{\mu^*}\|Q_1), 
	\end{align}	
	i.e., $Q \in \Ac_0(Q_1,\beta_2)$ which concludes the proof.
	\end{proof}
	
	Therefore we can upper bound the error probability by setting $\beta=1$. Since for every fixed $Q_1$ the test is  the Neyman-Pearson test using the mismatched distribution $Q_1$, the decision region boundary is characterised by a hyperplane and by \cite{Reeds} the conditional type-\RNum{2} error probability prefactor equals to $\frac{1}{\sqrt{n}}$ for every $Q_1$. Hence, substituting \eqref{eq:mismatchprefac} into \eqref{eq:errorprob1} we have
	\begin{align}
	\epsilon_1 (\cls) \leq \sum_{Q_1 \in \Pc_k(\Xc) \cap \Pc_{\delta_n} (\Xc)  }  \frac{k^{\frac{-|\Xc|+1}{2}}}{\sqrt{n}} e^{-kD(Q_1\|P_1)} e^{-nE_1(E_0,Q_1)}    \big (c+ o(1)\big). \label{eq:errorprob1saddle}
	\end{align}	
	 Observe that one can also use the similar approach in the proof of Theorem \ref{thm:E0}, to show stronger result that for all $0<\beta \leq 1 $ the prefactor equals to $\frac{1}{\sqrt{n}}$. Also notice that by choosing some small enough $\delta_n>0$, and constraining $\TX \in \Pc_{\delta_n}(\Xc)$, the approximation error caused by substituting $\TX'$ by $\TX$ can be summarized in the $c+o(1)$ term in \eqref{eq:errorprob1saddle}, and we get
	\begin{align}
	\epsilon_1 (\cls)&\asymp  n^{-\frac{|\Xc|}{2}}  \sum_{Q_1 \in \Pc_k(\Xc) }  e^{ -k\psi(Q_1) }, \label{eq:sumsaddle}\\
	\psi(Q_1)&=\frac{1}{\alpha}E_1(E_0,Q_1) +  D(Q_1\|P_1),     
	\end{align}	
	where we used the fact that $\alpha=\frac{k}{n}$ is fixed and independent of $n, k$. 
	
	Next we approximate the summation by an integral using  Theorem \ref{thm:reeds}. As shown in Theorem  \ref{thm:E1}
	\begin{equation}
	P_1=\argmin_{Q_1: D(Q_1\|P_1) <r_c} \psi(Q_1)
	\end{equation}
	for $\alpha^* _\beta < \alpha$. We show that the terms in the summation \eqref{eq:sumsaddle} where $Q_1$ is in the vicinity of the $P_1$ dominate the summation. Let $\Dc=\Pc(\Xc)$ and $\Dc'=\Bc(P_1,\epsilon)$ where $ \epsilon = \frac{ s \log k}{k} $, $s>0$. We can split the summation \eqref{eq:sumsaddle} as
	\begin{align}
	\sum_{Q_1 \in \Dc \cap \Pc_k(\Xc) }  e^{ -k\psi(Q_1) }& =  \sum_{Q_1 \in  \Dc'  \cap \Pc_k(\Xc) }  e^{ -k\psi(Q_1) } +\sum_{Q_1 \in  (\Dc')^c  \cap \Pc_k(\Xc) }  e^{ -k\psi(Q_1) }. \label{eq:tworegions}
	\end{align}
	  By Theorem \ref{thm:lower}, $\psi(Q_1)$ achieves its minimum at $Q_1=P_1$ for $\alpha>\alpha^*_\beta$, i.e., $\psi(Q_1)> \psi(P_1)$ for $Q_1\neq P_1$. Hence, by the mean value theorem and using a Taylor series expansion of $\psi(Q)$ around $Q_1=P_1$ for $ Q_1 \in  (\Dc')^c$ we get
	\begin{equation}
	\psi(Q_1) = \psi(P_1) +  \frac{1}{2}\thetav_{P_1}^{T} \tilde{\Hm}(Q_1)  \thetav_{P_1} 
	\end{equation}
	where $\thetav_{P_1}$ is the difference vector $Q_1-P_1$ defined in \eqref{eq:diffQ1}, and $\tilde{\Hm}(Q_1)$ is a positive definite matrix for every $Q_1 \in  (\Dc')^c$. Also, we have used the result form Theorem \ref{thm:lower} that the first order term of the expansion is zero. Furthermore, note that by the condition $ Q_1 \in  (\Dc')^c$ we have $\| \thetav_{P_1} \|_\infty^2 \geq O(\epsilon)$. Hence for large enough $n$ we have
	\begin{equation}
	\sum_{Q_1 \in  (\Dc')^c  \cap \Pc_k(\Xc) }  e^{ -k\psi(Q_1)} \leq (k+1)^{|\Xc|} e^{-k\Big (\frac{1}{\alpha} E_1(E_0,P_1) +\frac{1}{2}  \Lambda \frac{s \log k}{k}\Big )}
	\end{equation}
	where  $\Lambda= \min_{Q_1\in(\Dc')^c } \Lambda_{\min}\Big(\tilde{\Hm}(Q_1)\Big )$, and $\Lambda_{\min}(.)$ is the smallest eigenvalue of a matrix. Note that $\Lambda$ is strictly positive since $\psi(Q_1)> \psi(P_1)$ for $Q_1\neq P_1$. Finally, by setting $s \geq \frac{2 } { \Lambda } (|\Xc|+\frac{3}{2})  $ we have
	\begin{equation}
	\sum_{Q_1 \in  (\Dc')^c  \cap \Pc_k(\Xc) }  e^{ -k\psi(Q_1)} \leq n^{-\frac{1}{2}} e^{-n E_1(E_0,P_1) } O(n^{-1}).
	\end{equation}
	
	 It remains to bound the first term in \eqref{eq:tworegions}. By the envelope theorem \cite{Segal} as shown in Lemma \ref{lem:HessianE}, we can conclude that $\psi(\wv): \mathbb{R}^{|\Xc|-1} \rightarrow \mathbb{R}$  is continuously differentiable and hence Lipschitz on $\Dc'$.  Furthermore, $\Dc'$ is a $C^{\infty}$ manifold, and hence $\Dc'$ has a minimally smooth boundary. Therefore, by Theorem \ref{thm:reeds} we have
	\begin{align}
	\epsilon_1 (\cls)&\asymp  n^{\frac{|\Xc|-2}{2}}    \int_{\Dc'} e^{ -k\psi(\wv) } \ d\wv. 
	\end{align}	
	By using a Taylor series expansion of $\psi(\wv)$ around $\wv=P_1$, and Lemma \ref{lem:HessianE} we have
	\begin{align}
	\epsilon_1 (\cls)&\asymp  n^{\frac{|\Xc|-2}{2}}  \int_{\Dc'}   e^{ -n \big(E_1(E_0,P_1) + \frac{1}{2}<(\wv-P_1),(\Hm_2+\alpha \Jm) (\wv-P_1)> +O\big (\epsilon^{\frac{3}{2}} \big) \big )} d(\wv-P_1).
	\end{align}	
	Note that $\wv$ is bounded to the the $|\Xc|-1$ dimensional  probability simplex as $\wv \in \Dc'$. Hence for every $\wv-P_1$ there exists a $|\Xc| \times |\Xc|$ rotation matrix $\Rm$ such that for every $\wv \in \Pc(\Xc) $ we have $\ev_{|\Xc|}^T \Rm (\wv-P_1) =0$, where $\ev_{|\Xc|}$ is the unit vector with $1$ in the $|\Xc|$'th coordinate and $0$ otherwise. Letting $\wv' = \Rm(\wv-P_1)$ we have
	\begin{align}
	\epsilon_1 (\cls)&\asymp  n^{\frac{|\Xc|-2}{2}}  \int_{\Dc' \cap \R^{|\Xc|-1}}   e^{ -n \big(E_1(E_0,P_1) + \frac{1}{2 }\big<\Rm^T \wv',(\Hm_2+\alpha \Jm)\Rm^T \wv'\big> +O\big (\epsilon^{\frac{3}{2}} \big)    \big )}  d \Rm^T\wv',
	\end{align}	
	where we have used $\Rm \Rm^T=\Id$. Furthermore, from the proof of Theorem \ref{thm:lower}, $\Hm_2+\alpha \Jm$ is strictly positive definite and hence full rank for $\alpha^*_\beta < \alpha$. Using a change of variables $\yv=\sqrt{n}\wv'$ we have $d\yv=n^{-\frac{|\Xc|-1}{2}} d\wv'$, and
	\begin{align}
	\epsilon_1 (\cls)&\asymp  n^{-\frac{1}{2}} e^{ -n E_1(E_0,P_1)}   \int_{\Dc'}   e^{ - \frac{1}{2}\big<\Rm^T \yv,(\Hm_2+\alpha \Jm)\Rm^T \yv\big> +O\big (\sqrt{\frac{\log^3{n}}{n}}\big )  }d\yv,
	\end{align}	
	where the integral is a Gaussian integral  and equals to $c'+o(1)$ for a constant $c'$ and this concludes the proof.

\section{Proof of Thereom \ref{thm:lower}}
\label{sec:proofth5}

From \eqref{eq:optimalalpha} we have

\begin{equation}
	\underline{\alpha}_\beta\triangleq- \frac{\partial E_1(E_0,r)}{\partial r} \bigg|_{r=0} \leq \alpha^*_\beta.
\end{equation}
To find the derivative of $E_1(E_0,r)$ at $r=0$, we use a Taylor series expansion of $E_1(E_0,Q_1)$ around $Q_1=P_1$. We use the following lemmas in our proof. First lemma provides the first and second derivative of $\gamma (E_0,Q_1)$.
\begin{lemma}\label{lem:gamsen}
	For every $P_0, P_1 \in \Pc(\Xc)$, the gradient of threshold function $\gamma(E_0,Q_1)$ is
	\begin{align}\label{eq:firstdergam}
	\frac{\partial \gamma (E_0,Q_1)}{\partial Q_1 (i)}  = -\beta \frac{Q_{\mu^*}(i)}{Q_1(i)}, 
	\end{align}
	 Furthermore the Hessian matrix of $\gamma(E_0,Q_1)$ evaluated at $Q_1=P_1$ is
	\begin{align}
	\Hm_1 &= \Hm(Q_1) \Big |_{Q_1=P_1}, \\
	\Hm(Q_1)	  &= \nabla_{Q_1}^2  \gamma(E_0,Q)  \Big |_{Q_1=P_1} \\
	&= 	\beta \Jm\Bigg (\frac{\mu^*}{1+\mu^*}  	\Tm+  \frac{1}{(1+\mu^*)} (\Qm+ \Wm )  \Bigg ) \Jm,  \label{eq:Hessgam}
	\end{align}
	where $\Wm=\wv \wv^T, \Qm=\qv \qv^T $ and
	\begin{align}
	  \wv^T= \frac{1}{{\sqrt{ \Var_{Q_{\mu^*}} \Big (\log \frac{Q_{\mu^*}}{P_0} \Big ) }}} \Bigg ( Q_{\mu^*}(1) {\log\Big ( \frac{Q_{\mu^*}(1)}{P_0(1)} \Big )-E_0 }, \ldots,   Q_{\mu^*}(|\Xc|){\log\Big ( \frac{Q_{\mu^*}(|\Xc|)}{P_0(|\Xc|)} \Big )-E_0 } \Bigg),
	\end{align} 
	\begin{align}
	\qv= \Big (Q_{\mu^*} (1),\ldots, Q_{\mu^*}(|\Xc|)\Big),
	\end{align}
	 \begin{align}
	\Tm= \diag \Big(Q_{\mu^*}  \Big), \quad \Jm=\diag \Bigg(\frac{1}{Q_1}  \Bigg),
	\end{align}
and $Q_{\mu^*}, \mu^*$ are defined as in \eqref{eq:tiltedmu} and \eqref{eq:KKTgamma}.
\end{lemma}	

\begin{proof}
The proof can be found in Appendix \ref{sec:prooflemd1}.
\end{proof}

Next using the derivatives of $\gamma (E_0,Q_1)$, and \eqref{eq:firstopt} we can find the derivatives of the exponent function $E_1(E_0,Q_1)$.
\begin{lemma}\label{lem:HessianE}
	For every $P_0, P_1\in \Pc(\Xc)$, the gradient  of the  exponent function $E_1(E_0,Q_1)$  evaluated at $Q_1=P_1$ is
	\begin{align}\label{eq:zeroder}
	\frac{\partial E_1(E_0,Q_1)}{\partial Q_1(i)} \Bigg|_{Q_1=P_1}   = 0.
	\end{align}
	Furthermore, the Hessian matrix  evaluated at $Q_1=P_1$ is
	\begin{align}\label{eq:Hessian}
	\Hm_2 &\triangleq \nabla^2_{Q_1} E_1(E_0,Q_1) \Big |_{Q_1=P_1}\\
	&= \beta \eta^*_\beta \Jm \Big [  \Qm + \eta^*_1 \Vm +(1-\eta^*_1)\Wm    - \Tm    \Big ] \Jm
	\end{align}
where $\eta^*_\beta$ and $\eta^*_1$ are the optimal Lagrange multipliers in \eqref{eq:lagrangeE} when $0<\beta\leq 1$ and $\beta=1$ respectively, and $\Vm=\vv \vv^T$, 
\begin{align}
\vv=\frac{1}{\sqrt{\Var{Q_{\eta^*_\beta} }(\Omega)   }}    \Big ( Q_{\eta^*_\beta}(1)\big ( \Omega(1) -(E_0- \beta E_1)\big ),\ldots,Q_{\eta^*_\beta}(|\Xc|)\big (\Omega(|\Xc|)-(E_0- \beta E_1)\big)   \Big ),
\end{align}
\begin{align}
\Omega(i)=\beta \log \frac{P_1(i)}{P_0(i)} +(1-\beta)  \log \frac{Q_{\eta^*_\beta}(i)}{P_0(i)}, i \in \Xc.
\end{align}

\end{lemma}	

\begin{proof}
The proof can be found in Appendix \ref{sec:prooflemd2}.
\end{proof}

Lemma \ref{lem:HessianE} shows that the gradient of $E_1(E_0,Q_1)$ respect to $Q_1$ is zero at $Q_1=P_1$ which is due to the choice of the threshold function $\gamma(E_0,Q_1)$. Note that in the case of the plugin classifier with the fixed threshold the gradient is non zero which results in $\frac{\partial E_1(E_0,r)}{\partial r} \bigg|_{r=0}= \infty$. Next by using a Taylor series expansion of $E_1(E_0,Q_1)$  and the relative entropy in \eqref{eq:limitE} we can get the following lemma.
\begin{lemma}\label{lem:approx}
	For any $P_0, P_1 \in \Pc(\Xc)$, the approximation 
	\begin{align}\label{eq:worstapproxopt}
	E_1(E_0,r) =  E^*_1(E_0)+   \min_{\substack{ \frac{1}{2} \thetav_{P_1}^T \Jm \thetav_{P_1}  \leq r \\ \onev^T\thetav_{P_1}=0  } }  \frac{1}{2}\thetav_{P_1}^{T} \Hm_2 \thetav_{P_1} +     o(r),
	\end{align}	
	 where 
	\begin{align}\label{eq:diffQ1}
		\thetav_{P_1}&= \Big(Q_1(1)-P_1(1),\dotsc,Q_1(|\Xc|)-P_1(|\Xc|)\Big)^T,
	\end{align}
	holds.
\end{lemma}

\begin{proof}
The proof can be found in Appendix \ref{sec:prooflemd3}.
\end{proof}	

We can further simplify the convex optimization problem in \eqref{eq:worstapproxopt}, using the following lemma.		
\begin{lemma}\label{lem:optwo}
For $\Jm, \Hm_2$ defined in \eqref{eq:fisher}, \eqref{eq:Hessian}, we have
\begin{equation}\label{eq:simp}
 \min_{\substack{ \frac{1}{2} \thetav_{P_1}^T \Jm \thetav_{P_1}  \leq r \\ \onev^T\thetav_{P_1}=0  } }  \frac{1}{2}\thetav_{P_1}^{T} \Hm_2 \thetav_{P_1} = \min_{\substack{ \frac{1}{2} \psiv_{P_1}^T \psiv_{P_1}  \leq r  } }  \frac{1}{2}\psiv_{P_1}^{T} \Hm \psiv_{P_1}, 
\end{equation}
where 
\begin{equation}
\Hm=\ \beta \eta^*_\beta \sqrt{\Jm} \Big [  \Qm+ \eta^*_1 \Vm +(1-\eta^*_1)\Wm    -\Tm      \Big ] \sqrt{\Jm}.
\end{equation}
\end{lemma}	

\begin{proof}
The proof can be found in Appendix \ref{sec:prooflemd4}.
\end{proof}
	
Finally, by Lemma \ref{lem:approx} and \ref{lem:optwo} we only need to solve the optimization in \eqref{eq:simp}. Assuming $\Hm$ has negative eigenvalues which is  the case for $|\Xc| \geq 4$ by Weyl's inequality \cite{Johnson}, it can be shown that at the optimal solution $\psiv_{P_1}^*$, the inequality constraint of optimization \eqref{eq:simp} is satisfied with equality, hence
 \begin{align}
 \min_{\substack{ \frac{1}{2} \psiv_{P_1}^T \psiv_{P_1}  \leq r   } }  \frac{1}{2}\psiv_{P_1}^{T} \Hm \psiv_{P_1}  &= r  \min_{\substack{ \frac{1}{2} \psiv_{P_1}^T \psiv_{P_1}  \leq r   } }  \frac{1}{2}\frac{\psiv_{P_1}^{T} \Hm \psiv_{P_1}}{\psiv_{P_1}^T \psiv_{P_1}} \\
 &=  \Lambda_{\min}(\Hm) r
\end{align}	
where $\Lambda_{\min}(\Hm) $ is the smallest eigenvalue of $\Hm$. Taking the derivative of \eqref{eq:worstapproxopt} and setting $r=0$, we obtain
\begin{equation}
\underline{\alpha}_\beta= - \Lambda_{\min}(\Hm). 
\end{equation}
Further upper bounding the $\Lambda_{\min}(\Hm)$ using Weyl's inequality \cite{Johnson}, we get
\begin{align}
-\Lambda_{\min}(\Hm) \geq \beta \eta^*_\beta \Bigg [\frac{Q_{\eta^*_\beta}}{P_1} \Bigg ] _{(3)}
\end{align}
  where we have used the fact that  $\Qm+ \eta^*_1 \Vm +(1-\eta^*_1)\Wm$ can only have three nonzero eigenvalues as $\Qm,\Vm,\Wm$ are rank one matrices.

\section{Proof of Theorem \ref{thm:upper}}
\label{sec:proofth7}

By Lemma \ref{lem:inclusion} we can see that for every $\beta_1\leq \beta_2$ we have  $\Ac_0(Q_1,\beta_1) \subseteq \Ac_0(Q_1,\beta_2)$, and hence $E_1(E_0,Q_1)$ is non-increasing in $\beta$. Therefore, to find an upper bound for $\alpha^*_\beta$ we can find an upper bound to $\alpha^*=\alpha^*_{\beta=1}$ and that would be an upper bound to any test with $\beta <1$. Hence we let $\beta=1$ in the rest of the proof.

To upper bound  $\alpha^*$ in \eqref{eq:optimalalpha} we first find a lower bound to $E_1(E_0,r)$. Rewriting $E_1(E_0,Q_1)$ in the dual form we have \cite{BoroumandTran} 
\begin{align}
E_1(E_0,Q_1) =& \max_{\nu\geq 0}  -\nu \gamma(E_0,Q_1) - \log \sum_{a \in \mathcal{X} } P_{0}^{\nu}(a) P_1(a)  Q_{1}^{-\nu}(a). 
\end{align} 
Setting $\nu=\lambda^*$ where $\lambda^*$ is the optimal Lagrange multiplier in \eqref{eq:tilted} when $P_1$ is known in the test. We get 
\begin{align}\label{eq:lowerEQ}
E_1(E_0,Q_1) \geq -\lambda^* \gamma(E_0,Q_1) - \log \sum_{a \in \mathcal{X} } P_{0}^{\lambda^*}(a) P_1(a)  Q_{1}^{-\lambda^*}(a). 
\end{align} 
 By the Taylor remainder theorem, for every $Q_1 \in \mathcal{B}(P_1,r)$ where $r<r_c$, we have 
\begin{align}
E_1(E_0,Q_1) \geq  E_1(E_0,P_1) +\frac{1}{2}  \thetav_{P_1}^T \Sigm(\tilde{Q}_1) \thetav_{P_1}, 
\end{align} 	 
where 
\begin{align}
\Sigm({Q}_1)&=-\lambda^* \Hm(Q_1) - \lambda^*(1+\lambda^*) \diag \Big(\frac{\hat{Q}_{\lambda^*}}{Q_1^2}  \Big)+(\lambda^*)^2 \uv^T \uv,\\
\Hm(Q_1) & =  \nabla_{Q_1}^2  \gamma(E_0,Q),  \\
\uv&=  \Bigg ( \frac{\hat{Q}_{\lambda^*}(1)}{Q_1(1)}, \ldots,  \frac{\hat{Q}_{\lambda^*}(k)}{Q_1(k)}  \Bigg)^T,\\
\hat{Q}_{\lambda^*}(x)&= \frac{ P_{0}^{\lambda^*}(x) P_1(x) Q_1^{-\lambda^*}(x) } {\sum_{a \in \Xc }  P_{0}^{\lambda^*}(a) P_1(a) Q_1^{-\lambda^*}(a) },
\end{align}
evaluated at  some $\tilde{Q}_1 \in  \mathcal{B}(P_1,r)$, and $\thetav_{P_1}$ is difference vector $Q_1-P_1$ defined in \eqref{eq:diffQ1}. $\Hm(Q_1)$ is the Hessian matrix of the threshold function in \eqref{eq:Hessgam}.  Note that we used the fact that the gradient of RHS in \eqref{eq:lowerEQ} evaluated at $Q_1=P_1$ is equal to zero, hence we only need to control the second order term. Letting
\begin{equation}
\Lambda (\Sigm,r)=\min_{\tilde{Q}_1 \in \Bc (P_1,r) } \Lambda_{\text{min}} \big(\Sigm(\tilde{Q}_1)\big ),
\end{equation}
where $ \Lambda_{\text{min}} \big(\Sigm(\tilde{Q}_1)\big ) $ is the smallest eigenvalue of $ \Sigm(\tilde{Q}_1)$.  We can further lower bound $E_1(E_0,Q_1)$ by
\begin{align}\label{eq:lowerE}
E_1(E_0,Q_1) \geq  E_1(E_0,P_1) +\frac{\Lambda(\Sigm,r)}{2}  \thetav_{P_1}^T \thetav_{P_1},
\end{align} 	
Substituting \eqref{eq:lowerE} in \eqref{eq:limitE} we get
\begin{equation}\label{eq:Elowerapprox}
E_1(E_0,r) \geq  \min_{\substack{D(Q_1\|P_1)\leq r \\  Q_1 \in \Pc(\Xc) } } E_1(E_0,P_1) +\frac{\Lambda(\Sigm,r)}{2}  \thetav_{P_1}^T \thetav_{P_1}.
\end{equation}
Similarly to $E_1(E_0,Q_1)$, by using a Taylor series expansion of $D(Q_1\|P_1)$ around $Q_1=P_1$ and using the remainder theorem, for every $Q_1 \in \mathcal{B}(P_1,r)$ we have 
\begin{align}
D(Q_1\|P_1)  &\geq \frac{\Lambda(\Jm,r)}{2}  \thetav_{P_1}^T  \thetav_{P_1},\label{eq:lowerkl} \\
\Jm(\tilde{Q}_1)&=\diag\bigg( \frac{1}{\tilde{Q}_1(1)},\dotsc,\frac{1}{\tilde{Q}_1(|\Xc|)}\bigg). 
\end{align}	
Therefore, by \eqref{eq:Elowerapprox}, and \eqref{eq:lowerkl}  we get 
\begin{equation}
E_1(E_0,r) \geq  E_1(E_0,P_1) +\frac{\Lambda(\Sigm,r)}{\Lambda(\Jm,r)}r. 
\end{equation}
By Weyl's inequality \cite{Johnson} we have
\begin{align}
\Lambda(\Sigm,r)\geq& \min_{\tilde{Q}_1 \in \Bc (P_1,r) } -\lambda^*  \Bigg(	\frac{\mu^*}{1+\mu^*}  	  \Bigg[\frac{Q_{\mu^*}}{\tilde{Q}_1^2}  \Bigg]_{(1)}+ \frac{1}{(1+\mu^*)} \big(\|\wv_{\tilde{Q}_1}\|_2^2 +\|\qv_{\tilde{Q}_1}\|_2^2 \big) \Bigg)  -\lambda^*(1+\lambda^*)  \Bigg [\frac{Q_{\lambda^*}}{\tilde{Q}_1^2}  \Bigg]_{(1)}\\
\geq & -\lambda^*(4+\lambda^*  )\Bigg( \frac{1}{{Q}_{1}^{\text{min}}(r)} \Bigg ) ^2
\end{align}	
where 
\begin{equation}
Q_1^{\text{min}}(r) =\min_{\substack{ \tilde{Q}_1 \in \Bc(P_1,r) \\ a\in \Xc  }} \tilde{Q}_1(a), 
\end{equation}
and we used $Q_{\lambda^*}(a), Q_{\mu^*}(a) \leq 1$ for all $a \in \Xc $ and $\mu^*\geq0$. Finally,  by $\Lambda(\Jm,r) \geq  1$ we get
\begin{equation}
E_1(E_0,r) \geq  E_1(E_0,P_1) -\lambda^*(4+\lambda^*  )\bigg( \frac{1}{{Q}_{1}^{\text{min}}(r)} \bigg ) ^2 r.
\end{equation}
By the Pinsker's inequality \cite{Cover}, we have  $\|\thetav_{P_1}\|_1 \leq \sqrt{2r}$, hence $Q_1^{\text{min}} \geq  P_1^{\text{min}} -\frac{\sqrt{2r}}{2}$,  
\begin{equation}\label{eq:upperalpha}
E_1(E_0,r) \geq  E_1(E_0,P_1) -\lambda^*(4+\lambda^*  )\Bigg( \frac{\sqrt{r}}{ P_1^{\text{min}} -\sqrt{r} } \Bigg ) ^2 .
\end{equation}
Since $E_1(E_0,r) \geq 0$ for all $r\geq0$ we can improve the lower bound to 
\begin{equation}\label{eq:upperalpha}
E_1(E_0,r) \geq \Bigg \{ E_1(E_0,P_1) -\lambda^*(4+\lambda^*  )\Bigg( \frac{\sqrt{r}}{ P_1^{\text{min}} -\sqrt{r} } \Bigg ) ^2  \Bigg \} \mathbbm{1}  \{ r \leq \bar{r} \} .
\end{equation}
where 
\begin{equation}
\bar{r}= \Bigg(\frac{\kappa P_1^{\text{min}}}{1+\kappa}   \Bigg )^2. 
\end{equation}
Next by \eqref{eq:optimalalpha}, \eqref{eq:upperalpha}, we have
\begin{align}
\alpha^*_\beta &\leq \max \Bigg \{ \max_{0 \leq r \leq r_c} \lambda^*(4+\lambda^*  )\Bigg( \frac{\mathbbm{1}   \{ r \leq \bar{r} \} }{ P_1^{\text{min}} -\sqrt{r} } \Bigg ) ^2   , \frac{E_1(E_0,P_1)}{r_c} \Bigg  \}  \\
&=    \frac{ \lambda^*(4+\lambda^*  )           }{ (P_1^{\text{min}} -\sqrt{\bar{r}} )^2 }     \\
&= \frac{\lambda^*(4+\lambda^*)(1+\kappa)   }{(P_1^{\text{min}})^2} 
\end{align}
where we have dropped the second term in the maximization as it is derived by letting $E_1(E_0,r)$ to be a straight line from $E_1(E_0,r=0)$ at $r=0$ to $0$ at $r=r_c$ and the first term in the maximization is derived by the  lower bound to $E_1(E_0,r)$, and $ \bar{r}< r_c $, hence the first term dominates. This concludes the proof.

\section{Proof of Lemma \ref{lem:cont}} \label{app:cont}
\label{sec:prooflemb1}

We need to  prove the the continuity of $E_1(E_0,r)$ on $r\in [0,r_c)$. To prove this, we will use the Berge's Maximum Theorem stated below \cite{Berge}. 

 \begin{proposition}[Berge's Maximum Theorem]
 	Let $X\subseteq \mathbb{R}^n$ and $ \Theta \subseteq \R^m$. Let $F: X\times \Theta  \rightarrow \mathbb{R}$ be a continuous function, and let $\Gc : \Theta \rightarrow X$ be a compact valued and continuous correspondence. Then the maximum value function 
 	\begin{equation}
 	V(\Theta)= \max_{X\in \Gc(\Theta)} F(X,\Theta) 
 	\end{equation}
 	is well-defined and continuous, and the optimal policy correspondence
 	\begin{equation}
     X^*(\Theta) = \{X \in \Gc(\Theta) | F(X,\Theta) = V(\Theta)\}
 	\end{equation}
 	is nonempty, compact valued, and upper hemicontinuous.
 	\end{proposition}

We first prove that correspondence $\Gc_1(Q_1) = \{Q \in \mathcal{P} (\Xc):  D(Q\|P_0)-D(Q\|Q_1)    \geq  \gamma(E_0,Q_1)\}$ is a continuous correspondence for $Q_1 \in \mathcal{P}_{\delta} (\Xc)$. Then by Berge's theorem $E_1(E_0,Q_1)$ is continuous on $Q_1$  over $\mathcal{P}_{\delta} (\Xc)$. Next, by showing the continuity of the correspondence $\Gc_2(r) = \{Q_1 \in \mathcal{P}_{\delta} (\Xc):  D(Q_1\|P_1)   \leq  r\} $, we conclude the continuity of $E_1(E_0,r)$ in $r$ by Berge's maximum theorem. So we only need to prove that $\Gc_1, \Gc_2$ are compact and continuous correspondences. Note that $Q$ and $Q_1$ are subsets of $\mathbb{R}^{|\Xc|}$ and also $\Gc_1(Q_1), \Gc_2(Q)$ are both closed and bounded, hence By Heine-Borel theorem \cite{Rudin}  both $\Gc_1, \Gc_2$ are compact. To prove the continuity of a correspondence we need to prove the upper and lower hemicontinuity of the correspondence.

\begin{definition}
	The compact valued correspondence $\Gc$ is upper hemicontinuous at $\theta \in \Theta$ if $\Gc(\theta)$ is nonempty and if, for every sequence $\{\theta^{(j)} \}$ with $\theta^{(j)} \rightarrow \theta$ and every sequence $\{X^{(j)} \}$ with $X^{(j)} \in \Gc(\theta^{(j)})$ for all $j$, there exists a convergent subsequence $\{X^{(j_k)}\}$ such that $X^{(j_k)} \rightarrow X \in \Gc(\theta)$.
\end{definition}		
To prove upper hemicontinuity of $\Gc_1(Q_1)$, fix $Q_1 \in \mathcal{P}_{\delta} (\Xc)$ and  assume for every $1\leq j$, $Q_1^{(j)} \in \mathcal{P}_{\delta} (\Xc)$ be a sequence converging to $Q_1$. Moreover, assume for every $j$ there exists a $Q^{(j)}$  such that $Q^{(j)} \in \Gc_1(Q_1^{(j)})$ for all $j$. By the definition of convergence in  metric spaces, since $Q_1^{(j)} \rightarrow Q_1$, then there exists a closed bounded set $\mathcal{Q}_2\subseteq \mathcal{P}_{\delta}$, such that for all large enough $j$, we have $Q_1^{(j)}  \in \mathcal{Q}_2 \subset \mathbb{R}^{|Xc|-1} $.  Furthermore,  since for every $Q_1^{(j)}$, $\Gc_1(Q_1^{(j)})$ is an intersection of a  half space created by the classifier and the probability simplex, $Q^{(j)} \in \Gc_1(Q_1)$ will lie in a closed and bounded subset of $\mathbb{R}^{|\Xc|-1} $. Therefore, for large enough $j$, the  tuple $(Q^{(j)},Q_1^{(j)})$ also lies in a closed and bounded subset of $\mathbb{R}^{2|\Xc|-2} $. Finally, by  the Bolzano-Weierstrass theorem \cite{Rudin}, each bounded sequence in the Euclidean space  $\mathbb{R}^{2|\Xc|-2} $  has a convergent subsequence  $(Q^{(j_k)},Q_1^{(j_k)})$ with the limit point $(Q,Q_1)$. Finally, since each element of this convergent subsequence satisfies the constraint in $\Gc_1(Q_1)$, then by showing the continuity of $\gamma(E_0,Q_1)$ in $Q_1$, we can conclude that the limit point will also satisfy $Q \in \Gc_1(Q_1)$ and gives the upper hemicontinuity.  To prove the continuity of $\gamma(E_0,Q_1)$, note that $D(P\|Q)$ is continuous in the pair for any $ P \in \mathcal{P}(\Xc)$ and $Q \in \mathcal{P}_{\delta}(\Xc)$ where $\delta >0$. Hence, by the Berge's theorem,  $\gamma(E_0,Q_1)$ is continuous in $Q_1\in \Pc_{\delta}(\Xc)$, and the optimizer $Q_{\mu}(Q_1)$ is  upper hemicontinuous. Also, since by \eqref{eq:tiltedmu}, $Q_{\mu}(Q_1)$ is single valued , we get that $Q_{\mu}(Q_1)$ is in fact continuous in $Q_1$. 
\begin{definition}
	The correspondence $\Gc$ is lower hemicontinuous at $\theta \in \Theta$ if $\Gc(\theta)$   is nonempty and if, for every $X \in \Gc(\theta) $  and every sequence $\{\theta^{(j)} \}$ such that $\theta^{(j)}  \rightarrow \theta$, there is a $1\leq J$  and a sequence $\{ X^{(j)}\}$ such that $X^{(j)}  \in \Gc(\theta^{(j)})$ for all $J \leq j$ and $X^{(j)} \rightarrow X$.
	\end{definition}		

To prove the lower hemicontinuity of $\Gc_1(Q_1)$, let $Q_1 \in \mathcal{P}(\Xc)$ and $Q \in \Gc_1(Q_1)$ be fixed. Also, let  $Q_1^{(j)}  \in \mathcal{P}(\Xc) $ be a sequence converging to $Q_1$. Next we  construct the sequence $Q^{(j)}$ such that for every $j$,  $Q^{(j)} \in \Gc_1(Q_1^{(j)})$ and $Q^{(j)} \rightarrow Q$. Note that for every $Q_1^{(j)}$, the correspondence is a half space characterized by the hyperplane
\begin{equation}
\Hc^{(j)}=\{Q\in\Pc(\Xc): \big(Q-Q_{\mu}(Q_1^{(j)})\big)^T \nv^{(j)}=0\}, 
\end{equation} 
where $\nv^{(j)}$ is 
\begin{equation}
\nv^{(j)}= \Bigg(\log \frac{Q_{\mu}(Q_1^{(j)})(1)}{P_0(1)},\ldots, \log \frac{Q_{\mu}(Q_1^{(j)})(|\Xc|)}{P_0(|\Xc|)} \Bigg)^T.
\end{equation} 
Also for every $Q^{(j)}$, define the line passing $P_0$ and intersecting $\Hc^{(j)}$ as 
\begin{align}
L(Q^{(j)})&= \beta P_0 +(1-\beta)  Q_{\dagger}^{(j)},  \beta \in [0,1],\\
 Q_{\dagger}^{(j)}&= \{P: P=\beta P_0+(1-\beta)Q, \beta\in \mathbb{R}  \} \cap \Hc^{j}.
\end{align}
Moreover, for every $Q$ let $Q=\beta^* P_0 +(1-\beta^*)  Q_{\dagger}$. Finally, we can define the sequence $Q^{(j)}$ as
\begin{align}
Q^{(j)}=\beta^*P_0+ (1-\beta^*)  Q_{\dagger}^{(j)}.
\end{align}
By our construction, it is clear that  for every $j$,  $Q^{(j)} \in \Gc_1(Q_1^{(j)})$. Moreover, by continuity of $Q_{\mu}(Q_1^{(j)})$, we conclude that $\Hc^{(j)} \rightarrow \Hc$ and consequently $Q_{\dagger}^{(j)} \rightarrow Q_{\dagger} $, and $ Q^{(j)} \rightarrow Q$.

Next, we show the continuity of $\Gc(r)$. To show the upper hemicontinuity of the correspondence, we can use the same argument as the upper hemicontinuity of $\Gc_1(Q_1)$. So it only remains to prove the lower hemicontinuity. Let $ 0\leq r$ and $Q_1 \in \Pc_{\delta}(\Xc)$ be fixed. Also, let $0\leq r^{(j)}$ is a sequence converging to $r$. Construct $Q_1^{(j)}\in\Pc_{\delta}(\Xc)$ as  
\begin{equation}
Q^{(j)}_1=\bigg( \frac{r^{(j)}}{r}  \bigg) Q_1 + \bigg( 1- \frac{r^{(j)}}{r}  \bigg) P_1.
\end{equation}  
It is clear that $Q^{(j)}_1 \rightarrow Q_1$ as $r^{(j)} \rightarrow r$. Also by convexity of KL-divergence we get
\begin{equation}
D \big(Q^{(j)}_1  \| P_1 \big) \leq \frac{r^{(j)}}{r}  D(Q_1\|P_1) \leq r^{(j)}
\end{equation}
where in the last inequality we used the fact that $Q_1\in \Gc(r)$. Therefore, $Q^{(j)}_1 \in \Gc\big(r^{(j)}\big)$, and that gives the lower hemicontinuity of $\Gc(r)$, and concludes the proof.  


\section{Proof of Lemma \ref{lem:gamsen}}
\label{sec:prooflemd1}

By the envelope theorem \cite{Segal}, we can find the partial derivative simply by taking the derivative of Lagrangian and evaluating  at its optimal solution. Writing the Lagrangian we have
\begin{align}\label{eq:lagrange}
L(Q,Q_1, \mu, \nu)= D(Q\|Q_1) + \mu \big( D(Q\|P_0)-E_0\big ) + \nu \Big ( \sum_{x\in \Xc} Q(x)-1\Big ),
\end{align}
hence
\begin{align}
\frac{\partial \gamma}{\partial Q_1 (j)} &=\beta \frac{\partial L}{\partial Q_1(j)} \Bigg |_{Q=Q_{\mu^*}}\\
&=-\beta \frac{Q_{\mu^{*}}(j)}{Q_1(j)}. 
\end{align}
Taking the second order partial derivative, we get
\begin{align}\label{eq:secder}
\frac{\partial^{2} \gamma}{\partial Q_1 (i)Q_1(j)} =\beta \Big (-\frac{1}{Q_1(j)}\frac{\partial Q_{\mu^{*}}(j)}{\partial Q_1(i)} + \frac{Q_{\mu^*}(i)}{Q_1^2(i)} \mathbbm{1}\{i=j\}\Big ).
\end{align}
Therefore, we need to find the sensitivity of $Q_{\mu^{*}}$ ,the optimal solution of optimization \eqref{eq:clsfix},  to local changes in $Q_1$. \cite{Castil} presents a general approach to find the partial derivative of primal and dual variables solution with respect to any parameter of the optimization problem. For ease of reference, we state this result. The result in \cite{Castil}  is more general; however, we only state the version we need in our proof.
\begin{proposition}\label{prop:sensi}  
Consider the following primal non-linear programming problem:
\begin{equation}
\min_{\substack{\xv: ~g(\xv,\av)\leq 0 \\ ~~ h(\xv,\av)=0   }} z=f(\xv,\av),
\end{equation}
where $f,g,h: \mathbb{R}^n \times \mathbb{R}^p \rightarrow \mathbb{R}$, and $f,h,g \in C^2$. Let $\mu, \nu$ be the Lagrange multipliers corresponding to inequality and equality constraint. Furthermore, assume at the optimal solution the constraints are active and $\mu^*\neq0$, then

\begin{equation}
\Bigg (\frac{\partial \xv}{\partial \av} , \frac{\partial \nu }{\partial \av},  \frac{\partial \mu}{\partial \av}  , \frac{\partial z}{\partial a} \Bigg )^T= \Um^{-1} \Sm
\end{equation}
where
\begin{equation}
\Um=
\begin{pmatrix}
\Fm_{xx}& \Hm_x& \Gm_x & \bold{0}  \\
\Hm_x ^{T} & \bold{0}   & \bold{0}  & \bold{0}   \\
\Gm_x^{T} & \bold{0}  & \bold{0}  & \bold{0}  \\
\Fm_x^{T} & \bold{0} & \bold{0}   & -1   \\
\end{pmatrix}
, ~~\Sm=-
\begin{pmatrix}
\Fm_{xa} & \\
\Hm_a ^{T} &  \\
\Gm_a ^{T}&\\
\Fm_a^{T} & 
\end{pmatrix}
\end{equation}
 and $\bold{0}$ is the matrix with entries equal to zero with the corresponding dimensions and the submatrices are defined as
\begin{align}
\Fm_{xx}&=\nabla_{xx} f(x^*,a) + \mu^* \nabla_{xx} g(x^*,a) +\nu^* \nabla_{xx} h(x^*,a), \nonumber \\
\Fm_{xa}&=\nabla_{xa} f(x^*,a) + \mu^* \nabla_{xa} g(x^*,a) +\nu^* \nabla_{xa} h(x^*,a), \nonumber \\
\end{align}
where $(x^*,\mu^*,\nu^*)$ is the optimal primal and dual variable solutions. Moreover, $\Fm_x, \Fm_a, \Gm_x, \Gm_a, \Hm_x, \Hm_a$ are gradient of $f, g, h$ respect to $x, a$ evaluated at $x^*$ respectively.	
\end{proposition}	
Using this proposition we can find $\frac{\partial Q_{\mu^{*}}}{\partial Q_1} $ in equation  \eqref{eq:secder}. Letting $x=Q$, and $a=Q_1$ we have  
	\begin{align}
	&\Fm_{QQ}=(1+\mu^*)\diag \bigg ( \frac{1}{Q_{\mu^*}(1)},\ldots,\frac{1}{Q_{\mu^*}(|\Xc|)}\bigg),\\
	&\Fm_{QQ_1}=-\diag \Big ( \frac{1}{Q_1(1)},\ldots,\frac{1}{Q_1(|\Xc|)}\Big),\\
	&\Fm_{Q}= \Bigg ( 1+\log \frac{Q_{\mu^*}(1)}{Q_1(1)},\ldots,  1+\log \frac{Q_{\mu^*}(|\Xc|)}{Q_1(|\Xc|)}    \Bigg)^T,\\
	&\Fm_{Q_1}=\Bigg (-\frac{Q_{\mu^*}(1)}{Q_1(1)},\ldots,- \frac{Q_{\mu^*}(|\Xc|)}{Q_1(|\Xc|)}    \Bigg)^T,\\
	&\Gm_{Q}=\Bigg ( 1+\log \frac{Q_{\mu^*}(1)}{P_0(1)},\ldots,  1+\log \frac{Q_{\mu^*}(|\Xc|)}{P_0(|\Xc|)}    \Bigg) ^T,\\
	&\Gm_{Q_1}=\big ( 0,\ldots,  0    \big)^T,\\
	&\Hm_{Q}=\big ( 1,\ldots, 1 \big)^T,\\
	&\Hm_{Q_1}=\big ( 0,\ldots,  0    \big)^T.
	\end{align}
Writing matrix $\Um$ as
\begin{equation}\label{}
\Um=
\begin{pmatrix}
\Am^{\tiny (|\Xc| \times |\Xc|)} & \Bm^{\tiny (|\Xc| \times 3)}  \\
\Cm^{\tiny (3\times |\Xc|)} & \Dm^{\tiny (3\times3)}
\end{pmatrix},
\end{equation} 
we have
\begin{align}
\Am=(1+\mu^*)\diag \bigg ( \frac{1}{Q_{\mu^*}(1)},\ldots,\frac{1}{Q_{\mu^*}(|\Xc|)}\bigg), 
\end{align}
\begin{align}
 \Bm=
\begin{pmatrix}
1 & 1+ \log \frac{Q_{\mu^*}(1)}{P_0(1)}& 0 \\
\vdots & \vdots & \vdots\\
1& 1+ \log \frac{Q_{\mu^*}(|\Xc|)}{P_0(|\Xc|)} & 0\\ 
\end{pmatrix},
~~ \Cm=
\begin{pmatrix}
1 & \ldots & 1\\
 1+ \log \frac{Q_{\mu^*}(1)}{P_0(1)}& \ldots &  1+ \log \frac{Q_{\mu^*}(|\Xc|)}{P_0(|\Xc|)} \\
 1+ \log \frac{Q_{\mu^*}(1)}{Q_1(1)}& \ldots &  1+ \log \frac{Q_{\mu^*}(|\Xc|)}{Q_1(|\Xc|)} 
\end{pmatrix},
~~\Dm=
\begin{pmatrix}
0 & 0 & 0\\
0 & 0 & 0\\
0 & 0 & -1
\end{pmatrix}.
\end{align}
Also writing $\Sm$ as
\begin{equation}
\Sm=
\begin{pmatrix}
\Fm ^{\tiny(|\Xc|\times |\Xc|)}\\
\Km^{\tiny (3 \times |\Xc|)}
\end{pmatrix},
\end{equation}
we have
\begin{align}
&\Fm= \diag \bigg ( \frac{1}{Q_1(1)},\ldots,\frac{1}{Q_1(|\Xc|)}\bigg),\\
&\Km=
\begin{pmatrix}
0& \ldots & 0\\
0& \ldots & 0\\
\frac{Q_{\mu^*}(1)}{Q_1(1)} & \ldots & \frac{Q_{\mu^*}(|\Xc|)}{Q_1(|\Xc|)}
\end{pmatrix}.
\end{align}
By the blockwise inversion formula, we have
\begin{align}\label{eq:blockinv}
&\Um^{-1}=
\begin{pmatrix}
\Am^{-1} +  \Am^{-1} \Bm \Mm \Cm \Am^{-1} & -\Am^{-1} \Bm \Mm    \\
- \Mm \Cm \Am^{-1}  & \Mm
\end{pmatrix},
\end{align} 
where $\Mm= \big( \Dm -\Cm\Am^{-1} \Bm \big)^{-1}$.
Since we are only interested in $\frac{\partial Q_{\mu^{*}}}{\partial Q_1} $ it suffices to find the first $|\Xc|$ rows of $\Um^{-1}$. Applying the block inversion formula to find $\Mm$ we get 
\begin{align}
\Bm \Mm  = \frac{1+\mu^*}{\Var_{Q_{\mu^*}} \big (\log \frac{Q_{\mu^*}}{P_1} \big )}
\footnotesize\begin{pmatrix}
(1+E_0) \big ( \log \frac{Q_{\mu^*}(1)}{P_0(1)} -E_0 \big )- \Var_{Q_{\mu^*}} \big (\log \frac{Q_{\mu^*}}{P_0} \big ), & - \big ( \log \frac{Q_{\mu^*}(1)}{P_0(1)} -E_0  \big ),&  0\\
\vdots & \vdots & \vdots\\
(1+E_0) \big ( \log \frac{Q_{\mu^*}(|\Xc|)}{P_0(|\Xc|)} -E_0 \big )- \Var_{Q_{\mu^*}} \big (\log \frac{Q_{\mu^*}}{P_0} \big ), & -\big ( \log \frac{Q_{\mu^*}(|\Xc|)}{P_0(|\Xc|)} -E_0  \big ),&  0
\end{pmatrix}
\end{align}
where we used that at the optimal solution $D(Q_{\mu^*}\|P_0)=E_0$. Furthermore, we get
\begin{align}
\Am^{-1} \Bm \Mm \Cm \Am^{-1}  = -\frac{1}{(1+\mu^*)} (\Qm+\Wm).
\end{align}
Finally, by the structure of $\Km$, and $\Bm \Mm$ we can easily see that
\begin{align}\label{eq:Jacob}
-\diag \Big(\frac{1}{Q_1}\Big) *\Big ( \frac{\partial {Q}_{\mu^{*}}}{\partial Q_1} \Big)  \Bigg |_{Q_1=P_1} &=- \Fm\big (\Am^{-1} +  \Am^{-1} \Bm \Mm \Cm \Am^{-1} \big ) \Fm\\
&=\frac{1}{1+\mu^*} \Jm \Bigg [ \Tm - (\Qm+\Wm)   \Bigg]\Jm.
\end{align} 
Substituting \eqref{eq:Jacob} into \eqref{eq:secder} we get the Hessian matrix in \eqref{eq:Hessgam}.


\section{Proof of Lemma \ref{lem:HessianE}}
\label{sec:prooflemd2}

	Similar to Lemma \ref{lem:gamsen} by writing the Lagrangian and using the envelope theorem \cite{Segal} we have
	\begin{align}\label{eq:lagrangeE}
	L(Q,Q_1, \eta, \nu)= D(Q\|P_1)  + \eta \big( D(Q\|P_0)-\beta D(Q\|Q_1) +  \gamma(E_0,Q_1)  \big )+ \nu \Big ( \sum_{x\in \Xc} Q(x)-1\Big ),
	\end{align}
	hence
	\begin{equation}
	\frac{\partial E_1(E_0,Q_1)}{\partial Q_1 (j)} =\frac{\partial L}{\partial Q_1(j)} \Bigg |_{Q=\hat{Q}_{\eta^*}}= {\eta^*}\Bigg( \beta \frac{\hat{Q}_{{\eta^*}}(j)}{Q_1(j)}  +\frac{\partial \gamma}{\partial Q_1 (j)}   \Bigg),\label{eq:firstord}
	\end{equation}

	\begin{equation}\label{eq:tiltedMM1}
	\hat{Q}_{\eta^*}(x)= \frac{ P_1^{\frac{1}{1+\eta^*-\eta^* \beta}}(x)  Q_{1}^{\frac{-\eta^* \beta }{1+\eta^*-\eta^* \beta}}(x) P_{0}^{\frac{\eta^*}{1+\eta^*-\eta^* \beta}}(x) } {\sum_{a \in \Xc } P_1^{\frac{1}{1+\eta^*-\eta^* \beta}}(a)  Q_{1}^{\frac{-\eta^* \beta }{1+\eta^*-\eta^* \beta}}(a) P_{0}^{\frac{\eta^*}{1+\eta^*-\eta^* \beta}}(a) }.
	\end{equation}
	Letting $Q_1=P_1$, we get
	\begin{equation}\label{eq:tiltbeta}
	{Q}_{\eta^*_\beta}(x)= \frac{  P_{1}^{\frac{1-\eta^*_\beta \beta }{1+\eta^*_\beta-\eta^*_\beta \beta}}(x) P_{0}^{\frac{\eta^*_\beta}{1+\eta^*_\beta-\eta^*_\beta \beta}}(x) } {\sum_{a \in \Xc } P_{1}^{\frac{1-\eta^*_\beta \beta }{1+\eta^*_\beta-\eta^*_\beta \beta}}(a) 		P_{0}^{\frac{\eta^*_\beta}{1+\eta^*_\beta-\eta^*_\beta \beta}}(a) }.
	\end{equation}	
	Furthermore, since $D(Q_{\eta^*_\beta}\|P_0)=E_0$, it is easy to see that,  ${Q}_{\eta^*_\beta}=	Q_{\lambda^*}$ in \eqref{eq:tilted}, and $\frac{\eta^*_\beta}{1+\eta^*_\beta-\eta^*_\beta \beta}=\lambda^*$. Also, by  \eqref{eq:KKTgammaLRT}  ${Q}_{\eta^*_\beta}=	Q_{\mu^*}$, and finally from  \eqref{eq:firstdergam}  we	conclude that
	\begin{align}
	\Big( \beta \frac{\hat{Q}_{\eta^*}(j)}{Q_1(j)}  +\frac{\partial \gamma}{\partial Q_1 (j)} \Big) \bigg|_{Q_1=P_1}&=\beta \frac{{Q}_{\eta^*_\beta}(j)}{P_1(j)}  -\beta \frac{Q_{\mu^*}(j)}{P_1(j)} \\
	&=0, 
	\end{align}
	 which concludes \eqref{eq:zeroder}. Next, by taking the second derivative of \eqref{eq:firstord} we get
	\begin{align}\label{eq:secordE}
	\frac{1}{\beta} \frac{\partial^2 E_1(E_0,Q_1)}{\partial Q_1(i) Q_1(j)} &=  \frac{\partial {\eta^*} }{\partial Q_1(i) }   \bigg(  \frac{\hat{Q}_{\eta^*}(j)}{Q_1(j)}  + \frac{1}{\beta}  \frac{\partial \gamma}{\partial Q_1 (j)}   \bigg)\\
	&+  {\eta^*} \bigg( \frac{1 }{Q_1(j)} \frac{\partial \hat{Q}_{\eta^*}(j)}{\partial Q_1(i)} -  \frac{\hat{Q}_{\eta^*}(i)}{Q_1^2(i)} \mathbbm{1}\{i=j\}   + \frac{1}{\beta} \frac{\partial^2 \gamma}{\partial Q_1(i) Q_1 (j)}    \bigg).
	\end{align}
	Letting $Q_1=P_1$ the first term is zero, hence we only require to find $\frac{\partial \hat{Q}_{\eta^*}(j)}{\partial Q_1(i)} \Big|_{Q_1=P_1}$. For simplicity in calculations, we write all the expressions in term of the optimal Lagrange multipliers when $\beta=1$. Since for every 		$0<\beta \leq 1$ the optimizing distribution $Q$ when $Q_1=P_1$  is the tilted distribution of $P_0$, and $P_1$ and since for every such $Q$, the optimizing distribution should satisfy the condition \eqref{eq:KKTgammaLRT},  hence the tilted exponent is equal for every $\beta$ 	and by equating the exponents in \eqref{eq:tiltbeta}  	we can define
	\begin{equation}
	\rho \triangleq  \frac{\eta^*_\beta}{\eta^*_1}=1+\eta^*_\beta-\beta \eta^*_\beta, 	
	\end{equation}
	where $\eta^*_\beta$ is the Lagrange multiplier  for aribitrary $\beta$ and $\eta^*_1$ is the Lagrange multiplier when $\beta=1$.  Using proposition \ref{prop:sensi} and letting $x=Q$, and $a=Q_1$, and setting $Q_1=P_1$, we have
	\begin{align}
	&\Fm_{QQ}^\beta= \rho \Fm_{QQ} , \quad \Fm_{QQ_1}^\beta= \rho \Fm_{QQ_1} , \\
	& \Fm_{Q}^\beta = \Fm_{Q}, \quad \Fm_{Q_1}^\beta = \Bm_Q,     \\
	&\Gm_{Q}^\beta = \beta \Gm_{Q} + (1-\beta) \Fm_{Q} , \quad \Gm_{Q_1}^\beta =\Gm_{Q_1}, \\
	&\Hm_{Q}^\beta=\Hm_{Q} \quad  \Hm_{Q_1}^\beta=\Hm_{Q},
	\end{align}	
	where
	\begin{align}
	\Fm_{QQ}&=  \diag \Bigg ( \frac{1}{Q_{\eta^*_1}(1)},\ldots,\frac{1}{Q_{\eta^*_1}(|\Xc|)}\Bigg),\\
	\Fm_{QQ_1}&= \lambda^*_1  \diag \Big ( \frac{1}{P_1(1)},\ldots,\frac{1}{P_1(|\Xc|)}\Big),\\
	\Fm_{Q}&= \Bigg ( 1+\log \frac{Q_{\eta^*_1}(1)}{P_1(1)},\ldots,  1+\log \frac{Q_{\eta^*_1}(|\Xc|)}{P_1(|\Xc|)}    \Bigg)^T,\\
	\Fm_{Q_1}&= \big ( 0,\ldots,  0    \big)^T,     \\
	\Bm_{Q}&=\Bigg ( 1+\log \frac{Q_{\eta^*_1}(1)}{P_0(1)},\ldots,  1+\log \frac{Q_{\eta^*_1}(|\Xc|)}{P_0(|\Xc|)}    \Bigg)^T,\\
	\Gm_{Q}&=\Bigg ( \log \frac{P_1(1)}{P_0(1)},\ldots,  \log \frac{P_1(|\Xc|)}{P_0(|\Xc|)}    \Bigg) ^T,\\
	\Gm_{Q_1}&= \Bigg ( \bigg( \frac{\hat{Q}_{\eta^*_1}(1)}{Q_1(1)}  +\frac{\partial \gamma}{\partial Q_1 (1)} \bigg) \Bigg|_{Q_1=P_1}    ,\ldots, \bigg( \frac{\hat{Q}_{\eta^*_1}(|\Xc|)}{Q_1(|\Xc|)}  +\frac{\partial \gamma}{\partial Q_1 (|\Xc|)} \bigg) \Bigg|_{Q_1=P_1}    \Bigg)^T \\
	& =\big ( 0,\ldots,  0    \big)^T,\\
	\Hm_{Q}&=\big ( 1,\ldots, 1 \big)^T,\\
	\Hm_{Q_1}&=\big ( 0,\ldots,  0    \big)^T.
	\end{align}
	Similar to the  Lemma \ref{lem:gamsen} we can write $\Um$ as 
	\begin{equation}\label{}
	\Um_\beta =
	\begin{pmatrix}
	\Am^{\tiny (|\Xc| \times |\Xc|)}_\beta & \Bm^{\tiny (|\Xc| \times 3)} _\beta \\
	\Cm^{\tiny (3\times |\Xc|)}_\beta & \Dm^{\tiny (3\times3)}_\beta
	\end{pmatrix},
	\end{equation} 
	where
	\begin{align}
	\Am_\beta&=\rho \Am_1\\ 
	\Bm_\beta&=\beta \Bm_1 + (1-\beta) \Em_0\\
	\Cm_\beta &= \beta \Cm_1 +(1-\beta) \Lm_0\\
	\Dm_\beta&=\Dm_1
	\end{align}
	and
	  \begin{align}
	  \Am_1=\diag \Bigg ( \frac{1}{Q_{\eta^*_1}(1)},\ldots,\frac{1}{Q_{\eta^*_1}(|\Xc|)}\Bigg), \quad
	  \Dm_1=
	  \begin{pmatrix}
	  0 & 0 & 0\\
	  0 & 0 & 0\\
	  0 & 0 & -1
	  \end{pmatrix},
	  \end{align}
	  \begin{align}
	  \Bm_1=
	  \begin{pmatrix}
	  1 & \log \frac{P_1(1)}{P_0(1)}& 0 \\
	  \vdots & \vdots & \vdots\\
	  1& \log \frac{P_1(|\Xc|)}{P_0(|\Xc|)} & 0\\ 
	  \end{pmatrix}, \quad
	  \Em_0=
	  \begin{pmatrix}
	  1 &1+ \log \frac{Q_{\eta_1^*} (1)}{P_0(1)}& 0 \\
	  \vdots & \vdots & \vdots\\
	  1& 1+\log \frac{Q_{\eta_1^*} }{P_0(|\Xc|)} & 0\\ 
	  \end{pmatrix},
	   \end{align}
	   \begin{align}
	  \Cm_1=
	  \begin{pmatrix}
	  1 & \ldots & 1\\
	  \log \frac{P_1(1)}{P_0(1)}& \ldots &   \log \frac{P_1(|\Xc|)}{P_0(|\Xc|)} \\
	  1+ \log \frac{Q_{\eta_1^*}(1)}{P_1(1)}& \ldots &  1+ \log \frac{Q_{\eta_1^*}(|\Xc|)}{P_1(|\Xc|)} 
	  \end{pmatrix}, \quad
	 \Lm_0=
	  \begin{pmatrix}
	  1 & \ldots & 1\\
	 1+ \log \frac{ Q_{\eta^*_1}(1)   }{P_0(1)}& \ldots &   1+\log \frac{Q_{\eta^*_1}(1)}{P_0(|\Xc|)} \\
	  1+ \log \frac{Q_{\eta_1^*}(1)}{P_1(1)}& \ldots &  1+ \log \frac{Q_{\eta_1^*}(|\Xc|)}{P_1(|\Xc|)} 
	  \end{pmatrix},
	  \end{align}
	  Also writing $\Sm$ as
	  \begin{equation}
	  \Sm=
	  \begin{pmatrix}
	  -\Fm_{QQ_1}^\beta \\
	  \Km
	  \end{pmatrix},
	  \end{equation}
	  we have
	  \begin{align}
	  \Km^{\tiny (3 \times |\Xc|)}&=
	  \begin{pmatrix}
	  0& \ldots & 0\\
	  0& \ldots & 0\\
      	0& \ldots & 0
	  \end{pmatrix}.
	  \end{align}
By block inversion formula in \eqref{eq:blockinv}, and similar arguments as in Lemma \ref{lem:gamsen} we have
\begin{align}
\Am_\beta ^{-1} \Bm_\beta  \big( \Dm_\beta -\Cm_\beta \Am_\beta^{-1} \Bm_\beta \big)^{-1} \Cm_\beta \Am_\beta^{-1} = -\frac{1}{\rho} \big (\Tm^2 +\Vm \big ),
\end{align}
where $\Vm=\vv \vv^T$ and 
\begin{align}
\vv=\frac{1}{\sqrt{\Var{Q_{\eta^*_\beta} }(\Omega)   }}    \Big ( Q_{\eta^*_\beta}(1)\Omega(1),\ldots,Q_{\eta^*_\beta}(|\Xc|)\Omega(|\Xc|)   \Big ),
\end{align}
\begin{align}
\Omega(i)=\beta \log \frac{P_1(i)}{P_0(i)} +(1-\beta)  \log \frac{Q_{\eta^*_\beta}(i)}{P_0(i)}, i \in \Xc,
\end{align}
Finally, by the structure of $\Sm$, we can easily derive
\begin{align}\label{eq:Jacob2}
\diag \Big(\frac{1}{P_1}\Big) *\Big ( \frac{\partial \hat{Q}_{\eta^{*}}}{\partial Q_1} \Big)  \Bigg |_{Q_1=P_1} &=\Jm \Big [ -\Tm + (\Tm^2+\Vm)   \Big]\Jm. 
\end{align} 
Substituting \eqref{eq:Jacob2} and \eqref{eq:Hessgam}  into \eqref{eq:secordE}, we get 
\begin{align}
\Hm_2= \rho \eta^*_1  \Jm \Big [   -\Tm + \Qm+  \eta^*_1 \Vm +(1-\eta^*_1)\Wm        \Big ] \Jm,
\end{align}
%
%
where we have used the identity $\eta_1^*=\frac{\mu^*}{1+\mu^*}$ derived by setting the tilted distribution exponents to be equal for when $Q_1=P_1$.


    \section{Proof of Lemma \ref{lem:approx}}
	\label{sec:prooflemd3}
	
	By applying a Taylor expansion to $E_1(E_0,Q_1)$ around $Q_1=P_1$ we obtain
	\begin{align}\label{eq:linearapprox1}
	E_1(E_0,Q_1)=E_1(E_0,P_1) +   \thetav_{P_1}^{T}   \nabla E_1(E_0,Q_1)\big|_{Q_1=P_1}   + \frac{1}{2}\thetav_{P_1}^{T} \Hm_2 \thetav_{P_1}   + o(\| \thetav_{P_2} \|_{\infty}^{2}).
	\end{align}
   The first term in the expansion is  $E^*_1(E_0)$, also by Lemma \ref{lem:HessianE}, the gradient evaluated at $Q_1=P_1$ is zero.  Further approximating the constraint in \eqref{eq:secondopt}, we get 
	\begin{equation}\label{eq:KLapprox}
	D(Q_1 \| P_1 ) = \frac{1}{2}  \thetav_{P_1}^T \Jm \thetav_{P_1} + o (\| \thetav_{P_1} \|^2_{\infty}).
	\end{equation}
    By substituting the expansions \eqref{eq:linearapprox1} and \eqref{eq:KLapprox} in  \eqref{eq:limitE} we obtain
	\begin{align}\label{eq:approxworst}
	 E_1 (E_0,r)= E^*_1(E_0)+ \min_{\substack{ \frac{1}{2} \thetav_{P_1}^T \Jm \thetav_{P_1}+  o (\| \thetav_{P_1} \|^2_{\infty})  \leq r \\ \onev^T\thetav_{P_1}=0  } }  \frac{1}{2}\thetav_{P_1}^{T} \Hm_2 \thetav_{P_1}  + o(\|\thetav_{P_1} \|_{\infty}^2).
	\end{align}
	To find the error of above approximation  as a function of $r$, first note that we can take $o(\|\thetav_{P_1} \|_{\infty}^2)$ out of minimization and substitute it with $o(\|\thetav^*_{P_1}(r) \|_{\infty}^2)$ where $\thetav^*_{P_1} $ is the optimal solution of the minimization. Moreover, approximating the inequality constraint can result an error of $o(\sqrt{r})$ in $\|\thetav^*_{P_1}\|_{\infty}$. Therefore, by inequality constraint and the fact that it impose a constraint on the length of the vector $\thetav$ we have that $\|\thetav^*_{P_1}\|_{\infty} \leq c\sqrt{r}+o(\sqrt{r})$ where $c$ is  independent from $r$ and only depends on $\Jm, \Hm_2$.  This argument and  \eqref{eq:approxworst}  concludes \eqref{eq:worstapproxopt}.
	

	 \section{Proof of Lemma \ref{lem:optwo}}
	 \label{sec:prooflemd4}

	Since $\Jm$ is a diagonal matrix with non zero diagonal entries, we have
		\begin{align}
	\sqrt{\Jm} &=\diag\bigg( \frac{1}{\sqrt{P_1(1)}},\dotsc,\frac{1}{\sqrt{P_1({|\Xc|})}}\bigg).
	\end{align}
	Letting $\psiv_{P_1}= \sqrt{\Jm} \thetav_{P_1}$, we obtain
	\begin{align}
	 \frac{1}{2}\thetav_{P_1}^{T} \Hm_2 \thetav_{P_1} &= \frac{1}{2}\psiv_{P_1}^{T} \Hm \psiv_{P_1}, \\
	\frac{1}{2} \thetav_{P_1}^T \Jm \thetav_{P_1}&= \frac{1}{2} \psiv_{P_1}^T  \psiv_{P_1},\\
	\onev^T\thetav_{P_1}&= \onev^T  \sqrt{\Jm}^{-1}  \psiv_{P_1}.
	\end{align}
	Next we need to show that we can drop the equality constraint. Note that $\sigv_1 \defeq \frac{\sqrt{\Jm}^{-1} \onev}{\| \sqrt{\Jm}^{-1} \onev\|} $ is in the null space of $\Hm$, i.e.,
	\begin{align}
	\Hm\sqrt{\Jm}^{-1} \onev = \beta \eta^*_\beta \sqrt{\Jm} \Big [  \Qm+ \eta^*_1 \Vm +(1-\eta^*_1)\Wm    -\Tm      \Big ] \onev=\zerov
	\end{align}
since $\vv^T \onev=0, \wv^T \onev=0,  \Tm  \onev=\qv$. Now, assume the optimizer of the second optimization in \eqref{eq:simp} is of the form 
\begin{equation}\label{eq:sol}
\psiv_{P_1}=\rho_1 \sigv_1 + \rho_2 \sigv_2, 
\end{equation}
where 	$\sigv_1 \perp \sigv_2, \| \sigv_2\|=1$. Then
\begin{align}
\psiv_{P_1}^{T} \Hm \psiv_{P_1}&=  (\rho_1\sigv_1 +\rho_2 \sigv_2)^{T} \Hm (\rho_1\sigv_1 +\rho_2 \sigv_2)\\
&= \rho_2^2  \sigv_2^T  \Hm \sigv_2
\end{align}
where we also used $ \Hm \sigv_1 =0, \Hm=\Hm^T$. Assuming $\Hm$ has negative eigenvalues (otherwise, both optimizations are equal to zero), the inequality constraint should be satisfied with equality. Then 
\begin{align}
\psiv_{P_1}^{T} \Hm \psiv_{P_1}  = \Big (2r-\rho_1^2 \Big)   \sigv_2^T  \Hm \sigv_2.
\end{align}
Therefore, to achieve the minimum, $\rho_1$ must be zero, which concludes the proof.


 \section{Proof of Theorem \ref{thm:Stein}}
\label{sec:proofth41}

By Theorem \ref{thm:E1}, it can be shown that by taking  $E_0=\Theta({n}^{-1})$ there exists an $\alpha$ such that the proposed classifier \eqref{eq:clsfix} achieves $E_1=D(P_0\|P_1)$ which is equal to the Stein regime exponent of hypothesis testing with known distributions, and hence the \eqref{eq:Stein1trade} is equal to $D(P_0\|P_1)$. In fact, there is no need for a training sequence to achieve the \eqref{eq:Stein1trade}. Since Hoeffding's test achieves the optimal error exponent tradeoff only by having the distribution $P_0$, it is easy to see that the Hoeffding test achieves the Stein regime exponent for the unknown distribution for any $P_1$. However, to achieve the largest error exponent under $P_0$ such that it guarantees that for any $P_1$ the type-\RNum{2} probability of error is bounded by some $\epsilon \in (0,1)$, Hoeffding's and our proposed classifier are not universal, since for any choice of threshold $E_0>0$, the type-\RNum{2} probability of error for any distribution $P_1$ such that $D(P_1\|P_0)<E_0$ converges to one. We show that Gutman's universal test \cite{gutman1989asymptotically} achieves the largest type-\RNum{1} error exponent while the type-\RNum{2} probability of error is bounded away from one. Using the Gutman's test we obtain 
 \begin{equation}
 D_{\alpha}^{\rm GJS} (\Tx\|\TX) \leq \frac{1}{2n} G^{-1}_{ |\Xc |-1}(\epsilon), 
 \end{equation}
 where
 \begin{equation}
  D_{\alpha}^{\rm GJS} (Q\|P) =D\bigg(Q\bigg  \| \frac{Q+\alpha P}{1+\alpha}\bigg) + \alpha D\bigg(P \bigg \|\frac{Q+\alpha P}{1+\alpha}\bigg),
 \end{equation}  
 is the generalized Jensen-Shannon divergence and $ G^{-1}_a(.)$ is the inverse of the complementary CDF of a chi-squared random variable with $a$ degrees of freedom; \cite{Tan} shows that Gutman test with the chosen threshold achieves type-\RNum{2} error probability of $\epsilon$ for any $P_1$ and achieves the type-\RNum{1} error exponent $D_{\frac{\alpha}{1+\alpha}}(P_1\|P_0)$. We prove the converse showing that the achievable error exponent using Gutman's test is, in fact the Stein's regime exponent defined in \eqref{eq:Stein1trade}. Note that as $\alpha \rightarrow \infty $, i.e., when the  number of training samples is  much larger than the number of test samples, then the exponent $D_{\frac{\alpha}{1+\alpha}}  (P_1\|P_0)$  converges to $D(P_1\|P_0)$  which is the Stein regime error exponent under when both distributions are known.
We prove the asymptotic optimality of Gutman's test in this setting. Specifically, we show that for any test such that the type-\RNum{1}  error exponent $E_0^{(\epsilon)}>  D_{\frac{\alpha}{1+\alpha}}(P_1\|P_0) $, there exists a $P_1$ such that  type-\RNum{2} probability of error of the test goes to one, i.e.,  for any test $\phi_n$ such that 
	\begin{equation}	
	\limsup_{n\rightarrow \infty} \epsilon_1 (\phi_n) \leq \epsilon
	\end{equation}	
	for all $P_1\in \Pc(\Xc)$ then
	\begin{equation}\label{converse}	
	\lim_{n\rightarrow \infty} E_0^{(\epsilon)}(\phi_n) \leq D_{\frac{\alpha}{1+\alpha}}  (P_1\|P_0).
	\end{equation}	

First, by the lemma 6 in \cite{Tan}, any optimal test can be converted to a test based on samples type, and the error probabilities of such a test only change by a constant factor. Hence, the converted type-based test is asymptotically optimal, and we can limit the classifiers to type-based ones when studying the error exponent. Next, using the similar idea to \cite{Tan} we show that in order to have a type-\RNum{2} error probability bounded away form one for every $P_1$, the classifier necessarily needs to decide in favor of the second hypothesis when $\| \Tx-\TX\| \leq \delta$, for $\delta>0$, since under $P_1$ the types of the test and training sequence will converge to $P_1$. Also. we only look at deterministic test as it can be shown that randomizing the test cannot increase the error exponents.  We have the following lemma.

\begin{lemma}\label{lem:concentration}
	Let $\epsilon \in (0,1)$ and $\X^k, \xv^n$ be two independent  i.i.d sequences, generated by distribution $P$. Then
	\begin{equation}
	P \Big ( \max_{a \in \mathcal{X}} \Big |\TX(a)- \Tx(a)\Big| \geq \epsilon \Big ) \leq 2 |\mathcal{X}|e^{-n \frac{\epsilon^2}{2}} +2 |\mathcal{X}|e^{-\alpha n \frac{\epsilon^2}{2}} .
	\end{equation}
\end{lemma}

\begin{proof}
	
	By the triangle inequality, union bound, and the Hoeffding's inequality \cite{gabor}, we have
	
	\begin{align}
	P\Big (  \max_{a \in \mathcal{X}} \big |\TX(a)- \Tx(a) \big | \geq \epsilon \Big ) &=P\big( \max_{a \in \mathcal{X}}\big |\TX(a)-P(a) + P(a)- \Tx(a) \big| \geq \epsilon \Big ) \\
	&\leq P\Big (  \max_{a \in \mathcal{X}}  \big|\TX(a)-P(a)\big| + \big| \Tx(a)-P(a)\big|  \geq \epsilon \Big ) \\
	&\leq P \Big(  \max_{a \in \mathcal{X}}  \big|\TX(a)-P(a)\big|  + \max_{a \in \mathcal{X}} \big | \Tx(a)-P(a)\big|  \geq \epsilon \Big ) \\
	&\leq P \Big(  \max_{a \in \mathcal{X}}  \big|\TX(a)-P(a)\big| \geq \frac{\epsilon}{2} \cup \max_{a \in \mathcal{X}} \big | \Tx(a)-P(a)\big|  \geq \frac{\epsilon}{2} \Big)\\
	&\leq P \Big(  \max_{a \in \mathcal{X}} \big |\TX(a)-P(a)\big| \geq \frac{\epsilon}{2} \Big ) + P\Big (\max_{a \in \mathcal{X}} \big | \Tx(a)-P(a)\big|  \geq \frac{\epsilon}{2} \Big)\\
	& \leq  \sum_ {a \in \mathcal{X}} P \Big( \big |\TX(a)-P(a)\big | \geq \frac{\epsilon}{2} \Big )+ \sum_ {a \in \mathcal{X}} P\Big (\big | \Tx(a)-P(a)\big|  \geq \frac{\epsilon}{2} \Big)\\
	&  \leq 2 |\mathcal{X}|e^{-2n (\frac{\epsilon}{2})^2} +2 |\mathcal{X}|e^{-2k (\frac{\epsilon}{2})^2}  \label{hoef},
	\end{align}	
	which concludes the lemma.
\end{proof}

\begin{lemma}\label{lem:theconv}
	
	Let $(Q,Q_1) \in \Pc_{n}(\Xc)\times \Pc_k(\Xc) $ satisfy
	\begin{equation}\label{eq:lemmacondstein}
	\max_{a \in \mathcal{X}} \Big \{ \big |Q(a)-Q_1(a) \big | \Big \} \leq  \sqrt{\frac{2}{(\alpha\wedge 1)n} \log \frac{4|\Xc|}{1-\epsilon}}.
	\end{equation} 
	Then for any type based test $\phi_n(\Tx,\TX)$  such that for all distributions $\tilde{P}_1 \in \Pc(\Xc)$,
	\begin{equation}
	\epsilon_1\big(\phi_n(\Tx,\TX)\big) < \epsilon, \quad\epsilon \in(0,1), 
	\end{equation}
	we have $\phi_n(Q,Q_1)=1$. 
\end{lemma}

\begin{proof}
	We prove this lemma by contradiction. Assume there exists a type dependent test $\phi_n$ such that $\epsilon_1\big(\phi_n(\Tx,\TX)\big) < \epsilon $, and $\phi_n(Q,Q_1)=0$ for types $Q,Q_1$ satisfuing \eqref{eq:lemmacondstein}.  We have also used the notation $a \wedge b = \min\{a,b\}$. Then the type-\RNum{2} probability of error for such test is
	\begin{align}
	\epsilon_1(\phi_n)&=P_1 ( \phi_n(\Tx,\TX)=0)\\
	&\geq  P_1( \phi_n(\Tx,\TX)=0 , \Tx=Q, \TX=Q_1).
	\end{align}  
	 Conditioning on training and test sequences, we have
	\begin{align}
	\epsilon_1\big(\phi_n)&\geq  P_1 ( \phi_n(\Tx,\TX)=0 | \Tx=Q, \TX=Q_1\big) {P}_1 ( \Tx=Q, \TX=Q_1\big)\\ 
	 &={P}_1 ( \Tx=Q, \TX=Q_1\big)\\
	&\geq \Big(1-\frac{(1-\epsilon)}{2}-\frac{(1-\epsilon)}{2} \Big )\\
	&\geq \epsilon
	\end{align}  
	where in the last step we used the Lemma \ref{lem:concentration}.  Therefore, for any probability distribution $P_1$ the  type-\RNum{2} error probability exceeds $\epsilon$ which contradicts the assumption. Hence any type based 	classifier  that  achieves $\epsilon_1(\phi_n)< \epsilon$ for all distributions of $\tilde{P}_1 \in \Pc(\Xc)$ , should satisfy $\phi(\Tx,\TX)=1$ for sufficiently close types $\Tx,\TX$.  
\end{proof}

	Finally by the method of types, we can lower bound the type-\RNum{1} probability of error as 
	\begin{align}
	\epsilon_0&=P_0 (\phi_n(Q,Q_1)=1)\\
	& \geq  (n+1)^{-|\mathcal{X}|} (k+1)^{-|\mathcal{X}|}  \sum_{\phi_n(Q,Q_1)=1 } e^{-n\big ( D(Q\|P_0) + \alpha D(Q_1\|P_1) \big ) }\\
	& \geq (n+1)^{-|\mathcal{X}|} (k+1)^{-|\mathcal{X}|} e^{-n\big ( \min_{\phi_n(Q,Q_1)=1 }  D(Q\|P_0) + \alpha D(Q_1\|P_1) \big ) }.
	\end{align}
	Therefore for any type based test the type-\RNum{1} error exponent is upper bounded by
	\begin{equation}\label{eq:convexp}
	E_0(\phi) \leq \liminf_{n\rightarrow \infty} \min_{\phi_n(Q,Q_1)=1 } D(Q\|P_0) + \alpha D(Q_1\|P_1). 
	\end{equation}	
	Now by Lemma \ref{lem:theconv} for any test with type-\RNum{2} error probability bounded away from one, we have	
	\begin{align}
	E_0(\phi)& \leq \liminf_{n\rightarrow \infty} \min_{\max_{a \in \mathcal{X}} |Q(a)-Q_1(a)|  \leq \sqrt{\frac{2}{(\alpha\wedge 1)n} \log \frac{4|\Xc|}{1-\epsilon}} } D(Q\|P_0) + \alpha D(Q_1\|P_1)\\
	& = \lim_{n\rightarrow \infty} \min_{Q \in \Pc (\Xc) } D(Q\|P_0) + \alpha D(Q_1\|P_1) + o(1) \\
	&= \min_{Q \in \Pc (\Xc)}   D({Q}\|P_0) + \alpha D({Q}\|P_1)\\
	&= D_{\frac{\alpha}{1+\alpha}}  (P_1\|P_0)
	\end{align}	
	where the last step follows from e.g. \cite{Harremos}. This concludes the proof. Observe that the proof can be easily generalized for the case where both probability distributions are unknown and only training samples from both are given, i.e., the largest type-\RNum{1} error exponent achievable when the type-\RNum{2} error probability is bounded away from one is 	also $D_{\frac{\alpha}{1+\alpha}}  (P_1\|P_0)$, and Gutman's test can achieve it.

\section{Proof of Theorem \ref{thm:seq}}
\label{sec:proofth51}

We first find the error probabilities as a function of thresholds, and then we find the average stopping time under each hypothesis. The type-\RNum{1} error probability can be upper bounded by
	\begin{align}
\epsilon_0 \leq  \sum_{t=n}^{\infty} \PP \Big [ t \big(	D(\Tx \| {P}_0)-D(\Tx\| \TX') \big) \geq  \got \Big ].
\end{align} 
By the method of types \cite{Cover}, we have
	\begin{align}\label{eq:uppersensamp}
\epsilon_0 &\leq \sum_{t=n}^{\infty} \sum_{(Q,Q_1) \in	{\mathcal{Q}_{01}(t)} \cap \Pc_t(\Xc) \times \Pc_{\alpha t}(\Xc) } e^{ -t \big(D(Q\| P_0)+\alpha D(Q_1\|P_1)\big) }\\
&\leq \sum_{t=n}^{\infty} (\alpha t +1)^{|\mathcal{X}|}  (t+1)^{|\mathcal{X}|}  e^{- \Etzo},
\end{align} 
where 
\begin{equation}\label{eq:Eadverse1}
\Etzo=\min_{(Q,Q_1)\in \Qc_{01}(t)} t \big(D(Q\| P_0)+\alpha D(Q_1\|P_1)\big ),
\end{equation}	
\begin{equation}\label{eq:set}
{\mathcal{Q}_{01}}(t)=\Bigg \{(Q,Q_1):   D(Q\| {P}_0)-D(Q\| Q_1')   \geq \frac{\got}{t}, Q_1'=(1-\delta_n)Q_1+\delta_n U  \Bigg \}.
\end{equation} 
Similarly to the proof of the Theorem \ref{thm:E1}, we can expand all the exponents defined in this proof around $Q_1$ and by choosing $\delta_n=o(n^{-1})$ the error term of the expansion vanishes as $n$ goes to infinity and we can substitutde $Q_1'$ with $Q_1$ for $Q_1\in \Pc_{\delta_n}(\Xc)$. This is also true for all the exponent function we define in the rest of the proof.   For every fixed $Q_1$ we can use the dual form of the optimization \eqref{eq:Eadverse1} over $Q$ to get \cite{BoroumandTran} 
\begin{align}\label{eq:dualE0}
\Etzo=\min_{Q_1\in \Pc_{\delta_n}(\Xc)} \Big ( \max_{\lambda \geq 0}    \got \lambda -t \log   \sum_{x\in \Xc} P_0^{1-\lambda}(x) Q_1^{\lambda}(x)  \Big )+\alpha t D(Q_1\|P_1)\Big ).
\end{align}	
Substituting $\got=nD(Q_1\|P_0)+( 4|\Xc|+4)\log (t+1)$ and setting $\lambda=1$  to get a lower bound, we obtain
\begin{align}\label{eq:Eadverse}
\Etzo \geq( 4 |\Xc|+4)\log (t+1)+ n \min_{Q_1\in \Pc_{\delta_n}(\Xc)}  D(Q_1\|P_0) +\frac{\alpha t}{n} D(Q_1\|P_1),
\end{align}
where for large enough $n$
\begin{align}
\Etzo \geq ( 4|\Xc|+4)\log (t+1)+ nD_{\beta(t)}(P_1\|P_0), \label{eq:lowerE0}
\end{align}	
and
\begin{equation}
\beta(t)=\frac{\frac{\alpha t}{n}}{1+\frac{\alpha t}{n}}.
\end{equation} 
Furthermore as $\beta(t)$ is  strictly increasing function in $t$ and $D_{\beta(t)}(P_1\|P_0)$ is a non-decreasing function in $\beta$ \cite{Harremos}, then for all $ t\geq n$ we have
\begin{equation}\label{eq:lowerRenyi}
 D_{\frac{\alpha}{1+\alpha}}(P_1\|P_0)\leq D_{\beta(t)}(P_1\|P_0).
\end{equation}
Therefore, by \eqref{eq:lowerE0}, \eqref{eq:lowerRenyi} 
\begin{align}
\epsilon_0 &\leq \sum_{t=n}^{\infty} (\alpha t +1)^{|\mathcal{X}|}   (t+1)^{-3|\mathcal{X}|-4} e^{-n D_{\frac{\alpha}{1+\alpha}}(P_1\|P_0)}\\
&\leq  c e^{- nD_{\frac{\alpha}{1+\alpha}}(P_1\|P_0) },
\end{align}
where $c$ is a positive constant. 

Next, we find a lower bound to the type-\RNum{2} error exponent. Upper bounding the type-\RNum{2} error probability, we have
	\begin{align}
\epsilon_1 \leq  \sum_{t=n}^{\infty} \PP \Big [ & t(	D(\Tx\| \TX')-D(\Tx \| {P}_0)) \geq  \gzt \Big ].
\end{align} 
By the method of types, we have
\begin{align}
\epsilon_1 &\leq \sum_{t=n}^{\infty} \sum_{(Q,Q_1) \in	{\mathcal{Q}_{10}} (t) \cap \Pc_t(\Xc) \times \Pc_{\alpha t }(\Xc) } e^{ -t \big(D(Q\| P_1)+\alpha D(Q_1\|P_1)\big) }\\
&\leq \sum_{t=n}^{\infty}  (\alpha t+1)^{|\mathcal{X}|}  (t+1)^{|\mathcal{X}|} e^{- E_{10}(t)}, \label{eq:uppersensamp}
\end{align} 
where 
\begin{equation}\label{eq:Eadverse}
E_{10}(t)=\min_{(Q,Q_1) \in \Qc_{10}(t)} t \big(D(Q\| P_1)+\alpha D(Q_1\|P_1)\big ),
\end{equation}	
\begin{equation}\label{eq:set}
{\mathcal{Q}_{10}}(\gamma)=\Bigg \{(Q,Q_1):   D(Q\| Q_1')-D(Q\| {P}_0)   \geq \frac{\gzt}{t} , Q_1'=(1-\delta_n)Q_1+\delta_n U  \Bigg \}.
\end{equation} 
We show that there exists a finite $\alpha^*_{\rm seq}$ such that for every $\alpha\geq \alpha^*_{\rm seq}$, the achievable type-\RNum{2} error exponent is lower bounded by $nD(P_0\|P_1)$. Similarly to the fixed sample sized case, for every $Q_1 \in \Pc_{\delta_n}(\Xc)$, let 
\begin{equation}\label{eq:E10Q}
E_{10}(Q_1,t) =\min_{\substack{D(Q\|Q_1)-D(Q\|P_0)    \geq  \frac{\gzt}{t}  \\ Q \in \Pc(\Xc) }  } tD(Q\|P_1) ,
\end{equation}
which is the error exponent when the type of the training sequence  is $Q_1$. We also define
\begin{equation} 
\Etro=\min_{\substack{D(Q_1\|P_1)\leq r \\  Q_1 \in \Pc_{\delta_n}(\Xc) } } \Etqo.
\end{equation}
Next, we will Taylor expand $\Etro $. It is sufficient to show that there exists a finite $\alpha$ such that 
  \begin{equation}\label{eq:condseq}
\inf_{\substack{n\leq t \\ t \in \NN}} \inf_{0\leq r \leq \frac{r_c}{2}}  \Etro+\alpha tr \geq nD(P_1\|P_0)+ ( 2 |\Xc|+2)\log (t+1) .
  \end{equation}
 Equivalently,  \eqref{eq:condseq} can be written as the following condition 
\begin{align}\label{eq:alseq}
\Etro+\alpha tr  \geq nD(P_1\|P_0)+ ( 2 |\Xc|+2)\log (t+1) ~~ \forall r: ~ 0\leq r\leq \frac{r_c}{2}, n\leq t, t \in \NN.
\end{align}
Using a Taylor series expansion of $\Etro$ around $r=0$ we have \cite{BoroumandTran}
\begin{align}
\Etro \geq E_{10}(r=0,t)+ z_1(t) \sqrt{r} \mathbbm{1} \{t\geq n+1\}+   h r \mathbbm{1} \{t=n\}
\end{align}
where
\begin{align}
z_1(t)&=\inf_{0\leq r\leq \frac{r_c}{2}} \frac{ \partial \Etro  }{\partial \sqrt{r}} \\
&=\inf_{D(Q_1\|P_1) \leq \frac{r_c}{2} } -\sqrt{ \text{Var}_{P_1} \Bigg (\lambda^*(t) \frac{Q_{\lambda^*(t)}-\frac{n}{t}P_0}{Q_1} \Bigg )},\label{eq:firstterm}\\
h&=\frac{1}{2}\inf_{0\leq r \leq \frac{r_c}{2}} \frac{ \partial^2 E_{10}(r,t=n)  }{\partial \sqrt{r} ^2}. 
\end{align} 
where $Q_{\lambda^*(t)}, \lambda^*(t)$ are the minimizing distribution and the Lagrange multiplier in \eqref{eq:E10Q}. We have used the fact that for $t=n$ the optimization problem $E_{10}(r,t=n)$ is the same as the fixed sample sized classifier and hence we can use the result of the Theorem \ref{thm:upper} to lower bound the exponent by $h r \mathbbm{1} \{t=n\}$ for some finite $h$. Also for every $n<t$, the Taylor expansion of $E_{10}(r,t)$ has a nonzero  first order term, hence the expansion includes $\sqrt{r}$. Writing $E_{10}(r=0,t)$ in the dual form, we have the lower bound
\begin{equation}
E_{10}(r=0,t) \geq \max_{\frac{1}{2} \leq \lambda  }  \lambda \gamma -\log   \sum_{x\in \Xc} P_0^{\lambda}(x) P_1^{1-\lambda}(x),
\end{equation}
	where $\gamma=\frac{\gzt}{t}\Big |_{Q_1=P_1}$. Then $E_{10}(r=0,t)$ is convex as it is the supremum of linear functions in $\gamma$ \cite{Boyd}. Hence, we can lower bound $E_{10}(r=0,t) $ by expanding it around $\gamma=\frac{nD(P_0\|P_1)}{t}$ using the envelope theorem we get \cite{Segal}
	\begin{equation}
	E_{10}(r=0,t) \geq \tilde{E}_{10}(r=0,t) + ( 2 |\Xc|+2)\log (t+1)
	\end{equation}
	where
	\begin{equation}\label{eq:seqE1lower}
	\tilde{E}_{10}(r=0,t) =t\max_{\frac{1}{2} \leq \lambda  }  \lambda \frac{nD(P_0\|P_1)}{t} -\log   \sum_{x\in \Xc} P_0^{\lambda}(x) P_1^{1-\lambda}(x).
	\end{equation}
	Further expanding $\tilde{E}_{10}(r=0,t)$  around $t=n$ we have
\begin{align}\label{eq:E1apprx}
\Etro  &\geq nD(P_1\|P_0)+ ( 2 |\Xc|+2)\log (t+1) \nonumber \\
&+m(t-n) \mathbbm{1}\{ t > 2n\} +m(t-n)^2 \mathbbm{1}\{ t \leq 2n\}+  z_1(t) \sqrt{r} \mathbbm{1} \{t\geq n+1\}+   h r \mathbbm{1} \{t=n\}
\end{align}
where we have used the fact that $\frac{ \partial \tilde{E}_{10}(r=0,t) }{\partial t} \Big|_{t=n}=0$, and
\begin{equation}
m=\min \bigg \{  \inf_{2n< t} \frac{ \partial\tilde{E}_{10}(r=0,t) }{\partial t }, \frac{1}{2} \inf_{n\leq t \leq 2n} \frac{ \partial^2\tilde{E}_{10}(r=0,t) }{\partial t ^2} \bigg \}. 
\end{equation}
We have expanded $\tilde{E}_{10}(r=0,t)$ as above since it behaves linearly as $t$ goes to infinity while quadratically for $t$ close to $n$. Using Proposition \ref{prop:sensi}, we have
\begin{align}
\frac{ \partial \tilde{E}_{10}(r=0,t) }{\partial t }&= -\log   \sum_{x\in \Xc} P_0^{\tilde\lambda^*(t)}(x) P_1^{1-\tilde\lambda^*(t)}(x),\\
\frac{ \partial^2 \tilde{E}_{10}(r=0,t) }{\partial t ^2}&= \frac{ \Big (\sum Q^*(t) \log \frac{P_0}{P_1} \Big)^2 } {t\text{Var}_{P_1}\Big ( \log \frac{P_0}{P_1} \Big) },
\end{align}
which are finite and strictly positive for any $2n< t$, $n\leq t \leq 2n$, respectivly, since $Q^*(t)$ is a probability distribution and $\frac{1}{2} \leq \tilde \lambda^*(t)< 1$ is the optimizer in  \eqref{eq:seqE1lower}. From condition  \eqref{eq:alseq} by substituting the approximation \eqref{eq:E1apprx},  we need  that
\begin{align}\label{eq:alcons}
m\Big((t-n) \mathbbm{1}\{ t > 2n\} +(t-n)^2 \mathbbm{1}\{ t \leq 2n\}  \Big) +  z_1(t) \sqrt{r} \mathbbm{1} \{t\geq n+1\}+   h r \mathbbm{1} \{t=n\} +   \alpha t r \geq 0
\end{align}
for every $0\leq r, \ n\leq t, t\in \NN$. For $\alpha > \frac{|h|}{n}$, letting $r^*=\frac{z_1(t) \mathbbm{1} \{t\geq n+1\}}{2( \alpha t+h\mathbbm{1} \{t=n\}  ) }$ to minimize the LHS of \eqref{eq:alcons} over $r$ we get the condition
 \begin{align}
m\Big((t-n) \mathbbm{1}\{ t > 2n\} +(t-n)^2\mathbbm{1}\{ t \leq 2n\}\Big) -\frac{\rho(t)}{{4(\alpha t+h\mathbbm{1} \{t=n\} )}} \geq 0,
\end{align}
where
\begin{equation}
\rho(t)= {z^2_1(t) \mathbbm{1} \{t\geq n+1\}}
\end{equation}
and we have used that $0\leq z_1(t) < \infty$ by \eqref{eq:firstterm} since for every $Q_1$ satisfying $D(Q_1\|P_1)\leq \frac{r_c}{2}$, $\lambda^*(t)$ is finite and $Q_{\lambda^*(t)}$ is a probability distribution, hence the variance in \eqref{eq:firstterm} is finite. Moreover, $\rho(t)$ equals to zero at $t=n$ and it is finite for every $t\in \NN$ with finite limit as $t\rightarrow \infty$. Hence there exists a finite $c_1$ such that $\rho(t)\leq c_1\Big((t-n) \mathbbm{1}\{ t > 2n\} +(t-n)^2\mathbbm{1}\{ t \leq 2n\}\Big)$ for $t\in \NN$. Then, by further relaxing  condition \eqref{eq:alcons}, we need 
 \begin{align}
 \Big((t-n) \mathbbm{1}\{ t > 2n\} +(t-n)^2\mathbbm{1}\{ t \leq 2n\}\Big)\Big (m- \frac{c_1}{4\alpha t} \Big ) \geq   0,
\end{align}
where we have dropped the $h\mathbbm{1} \{t=n\}$ term as it is only nonzero for $t=n$ which sets $(t-n)$ and $(t-n)^2$ to zero. Therefore, if $\alpha >\frac{c_1}{4m}$, the sufficient condition is satisfied and 
	\begin{equation}\label{eq:lowerE10} 
	E_{10}(t) \geq nD(P_0\|P_1) + ( 2 |\Xc|+2)\log (t+1)
	\end{equation}
	for all $n\leq t, t \in \NN, 0 \leq r \leq \frac{r_c}{2} $. Therefore, for $\alpha>\max \Big \{ \frac{c_1}{4m}, |h|, \frac{2D(P_1\|P_0)}{r_c}+4|\Xc|+4 \Big \}$, by substituting \eqref{eq:lowerE10} in \eqref{eq:uppersensamp}, we get
\begin{align}\label{eq:experror1}
\epsilon_1 \leq  c e^{- nD(P_0\|P_1) },
\end{align}
where $c$ is a positive constant.

Next, we find the average stopping times of the proposed sequential classifier. We first show the convergence of $\tau$ in probability under each hypothesis, and by proving its uniform integrability, we can conclude its convergence in the $L^1$ norm.  The following lemma states that for every $n$ the classifier stops with the probability of one.

\begin{lemma}\label{lem:finiteness}.
Let $\tau_0, \tau_1$ be the the smallest time that the sequential classifier crosses threshold $\gzt$ or $\got$ respectively, i.e.,  
	\begin{align}\label{eq:tau1}
	\tau_{0}=\inf \{t\geq n: {S}_t \geq \gzt\}, \quad \quad    \tau_{1}=\inf \{t\geq n: {S}_t \leq -\got\}.
	\end{align}   
	Then for $i \in \{0,1\}$, $t\geq n$,
	\begin{equation}\label{eq:finiteupper}
	\PP_i[\tau_i > t] \leq      c_i(t+1)^{d_i |\mathcal{X}|}  e^{\xi_i n} e^{-t E_i},
	\end{equation}
	where $E_i, \xi_i,c_i,d_i >0$ and finite. 
\end{lemma}	
\begin{proof}
	By the method of types the probability of passing the threshold $\gzt$ under the first hypothesis at a time after $t\geq n$ can be upper bounded by
	\begin{align}
	\PP_0[\tau_0 > t] &\leq \PP_0 \big[ {S}_{t}  \leq \gzt  \big ]\\
	&= \PP_0 \bigg[ D(\Tx\|\TX')-D(\Tx\|P_0) \leq \frac{\gzt}{t}    \bigg  ] \\
	 &\leq  \sum_{(Q,Q_1) \in	{\mathcal{Q}_{00} \cap \Pc_t(\Xc)}  \times \Pc_{\alpha t}(\Xc) } e^{ -t \big(D(Q\| P_0)+\alpha D(Q_1\|P_1)\big) }\\
	&\leq (\alpha t +1)^{|\Xc|}(t+1)^{|\mathcal{X}|} e^{- \Etzz}, \label{eq:expt}
	\end{align}
	where 
	\begin{equation}
	\Etzz=\min_{(Q,Q_1)\in \Qc_{00}(t)} t \big(D(Q\| P_0)+\alpha D(Q_1\|P_1)\big ),
	\end{equation}	
	\begin{equation}\label{eq:set}
	{\mathcal{Q}_{00}}(t)=\bigg \{(Q,Q_1):  D(Q\| {P}_0) - D(Q\| Q_1)  \geq - \frac{\gzt}{t} \bigg \}.
	\end{equation} 
	For every fix $Q_1$ we can use the dual form of the optimization over $Q$ to get \cite{BoroumandTran} 
	\begin{align}
	\Etzz &=\min_{Q_1\in \Pc(\Xc)} \Big ( t\max_{ 0 \leq \lambda }    -\frac{\gzt}{t} \lambda - \log   \sum_{x\in \Xc} P_0^{1-\lambda}(x) Q_1^{\lambda}(x)  \Big )+\alpha t D(Q_1\|P_1)\Big )\\
	&\geq \min_{Q_1\in \Pc(\Xc)} \Big ( t\max_{ 0 \leq \lambda \leq 1}    -\frac{\gzt}{t} \lambda - \log   \sum_{x\in \Xc} P_0^{1-\lambda}(x) Q_1^{\lambda}(x)  \Big )+\alpha t D(Q_1\|P_1)\Big ) \label{eq:dualE00}
	\end{align}		
	Let
	\begin{align}
	E(\gamma)=\max_{0 \leq \lambda \leq 1 }  \lambda \gamma -\log   \sum_{x\in \Xc} P_0^{1-\lambda}(x) Q_1^{\lambda}(x) .
	\end{align}	
	Then $E(\gamma)$ is convex as it is the supremum of linear functions in $\gamma$ \cite{Boyd}. Therefore letting the $E_0$ to be the nonzero minimum of the optimization, we have
	\begin{equation}
	E\Big(\frac{\gamma}{t}\Big) \geq E_0 +\frac{\partial E(\gamma)}{ \partial \gamma}\Bigg |_{\gamma=0}  \frac{\gamma}{t}.  
	\end{equation}
	Applying the expansion to \eqref{eq:dualE00} we get
	\begin{align}
	\Etzz&  \geq t  \min_{Q_1\in \Pc(\Xc)}     \Big ( C(P_0\|Q_1)- \frac{n}{t} D(P_0\|Q_1)- (4 |\Xc|+4) \frac{\log(t+1)}{t} \Big)^{+} +\alpha  D(Q_1\|P_1)
	\end{align}
	where  we have used $\frac{\partial E(\gamma)}{ \partial \gamma} = \lambda^*$ with $\lambda^*\leq 1$ to be the Lagrange multiplier solving \eqref{eq:dualE00}, and lower bounding $\Etzz$ by setting $\lambda^*=1$. Moreover, $C(P_0\|Q_1)= \max_{0 \leq \lambda \leq 1} - \log   \sum_{x\in \Xc} P_0^{1-\lambda}(x) Q_1^{\lambda}(x) $ is the Chernoff information. Further lower bounding $\Etzz$ we get 
	\begin{align}
	\Etzz&  \geq t \min_{Q_1\in \Pc(\Xc)}    \Big ( C(P_0\|Q_1)- \frac{n}{t} D(P_0\|Q_1) \Big)^{+} \mathbbm{1} \{t \geq \zeta n  \}  - (4 |\Xc|+4) \frac{\log(t+1)}{t} +\alpha  D(Q_1\|P_1)\\
	& \geq - (4 |\Xc|+4) \log(t+1) + t \min_{Q_1\in \Pc(\Xc)} \Big ( C(P_0\|Q_1)- \frac{1}{\zeta} D(P_0\|Q_1) \Big)^{+} \mathbbm{1} \{t \geq \zeta n  \} +\alpha  D(Q_1\|P_1) \label{eq:lowerT0}
	\end{align}
	where $(x)^+=\max \{x,0 \}$. Choosing $\zeta > \frac{D(P_0\|P_1)}{C(P_0\|P_1)}$, for every $t\geq \zeta n$ the solution to the optimization   \eqref{eq:lowerT0} is nonzero since  $\alpha D(Q_1\|P_1)$ can be zero if and only if $Q_1=P_1$ while the first term  $  \Big ( C(P_0\|Q_1)- \frac{1}{\zeta} D(P_0\|Q_1) \Big)^{+} \mathbbm{1} \{t \geq \zeta n  \} $ is nonzero for that choice of $Q_1$. Therefore, we have
	\begin{align}
	\Etzz& \geq - (4 |\Xc|+4) \log(t+1) +  E_{0}t\mathbbm{1}  \{t \geq \zeta n  \} \\
	& \geq - (4 |\Xc|+4) \log(t+1) + E_{0} t(1-\zeta \frac{n}{t})\\
	&=- (4 |\Xc|+4) \log(t+1) +  E_{0}t -  \xi_0 n \label{eq:lowerexpt}
	\end{align} 
	where $\xi_0= E_{0} \zeta, E_0> 0$. Substituting \eqref{eq:lowerexpt} in \eqref{eq:expt} gives \eqref{eq:finiteupper}. To prove the result under the hypothesis $P_1$, let
	\begin{equation}
	\Etoo=\min_{(Q,Q_1)\in \Qc_{11}(t)} t \big(D(Q\| P_1)+\alpha D(Q_1\|P_1)\big ),
	\end{equation}	
	\begin{equation}
	{\mathcal{Q}_{11}}(t)=\Big \{(Q,Q_1):  D(Q\| Q_1) - D(Q\| P_0)  \geq - \frac{\got}{t} \Big \}.
	\end{equation} 
	Similar to the steps for hypothesis $P_0$, we have
	\begin{align}
	&\Etoo+  \lambda^*(4 |\Xc|+4) \log(t+1)       \\
	&\geq  t \min_{Q_1\in \Pc(\Xc)} \Big(   \max_{0 \leq \lambda } - \log   \sum_{x\in \Xc} P_1(x) Q_1(x)^{-\lambda} P_0^{\lambda}  \Big ) -\frac{n}{t}  \lambda^*  D(Q_1\|P_0) +\alpha  D(Q_1\|P_1) \label{eq:lowerE1t} \\
	& \geq t\min_{Q,Q_1: D(Q\|Q_1) \geq D(Q\|P_0)} D(Q\|P_1)+ \frac{\alpha}{2} D(Q_1\|P_1) + \min_{Q_1} \frac{\alpha}{2}tD(Q_1\|P_1) -n \lambda^* D(Q_1\|P_0) \\
	& \geq t\min_{Q,Q_1: D(Q\|Q_1) \geq D(Q\|P_0)} D(Q\|P_1)+ \frac{\alpha}{2} D(Q_1\|P_1) + n\min_{Q_1} \frac{\alpha}{2}D(Q_1\|P_1) -   \lambda^* D(Q_1\|P_0)\\
	&=E_{1}t -  \xi_1 n
	\end{align}
	where in the last step we used the fact that both optimizations are finite and $E_1 >0, \xi_1 >0$, and where $\lambda^*$ is the Lagrange multiplier solving the maximization in \eqref{eq:lowerE1t}. This concludes the proof.  
	\end{proof}
\begin{lemma}\label{lem:vanishing}
For $i\in\{0,1\}$
\begin{equation}
|\gamma_{i,n}(t+1)-\gamma_{i,n}(t)| \xrightarrow[]{a.s.} 0,
\end{equation}
as $t\rightarrow \infty$.
\end{lemma}
\begin{proof}
We first show the lemma for $i=0$. By the triangle inequality and the definition of KL-divergence, we have
\begin{align}
|\gamma_{0,n}(t+1)-\gamma_{0,n}(t)|  &\leq (4 |\Xc|+4) | \log(t+1) -\log(t)| +n \Bigg | \sum_{x\in\Xc} P_0(x)\log \frac{\TX'^{1:\alpha(t+1)}(x)}{\TX'^{1:\alpha t}(x)}\Bigg |\\
&\leq \frac{c}{t}+ n\Bigg | \sum_{x\in\Xc} P_0(x)\log \frac{\frac{ t}{t+1}\TX'^{1:\alpha t}(x) +\frac{1 }{t+1}\TX^{ \alpha t+1: \alpha(t+1)}(x)  }{\TX'^{1:\alpha t}(x)}\Bigg |\\
&\leq \frac{c}{t}+ n\Bigg | \sum_{x\in\Xc} P_0(x)\log \frac{ t}{t+1} + \frac{1 }{t+1}  \frac{\TX^{\alpha t+1: \alpha (t+1)} (x) }{\TX'^{1:\alpha t}(x)}\Bigg | \rightarrow 0
\end{align}
where $\TX^{i:j}$ is the type of the sequence $(X_i,\ldots,X_j)$, and in the last step, the logarithm will be either zero or goes to zero as $t$ goes to infinity. Next, for $\gamma_{1,n}(t)$, by the triangle inequality and the $L_1$ bound on entropy \cite{Cover} we have
\begin{align}
|\gamma_{1,n}(t+1)-\gamma_{1,n}(t)|  \leq &(4 |\Xc|+4)| \log(t+1) -\log(t)| +n \Big |H(\TX'^{ 1: \alpha (t+1)})- H(\TX'^{1: \alpha t}) \Big | \nonumber \\
&  + n \bigg|\sum_{x\in\Xc} (\TX'^{ 1: \alpha (t+1)} (x) - \TX'^{1: \alpha t}(x) ) \log P_0(x) \bigg |  \\
\leq& \frac{c}{t}-n \|\TX'^{ 1: \alpha (t+1)}- \TX'^{1: \alpha t} \|_1 \log \frac{\|\TX'^{ 1: \alpha (t+1)}- \TX'^{1: \alpha t} \|_1 }{|\Xc|}  + c'\|\TX'^{ 1: \alpha (t+1)}- \TX'^{1: \alpha t} \|_1 \\
\leq &\frac{c}{t}+ \frac{n }{t+1} \log  \frac{1 }{|\Xc| (t+1)}+\frac{c'}{t+1} \rightarrow 0,
\end{align}
where in the last step we have used
\begin{align}
\|\TX'^{ 1: \alpha (t+1)}- \TX'^{1: \alpha t} \|_1 \leq  \frac{1}{t+1} \|\TX'^{ 1: \alpha t}- \TX'^{ \alpha t : \alpha (t+1)} \|_1 \leq \frac{1}{t+1},
\end{align}
as $\TX'$ is a type of a training sequence, and $c,c'$ are positive constants.
\end{proof}

Now by the finiteness of $\tau_0$ for every $n$, and the definition of $\tau_0$, there exists a finite $\tau_0$ with probability one such that
\begin{align}\label{eq:convergp}
{S}_{\tau_0-1} < \gamma_{0,n}(\tau_0-1), ~~~ \gamma_{0,n}(\tau_0)  \leq {S}_{\tau_0}~~~   \text{w.p}.1.
\end{align}
Furthermore, for every $\tau_0$ we have
\begin{align}
\frac{{S}_{\tau_0}}{\tau_0}= D(\Tx^{\tau_0} \| \TX^{\tau_0})- D(\Tx^{\tau_0}\|P_1)+o(1).
\end{align}
Also, since by design $\tau_0 \geq n$, using the WLLN, and the continuous mapping theorem, as $n \rightarrow \infty$,  we get  
\begin{align}\label{eq:convergp1}
\frac{{S}_{\tau_0}}{\tau_0} \xrightarrow[]{p} D(P_0\|{P}_1), \quad \frac{{S}_{\tau_0-1}}{\tau_0-1} \xrightarrow[]{p} D(P_0\|{P}_1).
\end{align}
Therefore, by Lemma \ref{lem:finiteness}, Lemma \ref{lem:vanishing}, and \eqref{eq:convergp}, \eqref{eq:convergp1} we can conclude that 
\begin{equation}\label{eq:convP1}
\frac{\gamma_{0,n}(\tau_0)}{\tau_0}  \xrightarrow[]{p}   { D(P_0\|P_1)},
\end{equation}  	
as	$n \rightarrow  \infty$. Also, we have
\begin{equation}\label{eq:conveq}
\frac{\gamma_{0,n}(\tau_0)}{\tau_0}= \frac{n}{\tau_0} D(P_0\|\TX') +\frac{\log (\tau_0+1)}{\tau_0},
\end{equation}
and by assumtion $\max_{x\in\Xc} \frac{P_0(x)}{P_1(x)} \leq c$,  $D(P_0\|\TX')$ is a consistent estimator of $D(P_0\|P_1)$ \cite{KL}, and hence
\begin{equation}\label{eq:consistency}
D(P_0\|\TX')  \xrightarrow[]{p}    D(P_0\|P_1).
\end{equation}
Finally by \eqref{eq:convP1}, \eqref{eq:conveq}, \eqref{eq:consistency}, and using continuous mapping theorem \cite{Resnick}, we have
\begin{equation}
\frac{\tau_0}{n}\xrightarrow[]{p} 1,
\end{equation}
as $n\rightarrow \infty$.   

To show the convergence in $L^1$ we only need to prove the uniform integrability of the sequence of random variables $\frac{\tau_0}{ n}$, where $\tau_0$ depends on  $n$. Equivalently, we need to show that,
\begin{equation}\label{eq:uniform}
\lim_{t \rightarrow \infty} \sup_{n\geq 1} \mathbb{E}_{P_0} \Bigg [\frac{\tau_0}{ n} \mathds{1}\Big \{\frac{\tau_0}{ n} \geq t \Big  \}     \Bigg] =0.
\end{equation}
By  \eqref{eq:finiteupper}, we can upper bound the given expectation in \eqref{eq:uniform} as  
\begin{align}
\mathbb{E}_{P_0} \Big [\frac{\tau_0 }{n} \mathds{1}\big \{{\tau_0} \geq t  n \big  \}     \Big] &= \frac{1}{n} \sum_{m=1}^{\infty} \PP_0\big [\tau_0-t n \geq m \big ]\\
& \leq  \frac{1}{n} t e^{-n (t E_0-\xi_0)}  \sum_{m=0 }^{\infty} c(m+tn+1)^{4|\mathcal{X}|}     e^{-m E_0}.
\end{align}
Hence the expectation is vanishing as $t\rightarrow\infty$ for every $n$ 
giving the uniform integrability of $\frac{\tau_0}{ n}$, and hence convergence in $L^1$  \cite{Bill}, i.e,
\begin{equation} \label{eq:convergT1}
\lim_{n \rightarrow \infty} \mathbb{E}_{P_0} \Big [ \Big|\frac{\tau_0}{ n}- 1 \Big | \Big]=0.
\end{equation}	
Finally, we prove the convergence of $\tau$. By \eqref{eq:experror1}, \eqref{eq:convP1} and the union bound, we obtain 
\begin{align}
\PP_0 \Big[\Big|\frac{\tau}{ n}- 1\Big | \geq \epsilon \Big] & \leq \PP_0 \Big[\Big|\frac{\tau}{ n}- 1\Big | \geq \epsilon, {\phi}=0 \Big ]+ \PP_0[{\phi}=1 ]\\
&=\PP_0 \Big[\Big|\frac{\tau_0}{ n}- 1 \Big | \geq \epsilon\Big] + \epsilon_0,  
\end{align}	
which tends to $0$ as $n \rightarrow \infty$, establishing the convergence of $\frac{\tau}{ n}$ in probability. Now, using that $\tau\leq \tau_0$ we have	
\begin{align}\label{eq:upper1}
\mathbb{E}_{P_0} \Big [\frac{\tau}{ n} \mathds{1}\Big \{\frac{\tau}{ n} \geq t \big  \}     \Big] \leq  \mathbb{E}_{P_0} \Big [\frac{\tau_0}{ n} \mathds{1}\Big \{\frac{\tau_0}{ n} \geq t \Big  \}     \Big].
\end{align}
Therefore, uniform integrability of $\tau_0$ gives the uniform integrability of $\tau$, and hence convergence in $L^1$ norm and also expectation of $\frac{\tau}{ n}$, which concludes the proof.
\section{Proof of Theorem \ref{thm:converse}}
\label{sec:proofth52}
For the type-\RNum{2} error exponent, the converse for sequential hypothesis testing is applicable, i.e., for every sequential test with $\mathbb{E}_{P_0}[\tau] \leq n$,  we have $E_1 \leq D(P_0\|P_1)$ \cite{Poly}. To find an upper bound to $E_0$ we use the following lemma.
\begin{lemma}\label{lem:theconvseq}

	 For any type based sequential test $\Phi^{\rm seq}$, let $\tau \in \Nc_1^\epsilon$,  where $\Nc_1^\epsilon=\{1,...,t \}$ is the typical stopping time set such that
	 \begin{equation}
	  P_1(\tau \in \Nc_1^\epsilon)\geq 1-\frac{\epsilon}{2}. 
	 \end{equation}
	Also for every $t$, let $(\xv^t, \X^{\alpha t}) \in \Bc_t^\epsilon$ where
	\begin{equation}
	\Bc_t^\epsilon= \Bigg\{  (\xv^t, \X^{\alpha t}) : \max_{a \in \mathcal{X}} \Big \{ |\Tx(a)-\TX(a)| \Big \} \leq  \sqrt{\frac{2}{(\alpha \wedge 1) t} \log \frac{8|\Xc|}{\epsilon}} \Bigg \}.
	\end{equation}
	Then for any type based test $(\phi(\Tx,\TX),\tau)$   such that for all distributions ${P}_1 \in \Pc(\Xc)$,
        \begin{equation}
	\epsilon_1\big(\phi(\xv^\tau,\X^{\alpha \tau})\big) < \epsilon, \quad\epsilon \in\Big (0,\frac{1}{2}\Big), 
	\end{equation}
	we have 
	\begin{equation}\label{eq:steinseq}
	\PP \Big [ \phi(\xv^\tau,\X^{\alpha \tau})=1,  (\xv^\tau, \X^{\alpha \tau}) \in \Bc_\tau^\epsilon, \tau \in \Nc_1^\epsilon    \Big] \geq 1-2\epsilon.
	\end{equation}
\end{lemma}
\begin{proof}
We prove the lemma by contradiction. Assume 
	\begin{equation}
	\PP \Big [ \phi(\xv^\tau,\X^{\alpha \tau})=1,  (\xv^\tau, \X^{\alpha \tau}) \in \Bc_\tau^\epsilon, \tau \in \Nc_1^\epsilon    \Big] < 1-2\epsilon.
	\end{equation}
	Then
	\begin{align}
	2\epsilon &< \PP \Bigg [ \phi(\xv^\tau,\X^{\alpha \tau})\neq 1\cup (\xv^\tau, \X^{\alpha \tau}) \notin \Bc_\tau^\epsilon \cup \tau \notin \Nc_1^\epsilon    \Big]\\
	&\leq \PP \Big [ \phi(\xv^\tau,\X^{\alpha \tau})\neq 1 \Big] + \PP \Big [ (\xv^\tau, \X^{\alpha \tau}) \notin \Bc_\tau^\epsilon  \Big]+ \PP \Big [ \tau \notin \Nc_1^\epsilon    \Big]\\
	&\leq  \epsilon_1 + \frac{\epsilon}{4} +\frac{\epsilon}{4} +\frac{\epsilon}{2}. 
	\end{align}
	where in the last step we used Lemma \ref{lem:concentration}. Hence, $\epsilon \leq \epsilon_1$ which is a contradiction and hence \eqref{eq:steinseq} holds.
\end{proof}

	By Lemma \ref{lem:theconvseq} we can conclude that if the condition \eqref{eq:steinseq} does not hold there exists a distribution $P_1$ such the that the type-\RNum{2} error probability is bounded away from zero and hence the type-\RNum{2} error exponent of such test equals to zero.  Therefore, by  \eqref{eq:steinseq} we can lower bound the $\epsilon_0$ for any test with nonzero $E_1$ as 
	\begin{align}
	\epsilon_0 =& \PP[\phi(\xv^\tau, \X^{\alpha \tau}) =1]\\
	\geq &\PP \Big [ \phi(\xv^\tau,\X^{\alpha \tau})=1, (\xv^\tau, \X^{\alpha \tau}) \in \Bc_\tau^\epsilon, \tau \in \Nc_1^\epsilon    \Big]\\
	=& \PP \Big [ \phi(\xv^\tau,\X^{\alpha \tau})=1 \Big| (\xv^\tau, \X^{\alpha \tau}) \in \Bc_\tau^\epsilon, \tau \in \Nc_1 \Big]  \PP \Big [ (\xv^\tau, \X^{\alpha \tau}) \in \Bc_\tau^\epsilon , \tau \in \Nc_1    \Big]. 
	\end{align}
	By the previous lemma, we can lower bound the first probability by $1-2\epsilon$. Also, let $\epsilon$ to be sufficiently small, such that $n \in \Nc_1^\epsilon=\{1,...,N\}$, then by the method of types
	\begin{align}
	\epsilon_0 \geq& (1-2\epsilon) \sum_{t=1}^{N} \PP \Big [ (\xv^\tau, \X^{\alpha \tau}) \in \Bc_\tau^\epsilon \Big | \tau =t    \Big]  \PP_0 [ \tau=t]\\
	 \geq &(1-2\epsilon)\sum_{t=1}^{n}  (t+1)^{-|\mathcal{X}|} (\alpha t+1)^{-|\mathcal{X}|} e^{-t \min_{(Q,Q_1) \in \Bc_t^\epsilon}  D(Q\|P_0) + \alpha D(Q_1\|P_1) } \PP_0 [ \tau=t]\\
	 \geq &(1-2\epsilon)ce^{-n \min_{(Q,Q_1) \in \Bc_t^\epsilon}  D(Q\|P_0) + \alpha D(Q_1\|P_1) }  \PP_0 [ \tau\leq n]\\
	\end{align}	
	where $c$ is a positive constant. Now by $\mathbb{E}_{P_0}[\tau]\leq n$, the optimal test should stops by the time $n$  with a positive probability, i.e., $\PP_0 [ \tau\leq n]>0$, since otherwise, $\mathbb{E}_{P_0}[\tau] > n$. Finally, by letting $\epsilon \rightarrow 0, n \rightarrow \infty$, we have
	\begin{equation}
	E_0 \leq D_{\frac{\alpha}{1+\alpha}}(P_1\|P_0).
	\end{equation}
	Hence, for any sequential test  with a finite stopping time, and  the type-\RNum{2} error probability that is bounded away from one for every distribution $P_1$, the type-\RNum{1} error exponent is bounded by $D_{\frac{\alpha}{1+\alpha}}(P_1\|P_0)$,  which concludes the proof.

%
%
%
%
%
%

\bibliographystyle{ieeebib}
\bibliographystyle{ieeetr}
\bibliography{Classification_Journal}

\end{document}